\documentclass[conference,compsoc]{IEEEtran}
\renewcommand\IEEEkeywordsname{Keywords}
\usepackage{pgfplots}
\usepackage{algpseudocode} 
\usepackage[linesnumbered]{algorithm2e}
\usepackage{amsfonts} 
\usepackage{amsmath, amssymb}
\usepackage{amsthm}
\usepackage{animate}
\usepackage[contents= {}, opacity =0.1]{background}
\usepackage{caption}
\usepackage{cite}
\usepackage{dsfont,enumitem}
\usepackage[T1]{fontenc}
\usepackage{ifthen}
\usepackage{graphicx}
\usepackage{ifthen}

\usepackage[noprefix]{nomencl}
\makenomenclature
\usepackage{etoolbox}
\usepackage{etoolbox}

\usepackage{hyperref}
\usepackage{xfrac}    
\usepackage{tikz}
\usepackage{environ}

\usetikzlibrary{fit,positioning,calc}
\usetikzlibrary{shadows,hobby}
\usetikzlibrary{fadings}
\usetikzlibrary{shapes.arrows,calc,quotes,babel}
\usetikzlibrary{graphs,graphs.standard,arrows.meta, shapes.misc, positioning,decorations.pathreplacing,calligraphy}
\usetikzlibrary{colorbrewer}
\usepgfplotslibrary{colorbrewer}
\RestyleAlgo{ruled} 
\SetKwComment{Comment}{/* }{ */}

\usepackage{xcolor}
\newif\ifcomment
\commenttrue

\newcommand{\nodeset}{\mathcal{N}}
\newcommand{\branchset}{\mathcal{B}}
\newcommand{\conflictset}{\mathcal{C}}
\newcommand{\blockset}{\mathcal{T}}
\newcommand{\ledger}{\mathcal{L}}
\newcommand{\votingset}{\mathcal{V}}
\newcommand{\LedgerDAG}{D_{\ledger}}
\newcommand{\ConflictDAG}{D_{\conflictset}}
\newcommand{\ConflictGraph}{G_{\conflictset}}

\newcommand{\TangleDAG}{D_{\blockset}}
\newcommand{\VotingDAG}{D_{\votingset}}
\newcommand{\LocalTangleDAG}[1]{D_{\blockset_{#1}}}
\let\phi=\varphi
\newcommand{\cone}[3]{\mathrm{cone}_{#2}^{(#1)}\left(#3\right)}
\newcommand{\parent}[2]{\mathrm{par}_{#1}\left(#2\right)}
\newcommand{\child}[2]{\mathrm{child}_{#1}\left(#2\right)}
\newcommand{\maximal}[2]{\max_{#1}\left(#2\right)}
\newcommand{\minimal}[2]{\min_{#1}\left(#2\right)}

\newcommand{\val}{v}
\newcommand{\genesis}{\rho}
\newcommand{\cond}{\mathrm{cond}}

\newcommand{\reference}{\mathrm{ref}}

\newcommand{\unlock}{\mathrm{unlock}}

\newcommand{\branch}{\mathrm{branch}}

\newcommand{\inputs}{\mathrm{in}}
\newcommand{\outputs}{\mathrm{out}}

\newcommand{\weight}{\mathbf{w}}

\newcommand{\mathV}{\mathcal{V}}
\newcommand{\E}{\mathbb{E}}
\renewcommand{\P}{\mathbb{P}}
\newcommand{\tips}{\mathbf{T}}
\newcommand{\algo}{\mathbf{A}}

\newcommand{\M}{\mathcal{M}}

\newcommand{\Q}{\mathcal{Q}}

\newcommand{\dFPCS}{\mathcal{D}}
\newcommand{\ddRNG}{\mathbf{d}}

\newcommand{\AW}{\mathbf{AW}}

\newcommand{\WW}{\mathbf{WW}}
\newcommand{\supporter}{\mathrm{sprt}}

\newcommand{\supporterAW}{\mathrm{sprt}^{\ledger}}

\newcommand{\tx}[1]{\hat{#1}}
\let\phi=\varphi

\newcommand{\hash}{\mathop{\mathrm{hash}}}
\newcommand{\iss}{\mathbf{issue}}

\DeclareMathOperator*{\argmax}{arg\,max}

\newtheorem{theorem}{Theorem}

\newtheorem{cor}{Corollary}
\newtheorem{lemma}{Lemma}
\newtheorem{proposition}{Proposition}
\newtheorem{remark}{Remark}
\theoremstyle{definition}
\newtheorem{definition}{Definition}
\newtheorem{example}{Example}
\newtheorem{ass}{Assumption}

\numberwithin{definition}{section}
\numberwithin{lemma}{section}
\numberwithin{theorem}{section}
\numberwithin{cor}{section}
\numberwithin{ass}{section}
\numberwithin{remark}{section}
\numberwithin{example}{section}
\numberwithin{proposition}{section}

\begin{document}
\title{Tangle 2.0\\  Leaderless Nakamoto Consensus on the Heaviest DAG}

\author{
\IEEEauthorblockN{Sebastian M\"uller\IEEEauthorrefmark{1}\IEEEauthorrefmark{2}, Andreas Penzkofer\IEEEauthorrefmark{2},
Nikita Polyanskii\IEEEauthorrefmark{2},
Jonas Theis\IEEEauthorrefmark{2}, 
William Sanders\IEEEauthorrefmark{2},
Hans Moog\IEEEauthorrefmark{2}}
\IEEEauthorblockA{\IEEEauthorrefmark{1}Aix Marseille Universit\'e, CNRS, Centrale Marseille, I2M - UMR 7373,  13453 Marseille, France}
\IEEEauthorblockA{\IEEEauthorrefmark{2}IOTA Foundation, 
Berlin, Germany}}



\maketitle

\begin{abstract}
We introduce the theoretical foundations of the Tangle 2.0, a probabilistic leaderless consensus protocol based on a directed acyclic graph (DAG) called the Tangle. The Tangle naturally succeeds the blockchain as its next evolutionary step as it offers features suited to establish more efficient and scalable distributed ledger solutions. 

Consensus is no longer found in the longest chain but on the heaviest DAG, where PoW is replaced by a stake- or reputation-based weight function.
The DAG structure and the underlying Reality-based UTXO Ledger allow parallel validation of transactions without the need for total ordering. Moreover, it enables the removal of the intermediary of miners and validators, allowing a pure two-step process that follows the \emph{propose-vote} paradigm at the node level and not at the validator level.

We propose a framework to analyse liveness and safety under different communication and adversary models.  This allows providing impossibility results in some edge cases and in the  asynchronous communication model.  We provide formal proof of the security of the protocol assuming a common random coin.
\end{abstract}

\renewcommand\IEEEkeywordsname{Keywords}

\section{Introduction}

In distributed systems, different events may happen \emph{at the same time}, but participants may perceive them in different orders. 
In contrast, distributed ledger technologies (DLTs) such as Bitcoin~\cite{nakamoto2008bitcoin} typically use a totally ordered data structure, a blockchain, to record the transactions that define the state of the ledger. This design creates a bottleneck, e.g. a miner or validator, through which each transaction must pass. The creation of blocks can also happen concurrently at different parts of the network, leading to bifurcations of the chain that must be resolved. This is typically done by the longest--chain rule ~\cite{nakamoto2008bitcoin} or  some variant of the heaviest sub-tree \cite{GHOST}.  To guarantee the security of the system, the throughput of the system is artificially suppressed so that each block propagates fully before the next block is created, and very few ``orphan blocks''  spontaneously split the chain. 
Another effect that limits scalability is that the transactions are handled in batches. The miners create these batches or blocks of transactions and the blockchain can be seen as a three-step process. In the first step, a client sends a transaction to the block producers, then some block producer proposes the block containing a batch of transactions, and in the last step, validators validate the block.

A more novel approach that addresses the asynchronous setting of the distributed system has been taken by IOTA~\cite{popov2015}. This approach eliminates the need for clustered transactions and
uses a directed acyclic graph (DAG) (as the underlying data structure) to express simultaneous events. In this model, individual transactions are added to the ledger, and each transaction refers to at least two previous transactions. This property reduces the update of the ledger to two steps: One node proposes a transaction to the ledger and waits for the other nodes to validate it.  The removal of the intermediary of miners or validators promises to solve (or at least mitigate) several problems associated with them, e.g. mining races~\cite{DEVRIES2021105901}, centralisation~\cite{NBERw29396}, miner extractable value \cite{MEV}, and negative externalities \cite{Rosenthal22}  and allows for a fee-less architecture.  
However, the parallelism involved in adding new transactions to the ledger means that consensus must be found on a ``wider'' subgraph than just the longest chain or the heaviest sub-tree. 

\subsection{Results}
Two main problems of Nakamoto's ``longest-chain rule'' are the severely limited scalability and the lack of parallelisability.  The lack of parallelisability results in the underlying communication network requiring strong assumptions about synchronicity. 
We propose a consensus protocol that works efficiently and fast in an asynchronous model and allows a high degree of parallelisation.  This is achieved by replacing the ``longest-chain rule'' with the ``heaviest-DAG rule''.
As the resulting consensus is not based on a total ordering of the transactions, it enables the transactions to be stream processed. An optimization that becomes more and more relevant in the validation of smart contract updates and optional sharding solutions.

Another disadvantage in blockchains, which is perhaps not so well known, is the need for intermediaries in the form of miners or validators. By enabling leaderless writing access to the ledger we remove this dependency and reduce the system to a dichotomy of fund owners and nodes, where nodes take additional roles akin to validators. Nodes propose new blocks, which contain transactions from fund owners, and append them to the Tangle.  Nodes utilise the append process to   validate and vote on previous blocks in a highly efficient implicit voting scheme.  

We propose a generalisation of the voting power of nodes in form of a generalised weight function. This generalisation allows for a high level of configurability of our protocol, making it   adaptable to the needs and security requirements of the system in which it should be implemented, such as permissionless or permissioned.

We introduce an asynchronous leaderless protocol that employs a weight-based voting scheme on the Tangle. In this scheme, the supporters of transactions, which are the nodes, are tracked through implicit votes. The confirmation status of transactions can be determined using threshold criteria. We provide the algorithms for the various core components. More specifically, we describe how the supporter lists are updated through the implicit voting scheme and how nodes should attach their blocks to the Tangle. We provide theorems for the convergence, as well as the liveness and safety of the system. First, given a random, unpredictable influx of blocks, Theorem~\ref{thm: security Random Blocks} gives guarantees that the system will converge eventually on a consensus state if an adversary has less than 50\% of the weight, however, no safety guarantees are given in this case. Second, we give safety and liveness guarantees by extending the protocol and incorporating the capability to synchronise the nodes at certain intervals with the help of a common coin. The security guarantees for this extended protocol are given in Theorem~\ref{thm: SOTV}. Finally, we provide an overview of simulation results that display the performance of the protocol.

\subsection{Structure of the Paper}

The document is structured as follows. 
In Section~\ref{sec:background} we give an overview of essential aspects relevant to the design of a DLT solution. In Section \ref{sec:related work} we provide an overview of other recent DAG-based protocols and highlight the differences to our proposal. 
Section~\ref{sec: lists} provides an overview of used symbols, acronyms and glossary. 
Section~\ref{sec: graph-theoretical notions}  gives an overview of some of the graph-theoretical preliminaries used in this paper.
In Section~\ref{sec: nodes sybil} we provide a basic network setting within which the proposed Sybil protection mechanism operates.
 Section~\ref{sec: tangle} describes the functionality of the Tangle data structure and how it is utilised to confirm blocks.
Section~\ref{sec: ledger} introduces an overview of the Reality-based UTXO Ledger, which forms a central component in our approach that helps with tracking the opinions of honest nodes about conflicting transactions.
In Section~\ref{sec: Voting}, we describe the voting protocol and confirmation of transactions. 
In Section~\ref{sec: comm Adversary Model} we define the communication and adversary models and address the liveness and security of the system in Sections~\ref{sec: security} and~\ref{sec: impossibility Metastability}. In particular,  we show that certain attacks that attempt to create a ``metastable'' situation, could become problematic under specific circumstances and strong assumptions about the adversary. In Section~\ref{sec: SOTV} we provide a solution to this by introducing a synchronization of nodes at larger time intervals. In Section~\ref{sec: implementation}, to showcase the performance of the protocol, we provide results from simulation studies.
Finally, we conclude the paper with Section~\ref{sec: outlook}, where we describe future research directions.

\subsection{Background}\label{sec:background}

Consensus protocols in general and even DLTs, in particular, are such a large research area that we have to refer to some review articles for a more detailed introduction, e.g. ~\cite{SOK, vademecum, Wang2020SoKDI}. Although a consensus protocol depends on many different aspects, we focus, in the remaining part of the introduction, on those that are most relevant for the design choices of our proposed protocol.

\subsubsection*{Ledger Model}
Distributed ledgers (DLs) generally arrive in two flavours of balance keeping: an account-based model, where funds are directly associated with the account of a user, such as is the case with Ethereum~\cite{buterin2013ethereum}; and an unspent transaction output (UTXO) model, where tokens are linked to a so-called output, and users own the keys to the output, as is the case with Bitcoin~\cite{nakamoto2008bitcoin} and many of its derivatives, as well as Cardano \cite{cardanoEUTXO}, Avalanche~\cite{Ava19}, and IOTA \cite{2020coordicide}. 
As an important observation in the latter case, the UTXOs form a DAG themselves. A total ordering of the transactions is unnecessary for many use cases and situations, as most of them are parallelisable. However, the append-only nature of the UTXO ledger hinders this advantage of parallelisation in the presence of conflicting transactions. In \cite{RealitiesLedger2022} we propose an augmented UTXO ledger model that optimistically updates the ledger and tracks the dependencies of the possible
conflicts. 
We construct a consensus protocol that utilises this ledger model to enable fast and parallelisable conflict resolution.

\subsubsection*{The Tangle and Partial Order}

The Tangle is the DAG that stores all transactions of the  distributed ledger (DL). Every DAG induces a partial order on the set of vertices, the collection of transactions in our setting. This property contrasts with a blockchain where a total order of transactions is established. As in systems with crash failures, atomic broadcast and consensus are equivalent problems, see~\cite{ChTuTo:96}, the partial order of the DAG induces additional ``difficulties'' in the consensus protocol. 
More precisely, there have been serious limitations concerning the security of a DAG-based DLT. In the original proposal of the Tangle,~\cite{popov2015}, the longest chain rule was replaced by the ``heaviest sub-graph'', i.e.  the sub-DAG containing the most transactions. However, it turned out that this design is vulnerable to various types of attacks and would rely too much on the Proof-of-Work necessary to issue a transaction, e.g.  \cite{penzkofer2020}. Another critical element of the design that is common to many other DAG-based proposals is that it suffers a liveness problem. Honest transactions that refer to transactions that turn out to be malicious in the future can not be added to the ledger state. The protocol we propose in this paper solves the security problems by relying on a weight function for nodes and by using the Reality-based Ledger. It also treats the problems of liveness by separating transactions from their containers, which are blocks,\footnote{Unlike many blockchain protocols, we require each block to contain precisely one transaction. However, in principle, the protocol can be adapted such that blocks contain more than one transaction.} and by applying a new block referencing scheme. In particular, this batch-less architecture enables a stream process-oriented design of the DLT.

\subsubsection*{Sybil Protection}
Sybil protection plays a crucial role in a ``permissionless environment'' where everyone can participate. By leveraging Proof-of-Work (PoW), Bitcoin’s Nakamoto consensus was the first to achieve consensus in such an open environment. As PoW leads to enormous energy waste and many negative externalities, a lot of effort has been put into proposing more sustainable alternatives. The most prominent of them is called Proof-of-Stake (PoS), where the validators' voting power is proportional to their stake (i.e.  in terms of the underlying cryptocurrency) in the system.

The Sybil protection used in this paper is based on node identities. We describe it generically as a function of a scarce resource or an abstract reputation function. This function, called \emph{weight} assigns every node identity a positive number. 
For example, this weight can correspond to an amount of staked tokens, delegated tokens, or the ``mana'' described in  \cite{2020coordicide}. We want to note that the weight does not have to be connected to the underlying token but can be replaced by any other ``weight'' serving as a good Sybil protection. In particular, our framework can also be used in a permissioned setting, where only the pre-defined validators would have a positive weight and can apply to the situation with dynamic committee selections. 

\subsubsection*{Nakamoto Consensus}
Distributed consensus allows participants to agree on a constantly growing log of transactions. It has been an important research topic in recent decades, and its importance in computer science has never been disputed. There are many ways to categorize consensus protocols. For instance, there are the classical landmark results on PAXOS and BFTs, and the newer  \textit{Nakamoto} type consensus mechanisms. 

We understand as Nakamoto consensus the rule to select the longest sub-chain, e.g. see\cite{Wang2020SoKDI}, and as a variant also the heaviest weighted sub-chain. We extend this concept to the heaviest sub-DAG. More precisely we consider, a Nakamoto  blockchain consensus to follow the \emph{propose-vote} paradigm and that it can be described as follows. The time is divided into epochs, and for each epoch, there is an ``elected'' leader. This leader batches  transactions into a new block and proposes this block. Then the other participants vote on the proposed block, e.g. by extending the chain to which the proposed block is attached. Once the number of votes reaches a certain threshold, the proposed block is considered part of the ledger. The specific definition of the various elements mentioned above may vary and lead to different variants of the Nakamoto Consensus. To some extent, the above paradigm reduces to the necessity to agree on a unique leader in each epoch. 
Once the participants have a consensus on the leader, the linearity of the blockchain implies consensus on the ledger state. However, the fact that only a leader can advance the ledger state creates an obvious bottleneck with well-known performance limitations. 
In our proposal, we remove the role of the ``leader'' entirely and allow the participants to propose their blocks and the contained transactions concurrently. Once a block is proposed, all participants can vote and participate in the consensus finding. The weight of the vote is proportional to the \emph{weight} of the node, introduced above, such that the protocol adapts to different weight distributions. 
The protocol is also classified as a non-binary consensus protocol since it can decide on several transactions simultaneously and is an ever-ongoing voting procedure forming a progressively-growing history.\footnote{Both the ledger DAG structure and the Tangle are technically transient data structures since they can be pruned in theory. 
Thus, the voting is also transient knowledge. However, for simplicity, we assume that both data structures are not pruned.}
It also relates to a probabilistic consensus in the sense that the more supporting nodes a transaction accumulated the more likely it is that this transaction is eventually confirmed and added to the ledger.

\subsubsection*{Voting}

In our non-linear architecture, each new block references at least two existing blocks. This results in a DAG structure as mentioned above. As with a blockchain, a new block not only votes on its direct references but also on its past cone. Although this is an efficient voting scheme, there is the problem of orphanage or liveness. If a block contains an invalid block in its past cone, it can no longer be voted for and, thus, the contained transaction cannot be included in the ledger. We solve this problem by introducing two different references. The first reference is to the Tangle structure and the second is to the DAG structure originating from the UTXO ledger. The last reference allows voting for transactions that were originally orphaned and also to change previously issued votes. Eventually, both types of votes accumulate in a voting weight, which we call the Approval Weight (AW). The higher this AW the higher the probability that the transaction is eventually included in the ledger. We refer to Figure \ref{fig:octopus_example} for an example of the voting mechanism.
\begin{figure}
    \centering
    \definecolor{megalightgray}{RGB}{244, 244, 244}
\definecolor{mygray}{RGB}{240, 240, 240}
\definecolor{myblue}{RGB}{102, 178, 255}
\definecolor{myred}{RGB}{255, 102, 102}
\tikzstyle{rounded_block}=[draw, rectangle, thick, minimum height=\heightBlock cm, minimum width = \widthBlock cm, text centered, rounded corners, draw=darkgray, font = \small]

\begin{tikzpicture}[use Hobby shortcut, scale = 0.95]
\def\xCoordinate{0.0}
\def\yCoordinate{0.0}
\def\yAdd{1.5}
\def\xAdd{0.75}
\def\innerSepar{-1.5}
\def\lineWidth{1.5}
\def\roundedCorners{1}
\def\opacityInternal{0.5}
\def\minHeight{20}
\def\fracyAdd{1/6}
\def\lineWidthBelow{3}
\def\minWidth{20}
\def\scaleFactorSupp{0.12}
\def\minWidthCM{\minWidth*0.0352778*3/4}
\tikzstyle{block}=[draw, rectangle, minimum height=\minHeight pt, minimum width = \minWidth pt, text centered, rounded corners=\roundedCorners pt, draw=darkgray, font=\large]
\tikzstyle{block_colored}=[draw, rectangle, minimum height= 2 pt, minimum width = \minWidth * 5 / 6 pt, rounded corners=\roundedCorners  pt, draw=darkgray]
\def\arrowStyle{-latex}

\path
  (-0.2,-2.5) coordinate (z0)
  (-1.5,-4.2) coordinate (z1)
  (0.4, -6.5) coordinate (z2)
  (1.0, -5.4) coordinate (z23)
  (1.4,-5.0) coordinate (z233)
  (1.5,-4.2) coordinate (z3)
  (0.5,-3.35) coordinate (z4)
  (0.35,-2.5) coordinate (z5);
  \draw[closed, draw = myred, fill = megalightgray] (z0) .. (z1) .. (z2) .. (z23) .. (z233) .. (z3) .. (z4) .. (z5);

\path
  (1.6,-2.5) coordinate (y0)
  (1.6,-3.5) coordinate (y01)
  (2.1,-4.0) coordinate (y1)
  (1.9, -4.9) coordinate (y12)
  (0.75, -6.5) coordinate (y2)
  (-0.5, -9) coordinate (y23)
  (-1.1,-11.5) coordinate (y233)
  (3.8,-11.5) coordinate (y3)
  (3.6,-6.5) coordinate (y34)
  (3.5,-3.35) coordinate (y4)
  (2.5,-2.5) coordinate (y5);
  \draw[draw = myblue, fill = megalightgray] (y3) .. (y34) .. (y4) .. (y5) .. (y0) .. (y01) ..  (y1) .. (y12) .. (y2) .. (y23) .. (y233);

\node [rounded_block, dash dot, fill=none, minimum height=1.2 cm, minimum width = 3 cm] at (1,-3) {\begin{tabular}{l}
    double \\
    spend
\end{tabular}} ;

\node (node_0_0) at (0, 0) {};
\node (node_0_1) at (2, 0) {};
\node[block] (node_1_0) at (1, -1.5) {}; 

\draw[\arrowStyle, path fading = north] (node_1_0) -- (node_0_0); 
\draw[\arrowStyle, path fading = north] (node_1_0) -- (node_0_1); 

\node[block, fill = white, draw = myred, very thick] (node_2_0) at (0, -3) { $\tx{x}$}; 
\node[block,  fill = white, draw = myblue,very thick] (node_2_1) at (2, -3) { $\tx{y}$};
\draw[line width = 3, line cap = round, myred](-0.25, -3.25)--(0.25, -3.25);
\draw[line width = 3, line cap = round, myblue](1.75, -3.25)--(2.25, -3.25);

\draw[\arrowStyle] (node_2_0) -- (node_1_0); 
\draw[\arrowStyle] (node_2_1) -- (node_1_0); 

\node[block, fill = white] (node_3_0) at (-1, -4.5) {}; 
\node[block, fill = white] (node_3_1) at (1, -4.5) {}; 
\node[block, fill = white] (node_3_2) at (3, -4.5) {}; 

\draw[\arrowStyle] (node_3_0) -- (node_2_0); 
\draw[\arrowStyle] (node_3_1) -- (node_2_0);
\draw[\arrowStyle] (node_3_2) -- (node_2_1);
\draw[line width = 3, line cap = round, brown](-1.25, -4.75)--(-0.75, -4.75);
\draw[line width = 3, line cap = round, purple](0.75, -4.75)--(1.25, -4.75);
\draw[line width = 3, line cap = round, orange](2.75, -4.75)--(3.25, -4.75);

\node[block, fill = white] (node_4_0) at (0, -6) {}; 
\node[block, fill = white] (node_4_1) at (2, -6) {}; 
\draw[\arrowStyle] (node_4_0) -- (node_3_0); 
\draw[\arrowStyle] (node_4_0) -- (node_3_1);
\draw[\arrowStyle] (node_4_1) -- (node_3_2);
\draw[\arrowStyle,dashed] (node_4_1) -- (node_3_1);
\draw[line width = 3, line cap = round, brown](-0.25, -6.25)--(0.25, -6.25);
\draw[line width = 3, line cap = round, teal](1.75, -6.25)--(2.25, -6.25);

\node[block, fill = white] (node_5_0) at (1, -7.5) {}; 
\node[block, fill = white] (node_5_1) at (3, -7.5) {}; 
\draw[line width = 3, line cap = round, brown](0.75, -7.75)--(1.25, -7.75);
\draw[line width = 3, line cap = round, teal](2.75, -7.75)--(3.25, -7.75);

\draw[\arrowStyle, dashed] (node_5_0) -- (node_4_0); 
\draw[\arrowStyle] (node_5_0) -- (node_4_1);
\draw[\arrowStyle] (node_5_1) -- (node_3_2);

\node[block, fill = white] (node_6_0) at (0, -9) {}; 
\node[block, fill = white] (node_6_1) at (2, -9) {}; 
\draw[line width = 3, line cap = round, myblue](-0.25, -9.25)--(0.25, -9.25);
\draw[line width = 3, line cap = round, orange](1.75, -9.25)--(2.25, -9.25);

\draw[\arrowStyle] (node_6_0) -- (node_5_0);
\draw[\arrowStyle,dashed] (node_6_0) -- (node_3_0);

\draw[\arrowStyle] (node_6_1) -- (node_5_1);
\draw[\arrowStyle] (node_6_1) -- (node_4_1);

\node[block, fill = white] (node_7_0) at (1, -10.5) {}; 
\node[block, fill = white] (node_7_1) at (3, -10.5) {}; 
\draw[line width = 3, line cap = round, purple](0.75, -10.75)--(1.25, -10.75);
\draw[line width = 3, line cap = round, teal](2.75, -10.75)--(3.25, -10.75);

\draw[\arrowStyle] (node_7_0) -- (node_6_1);
\draw[\arrowStyle] (node_7_0) -- (node_6_0);
\draw[\arrowStyle] (node_7_0) -- (node_5_0);
\draw[\arrowStyle] (node_7_1) -- (node_6_1);
\draw[\arrowStyle] (node_7_1) -- (node_5_1);

\node (node_8_0) at (0, -12) {};
\node (node_8_1) at (2, -12) {};
\node (node_8_2) at (4, -12) {};

\draw[\arrowStyle, path fading = south] (node_8_0) -- (node_7_0);
\draw[\arrowStyle, path fading = south] (node_8_1) -- (node_7_0);
\draw[\arrowStyle, path fading = south] (node_8_1) -- (node_7_1);
\draw[\arrowStyle, path fading = south] (node_8_2) -- (node_7_1);

\draw[dotted] (-1.5,-6.75) -- (7,-6.75);
\draw[dotted] (-1.5,-11.25) -- (7,-11.25);

\node[block, fill = white, draw = myred, very thick] (node_2_0) at (4.5, -6) { $\tx{x}$}; 
\node[block,  fill = white, draw = myblue,very thick] (node_2_1) at (6, -6) { $\tx{y}$};
\draw[line width = 3, line cap = round, myred](4.25, -6.25)--(4.75, -6.25);
\draw[line width = 3, line cap = round, myblue](5.75, -6.25)--(6.25, -6.25);
\begin{scope}[transform canvas = {scale = \scaleFactorSupp}, shift = {(4.85 / \scaleFactorSupp, -5.65/\scaleFactorSupp)}]
\foreach \i/\j/\k in{0/60/myred, 60/120/white, 120/180/purple, 180/240/white, 240/300/brown, 300/360/white}
{
\draw[fill = \k](0, 0)  -- (\i:2) arc(\i:\j:2);
}
\end{scope}

\begin{scope}[transform canvas = {scale = \scaleFactorSupp}, shift = {(6.35 / \scaleFactorSupp, -5.65/\scaleFactorSupp)}]
\foreach \i/\j/\k in{0/60/white, 60/120/myblue, 120/180/white, 180/240/orange, 240/300/white, 300/360/teal}
{
\draw[fill = \k](0, 0)  -- (\i:2) arc(\i:\j:2);
}
\end{scope}

\node[block, fill = white, draw = myred, very thick] (node_2_0) at (4.5, -10.5) { $\tx{x}$}; 
\node[block,  fill = white, draw = myblue,very thick] (node_2_1) at (6, -10.5) { $\tx{y}$};
\draw[line width = 3, line cap = round, myred](4.25, -10.75)--(4.75, -10.75);
\draw[line width = 3, line cap = round, myblue](5.75, -10.75)--(6.25, -10.75);
\begin{scope}[transform canvas = {scale = \scaleFactorSupp}, shift = {(4.85 / \scaleFactorSupp, -10.15/\scaleFactorSupp)}]
\foreach \i/\j/\k in{0/60/myred, 60/120/white, 120/180/white, 180/240/white, 240/300/white, 300/360/white}
{
\draw[fill = \k](0, 0)  -- (\i:2) arc(\i:\j:2);
}
\end{scope}

\begin{scope}[transform canvas = {scale = \scaleFactorSupp}, shift = {(6.35 / \scaleFactorSupp, -10.15/\scaleFactorSupp)}]
\foreach \i/\j/\k in{0/60/white, 60/120/myblue, 120/180/brown, 180/240/orange, 240/300/purple, 300/360/teal}
{
\draw[fill = \k](0, 0)  -- (\i:2) arc(\i:\j:2);
}
\end{scope}

\end{tikzpicture}
    \caption{The Tangle is utilised as a voting layer for nodes to reach a consensus about the outcome of a conflict. Nodes agree on the winner between conflicting transactions $\tx{x}$ and $\tx{y}$ using a leaderless protocol. Different colours represent signatures of different nodes. The number of supporting nodes, shown on the right, increases for transaction $\tx{y}$ with time. The dashed references are so-called transaction references and allow to ``rescue'' transactions that voted for the ``losing part''.}  
    \label{fig:octopus_example}
\end{figure}
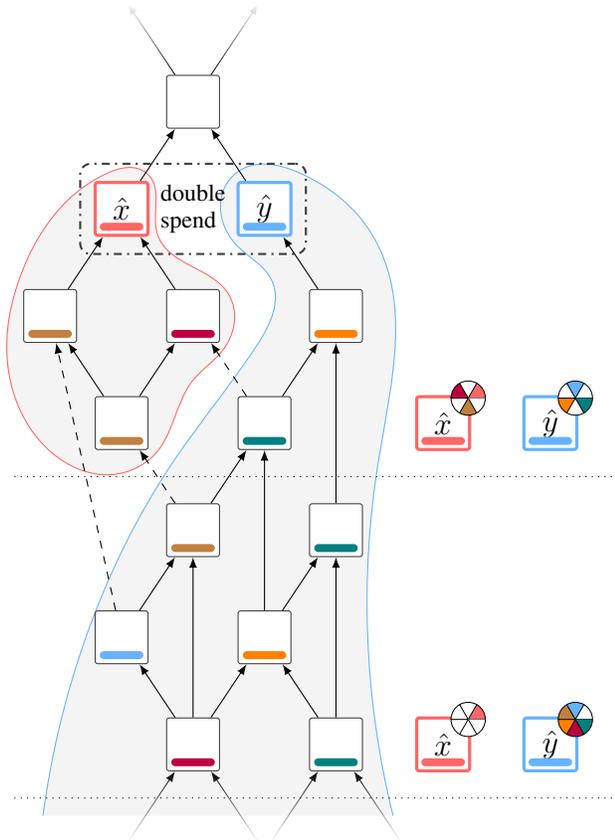

Generally, the voting mechanism can be applied to any DAG-based data structure with an append process that allows for referencing previous blocks. It requires three main ingredients:
the first essential ingredient is a reference scheme that efficiently casts and propagates votes. 
The second necessary ingredient is the construction of a generalised invariant data structure that allows conflicts to coexist (see Section~\ref{sec: ledger}). This feature allows to treat transactions ``optimistically''; every new incoming transaction is considered ``honest'' unless it conflicts with another transaction. Consequently, nodes may start to build on top of every new transaction, even though this transaction may turn out to be conflicting.  
The third ingredient is a voting mechanism, dubbed On Tangle Voting (OTV), that efficiently votes on a possible unbounded number of transactions simultaneously. 
The efficiency is achieved by maintaining a low block overhead since votes of other nodes can be piggy-backed through the  implicit voting mechanism. Also in contrast to classical Byzantine fault tolerance, nodes don't have to be monitored for activity since the issuance of transactions (casting of votes) is a clear sign of being functional.

\subsubsection*{Security}
Since the beginning of research on consensus protocols, the concept of security has been at the centre of attention. Any consensus protocol aims to reach consensus on a data. Some of the participants may be faulty or even active in preventing a consensus, and one is interested in the conditions under which consensus can be achieved. 

The security of a propose-vote consensus protocol is usually divided into two points; liveness and safety. Liveness means that any correct transaction is finally accepted by all honest participants, and safety means that all participants finally agree on the same set of transactions. The question of whether a given consensus protocol fulfils these properties depends largely on the model assumptions. Roughly, these can be divided into the communication model and the attacker model. 

In the most restrictive communication model, the synchronous model, many different solutions are known since the landmark result~\cite{LaShPe:82}.
However, this is not the case under the most general communication model, the asynchronous model, which does not assume any bounds on the transmission delay of block; commonly denoted by $\Delta$. One of the most famous results on consensus protocols is the FLP impossibility result~\cite{FLP} stating that in an asynchronous communication model, a single faulty participant can hinder the consensus finding. As the FLP impossibility result relies on specific configurations of block delays, many practitioners argued that it does not apply to real-world implementations as these particular situations are very unlikely to occur. In between these two extreme communications models, several intermediary models have been proposed, and many positive results have been obtained under stronger assumptions on the network delay, e.g. the partially synchronous model~\cite{Dwork1988}, the timed asynchronous model~\cite{CrFe:99}, and the asynchronous model with failure detection~\cite{ChTuTo:96}.

Besides the communication model, the adversary model plays an important role, especially in the security analysis of Nakamoto protocols. The protocol's security is commonly expressed in the amount of scarce resources, e.g. energy or computing power, that is necessary to attack the protocol and revert already confirmed transactions. Nakamoto~\cite{nakamoto2008bitcoin} analyzed this property by considering a specific attack, the so-called private-double spend attack. Note that here the classic communication model is the partial synchronous one. 
Over the last decade, a pertinent research question was the search for worst-case attacker strategies and the identification of the security threshold in terms of the percentage of the scarce resource controlled by an adversary. Tight consistency bounds were recently given in \cite{everythingIsRace} and \cite{tightConsistency} for several classes of  longest-chain type protocols. While these security thresholds do hold in the partial synchronous situation, they fail in the asynchronous setting, e.g.~\cite{PSS:17}. 

There is also a line of research that studies how an attacker can compensate its lower weight with more influence on the communication level. The most prominent of  such an attack is the balance attack,~\cite{balanceAttack}, which consists of delaying network communications between multiple subgroups of nodes with balanced mining power.

This discussion is of particular interest to us because we propose a framework for modelling the communication level and adversarial level jointly. Unsurprisingly, we obtain impossibility results in the asynchronous communication model. Still, under further synchronicity assumptions, we prove that the protocol guarantees liveness and safety (with a very high probability) if the adversarial weight does not exceed certain thresholds. The obtained security bounds are established for any possible attack strategy and are configurable by the protocol. 

The situations that lead to the impossibility results in the asynchronous model are frequently considered irrelevant for practicable purposes, e.g.~\cite{Guerraoui97, Aguilera}. The argument for this is that in real-world applications, the randomness in the block delays is so great that the particular situation cannot occur. While we partly agree with this reasoning concerning our OTV, we added a second synchronicity level to our core voting protocol to obtain a rigorous security threshold. For this reason, we see our consensus protocol as a two-layer solution. The first layer works in an asynchronous setting and allows fast and secure confirmation under normal network conditions. The second layer is based on an optional synchronization of the nodes that allows consensus finding under worst-case scenarios. The synchronized level relies on a decentralised random beacon or common coin that makes the protocol robust against attacks similar to the balanced attack described above. Randomization of consensus protocols to circumvent the impossibility results are known since \cite{Ben-Or}, which  introduces local randomness. A common coin was introduced \cite{Rabin} and is used in several approaches to increase the security in the asynchronous setting.

\subsubsection*{Performance}

Defining a measure for the efficiency of a consensus protocol is not an easy task since it relies on many different aspects. Natural choices are the number of blocks sent between the participants and, in synchronous models, the number of communication steps. In DLTs, common measures are the number of transactions per second and the time to confirmation. As our protocol uses implicit voting and no direct blocks are exchanged between the nodes, it is optimal in block complexity (if votes are cast through blocks that would have been sent anyway). We present estimates for the time to confirmation and show their dependence on the distribution of the weights. We do not evaluate quantitative performance measures such as throughput and energy consumption in this work. This type of study will be addressed in follow up research.

A common misunderstanding is that asynchronous consensus protocols are not appropriate for time-critical applications~\cite{Guerraoui97}. The fallacy is that synchronous  protocols assume strong synchronicity assumptions; however, the security is harmed once these assumptions are not satisfied. We argue that it is even the converse and that asynchronous protocols might be better suited for time-critical applications. Under a good communication situation,  transactions are approved much faster than in synchronous models based on network delay estimations with an essential security margin.

One main drawback of the leader-based architecture of blockchains is its lack of scalability capability. To make this more precise, let $\Delta$ be the network latency, $\lambda$ the block issuing rate, and $q$ the weight of the adversary. Then, following~\cite{GHOST, everythingIsRace}, the condition for the security of the protocol is expressed as
\begin{equation}
q < \frac{1-q}{1+(1-q)\lambda \Delta}.
\end{equation}
For the design of a system that should support resilience against a maximum adversary weight $q$, this equation informs about the bound on the maximum rate at which blocks can be issued safely.
A safety violation can occur, for example, if there is disagreement about the recent leader. These disagreements can be caused by blocks being produced in parallel~\cite{Misic2019} or due to certain attack scenarios~\cite{Iqbal2021,neuder2021}. As a consequence re-organisations of the blockchain may occur, in particular for DLTs, where the block production rate is high~\cite{lovejoy202}.

In our case there is no theoretical upper limit for the throughput of the protocol in this paper; however, the limits of scalability of our protocol still need to be investigated in future work.

\subsection{Related Work on DAG-Based Protocols}\label{sec:related work}

We already mentioned various related works in the general introduction. This section focuses on the general architecture and mention previous proposals that use DAGs in the underlying data structures. Blockchain-based protocols rely on a chain or ``linearisation'' of blocks that create a total order of transactions. These blocks have three purposes: \textit{leader election}, \textit{data transmission} and \textit{voting} for ancestor blocks through the chain structure, see \cite{Prism}.
Each of these aspects can be, individually or combined, addressed through a DAG-based component. The various proposals differ in how and which of  these components are replaced by a DAG. 

The most common approach is to use a DAG structure for the data transmission. This is the most natural approach since if blocks are created at a high rate compared to their propagation time, many competing or even conflicting blocks are created, leading to frequent bifurcation points of the chain. As this results in a performance loss, a natural proposal is to include not only the ``main chain'' but also bifurcations using additional references, e.g. \cite{Sompolinsky2013AcceleratingBT, lewenberg2015inclusive, SPECTRE, GHOST, PHANTOM}.

Protocols can also achieve a higher degree of parallelisation of the data transmission or writing access if all participants can write and propose blocks. This concurrent writing access removes considerably performance limitations of traditional blockchains. In blockchains where only a tiny proportion of participants can write to the ledger, and these participants are randomly chosen, e.g. by PoW or PoS, participants need to communicate the set of pending transactions to all their peers. This {\it memory pool} is a considerable performance limitation as nodes must broadcast transactions twice. Several interesting proposals allow participants to add concurrent blocks to the ledger and to construct a {\it distributed memory pool} in the form of a DAG. In the following, we give two approaches that differ in how consensus is achieved and in the underlying Sybil protection. More specifically the first utilises a permissioned setting, while the second employs a permissionless setting.

In the permissioned setting there is the following interesting line of research.
The aim is to construct an atomic broadcast protocol based on a combined  encoding of the data transmission history and voting on ``leader blocks''. Such protocols allow the network participants to reach a consensus on a total ordering of the received transactions, and this linearised output forms the ledger. The most robust protocols achieve Byzantine fault tolerance in asynchronous settings and reach optimal communication complexity, see Honeybadger\cite{HoneyBadger} and \cite{BEAT}. Improvements are proposed, for example, in  Hashgraph \cite{Hashgraph} and Aleph \cite{Gagol2019} and more recently in Narwhal \cite{Narwhal22} based on the encoding of the ``communication history'' in the form of a DAG. 
These protocols remove the  bottleneck of data dissemination of the classical Nakamoto consensus by decoupling the data dissemination from the consensus finding. Promising improvements for the consensus finding on top of the DAG-based memory pool were recently made in DAG Rider \cite{DAGRider} and Bullshark \cite{Bullshark}. We also want to mention \cite{BlockDAG} that analyses and discusses this kind of protocol from a more abstract and general point of view.

There is a common point with our approach to mention here. A DAG structure serves as a ``testimony'' of the communication among the nodes, and new blocks are used for (implicit) voting on previous blocks. In other words, the DAG is used for the two purposes of data transmission and voting. 
However, voting is done only over so-called ``anchor blocks'', leading to an {\it a posteriori} leader election and total ordering of the transactions. 
Furthermore, and as mentioned above, these DAG-based broadcast protocols are designed for permissioned networks, which leads to similar safety-liveness properties to standard BFT protocols. A difference  is, thus, that our protocol is designed for an asynchronous network environment and is not round-based as these proposals above.

In the permissionless setting, another route is taken by Prism \cite{Prism}. This approach explicitly decomposes the three purposes of blocks into three types: proposer blocks, transaction blocks and voter blocks. Having separate transaction blocks allows participants to issue transactions and removes the need for a memory pool. The three types of blocks form a structured DAG that allows a very efficient way to vote on ``leader blocks'' that eventually give consensus via total ordering. Our approach is orthogonal in that we do not distinguish between different kinds of blocks but that the underlying DAG delivers consensus without an additional tool. 
In an implementation \cite{Prism10000} of Prism, another DAG was used to increase the performance of the execution of the transaction. More precisely, \cite{Prism10000} used a scoreboarding technique to execute the (totally) ordered UTXO transactions in parallel. In our approach, we actively construct a DAG, called the Ledger DAG, that encodes the dependencies of the transactions. This DAG is created before reaching consensus and allows tracking dependencies between pending or conflicting transactions. It was  demonstrated in \cite{prismSC} that Prism can also support smart contract platforms and that in their implementation, the bottleneck is no longer the consensus but the execution of the smart contracts.

The main difference of our proposal to all the aforementioned protocols is that consensus is found on the heaviest DAG without the need for a ``linearisation'' using any leader selection. This reduces the purposes of blocks  to data transmission and voting.

We want to mention another class of DAG-based and leaderless consensus protocols. However, it is conceptually different from the proposals above and our proposal. In this kind of protocol, e.g.~\cite{Ava19,muller2021}, the voting is performed via direct queries between the peers and hence necessities an additional communication layer. A DAG structure is used in Avalanche \cite{Ava19} to  ``transitively'' vote on several blocks at once. We note, however, that the authors of \cite{Ava19} fail to analyze their proposed protocol properly, and the question of whether it has the desired properties remains unclear, e.g.~\cite[Section 2.3]{PoMu:21}.

Finally, let us note that the above is only a selection of previous work on DAG-based DLTs and refer the reader to \cite{Wang2020SoKDI} for a more detailed summary.

\subsection{List of Acronyms and Symbols}\label{sec: lists}

For the reader's convenience, in this section, we summarize important notations and acronyms that are used throughout the paper. Furthermore, in Appendix~\ref{sec: glossary} we provide a glossary of the terms in use in this paper.

\subsection*{Symbols}
{    
\raggedleft
\hspace{-.3cm}
    \begin{tabular}{ll}
    \multicolumn{2}{l}{\textbf{Set Symbols}}\\
    $\branchset $ &  set of branches \\
    $\conflictset$ &  set of conflicts \\
    $\nodeset$ &  set of nodes in network \\
    $\ledger$ &  ledger or set of transactions  \\
    $\blockset$ &  set of blocks \\

    \multicolumn{2}{l}{\textbf{DAGs}}\\
    $\LedgerDAG$ & Ledger DAG  \\
    $\TangleDAG$ & Tangle DAG  \\
    $\VotingDAG$ & Voting DAG  \\
    $\child{V}{x}$   &  set of children of vertex $x$ in DAG $D$=$(V,E)$ \\
    $\cone{f}{V}{x}$    & future cone of vertex $x$ in DAG $D$=$(V,E)$ \\
    $\cone{p}{V}{x}$  &  past cone of vertex $x$ in DAG $D$=$(V,E)$ \\
    $D$=$(V,E)$  &  directed acyclic graph (DAG) with vertex \\
    & set $V$ and edge set $E$ \\
    $\genesis$ &   genesis or vertex with out-degree zero\\
    $\maximal{V}{S}$ &  set of maximal elements in set $S$ (maximal \\
    & according to DAG $D$=$(V,E)$) \\
    $\minimal{V}{S}$ &  set of minimal elements in set $S$ (minimal \\
    & according to DAG $D$=$(V,E)$) \\
    $N_V(x)$ &  set of neighbours of a vertex $x$ in \\
    & graph $G$=$(V,E)$ \\
    $\le_V$ &  partial order on set $V$ (usually induced by\\
    & a given DAG $D$=$(V,E)$) \\
    $\parent{V}{x}$ &  set of parent of vertex $x$ in DAG $D$=$(V,E)$ \\
    $\mathrm{sprt}_{V}(x)$ &  supporters of $x$ in DAG $D$=$(V,E)$ \\

    \multicolumn{2}{l}{\textbf{Time Symbols}}\\
    $\tau_f(\cdot)$ &  time to confirmation defined on $\blockset$ \\
    $\tau_{cf}(\cdot)$ &  confluence time defined on $\blockset$ \\
    $\tau_{s}(\cdot)$ &  solidification time defined on $\blockset$ \\


\multicolumn{2}{l}{\textbf{Weight Functions}}\\
    $\weight(\cdot)$ &  weight function defined on $\nodeset$ \\
    $\AW(\cdot)$ &  Approval Weight defined on $\ledger$ \\
    $\WW(\cdot)$ &  Witness Weight defined on $\blockset$ \\
     
\end{tabular}
}

\subsection*{Acronyms}
{    
\raggedleft
\hspace{-.3cm}
\begin{tabular}{ll}
    \textbf{AW} & Approval Weight \\
    \textbf{dRNG} & Distributed Random Number Generator \\
    \textbf{DAG} & Directed Acyclic Graph \\
    \textbf{DLT} & Distributed Ledger Technology \\
    \textbf{OTV} & On Tangle Voting\\
    \textbf{P2P} & Peer-to-Peer \\
    \textbf{PoW} & Proof-of-Work \\
    \textbf{PoVP}  & Proof-of-Voting-Power \\
    \textbf{TSA}  & Tip Selection Algorithm \\
    \textbf{TTC}  & Time to Confirmation \\
    \textbf{UTXO}  & Unspent Transaction Output \\
    \textbf{WW}  & Witness Weight 
\end{tabular}
}

\subsection*{Graph structures}

We employ several graph structures as a base for the consensus protocol. Table~\ref{tab: graph overview} gives an overview of the utilised graphs.

\begin{table}[ht]
    \centering
    \begin{tabular}{|p{50pt}|p{75pt}|p{75pt}|}
        \hline
         DAGs & Vertices & Edges   \\    \hline
         Tangle & blocks & references  \\
         Ledger DAG & transactions  & spending relations  \\
          Voting DAG & blocks, transactions &  voting references \\
          \hline
    \end{tabular}
    \caption{Overview of DAGs}
    \label{tab: graph overview}
\end{table}

\section{Graph Theoretical Preliminaries}\label{sec: graph-theoretical notions}

In this section, we summarize basic graph theoretical notations that are used in the remaining part of the paper.

The set of integers between $1$ and $m$ is denoted by $[m]$. A \emph{graph} ${G}$ is a pair $(V,E)$, where $V$ denotes the set of vertices and $E$ denotes the set of edges. A graph is called \emph{directed} if every edge has its direction, e.g.  for an edge $(u,v)$, the direction goes from $u$ to $v$. 

\begin{definition}[DAG]
A \emph{directed acyclic graph (DAG)} is a directed graph with no directed cycles, i.e.   by following the directions of edges, we never form a closed loop.
\end{definition}

A vertex $v$ in a graph ${G}=(V,E)$ is called \textit{adjacent} to a vertex $u$ if $(u,v)\in E$. An edge $e\in E$ is said to be adjacent to a vertex $v\in V$ if $e$ contains $v$. The \textit{out-degree} and \textit{in-degree} of a vertex $v$ in a directed graph $G=(V,E)$ is the number of adjacent edges of the form $(v,u)$ and, respectively, $(u,v)$. A vertex in a graph is called \textit{isolated} if there is no edge adjacent to it.
\begin{definition}[Neighbours in a graph]~\label{def: neighbours in graph}
Let ${G}=(V,E)$ be a graph. For a vertex $v\in V$, define the \emph{set of neighbours} (or ${G}$-neighbours), written as $N_{{V}}(v)$\footnote{In the remainder of the paper, we will often identify the graph with its vertex set since for a given set of vertices $V$, we will have only one DAG $D=(V,E)$. Thereby, the set of neighbours $N_{{V}}(v)$ and other concepts that use $V$ as a subscript will be clear from the context.}, to be the vertices adjacent to $v$. 
\end{definition}
\begin{definition}[Parents, children and leaves in a DAG]
Let ${D}=(V,E)$ be a DAG. For a vertex $v\in V$, define the set of \emph{parents}, written as $\parent{V}{v}$, to be the set of vertices $u\in V$ such that $(v,u)\in E$. Similarly, we define the set of \emph{children}, written as $\child{V}{v}$, to be the set of vertices $u\in V$ such that $(u,v)\in E$. A vertex $v\in V$ with in-degree zero is called a \emph{leaf}. 
\end{definition}

\begin{definition}[Partial order induced by a DAG]
Let ${D}=(V,E)$ be a DAG. We write $u\le_{{V}} v$  for some $u,v\in V$ if and only if there exists a directed path from $u$ to $v$, i.e.  there are some vertices $w_0=u,w_1,\ldots,w_{s-1},w_s=v$ such that  $(w_{i-1},w_{i})\in E$ for all $i\in[s]$.  Furthermore, we note $u<_{{V}}v$ if $u\le_{{V}} v$ and $u \neq v$.
\end{definition}
Note that there could be different DAGs producing the same partial order. The DAG with the fewest number of edges that gives the partial order $\le_{{V}}$ is usually called the \emph{transitive reduction} of ${D}$ or the \emph{Hasse diagram} of $\le_{{V}}$.

\begin{definition}[Minimal subDAG induced by a set of vertices]\label{def: minimal subDAG}
Let ${D}=(V,E)$ be a DAG.  For a subset of vertices $S\subseteq V$, we define the \emph{minimal subDAG} of ${D}$ induced by $S$ to be the DAG ${D'}=(V',E')$ whose vertex set is $V'=S$ and there is an edge $(v,u)\in E'$ if and only if $u,v\in S$, $v<_{{V}}u$ and there is no $w\in S\setminus \{u,v\}$ such that $v<_{{V}} w <_{{V}}u$. 
\end{definition}

\begin{definition}[Maximal and minimal elements]\label{def: min and max elements} Let ${D}=(V,E)$ be a DAG and let $\le_{{V}}$ be the partial order induced by ${D}$. For a subset of vertices $S\subseteq V$,  an element $u\in S$ is called \emph{${D}$-maximal} (\emph{${D}$-minimal}) in $S$ if there is no $v\in S\setminus\{u\}$ such that $u\le_{{V}}  v$ ($v\le_{{V}}  u$). Define $\maximal{V}{S}$ and $\minimal{V}{S}$ to be the set of ${D}$-maximal and, respectively, ${D}$-minimal elements in $S$.
\end{definition}

\begin{definition}[Future and past cones] \label{def: future and past cones}
Let ${D}=(V,E)$ be a DAG. For $x\in V$, define the \emph{past cone} of $x$ in ${D}$, written as $\cone{p}{V}{x}$ 
to be the set of all vertices $y\in V$ such that $x\le_{{V}} y$. Similarly, define the \emph{future cone} of $x$ in ${D}$, written as $\cone{f}{V}{x}$ 
to be the set of all vertices $y\in V$ such that $y\le_{{V}} x$. 
\end{definition}

\begin{definition}[Past-closed sets]\label{def: past closed}
Let ${D}=(V,E)$ be a DAG. A subset $S\subset V$ is called \emph{${D}$-past-closed} if and only if for every $u\in S$, the past cone $\cone{p}{{V}}{u}$ is contained in $S$. 
\end{definition}

\section{Nodes and  Participation}\label{sec: nodes sybil}

At a high level, DLTs can be divided into permissioned and permissionless networks. In a permissioned setting, only selected parties can participate, while in the permissionless setting, anyone can join the network at any time. 
In a permissioned network, participants have either reading access or writing (validation) rights. A ``fully'' permissioned (or private) DLT  selects the participants in advance and restricts any activity in the network to these only. This is in contrast to a permissionless network where anybody can participate in the network and validate the ledger. 
Our protocol can work in both settings using a generic weight function on the participating nodes. In the permissionless setting, this weight function serves as a Sybil protection, and in the permissioned setting, this weight function regulates the participant's influence.

In Section~\ref{sec: nodes}, we introduce the network participants called nodes. In Section~\ref{sec: sybil protection} we describe a Sybil protection mechanism based on assigning specific weights to nodes. Finally, in Section~\ref{sec: rate control} we discuss how the writing ability of nodes is controlled by their weight.

\subsection{Network}\label{sec: nodes}
The network participants in the DLT are called \emph{nodes}, and we denote the set of all nodes by $\nodeset:=\{1,\ldots, N$\}, where $N$ is the total number of nodes.  
\emph{A priori}, different nodes may have different perceptions of the set of nodes.  For example, in a permissionless setting, for a node to join the network, the knowledge of a single node entrance point is sufficient. For the sake of a better presentation, we assume that every node is aware of every other node. Nodes directly communicate with a subset of other nodes, i.e.  its neighbours, via bidirectional channels. Thus, together all nodes create a peer-to-peer (P2P) overlay network. Nodes use public-key cryptography for their identification. Their unique node ID is derived from the public key, and all their blocks are signed with their private keys.

In contrast to other DLTs, where nodes can be divided into separate functional classes, we assume all nodes behave in the same way. Specifically, all nodes have two main roles. First, they propagate specific blocks through the network by receiving and sending these from and to their neighbours. 
Second, by creating new blocks and appending them to the data structure, nodes implicitly vote on the state of the previous blocks and their contained transactions; this procedure is called On Tangle Voting (OTV), see Section~\ref{sec: Voting}.
For the voting part, we assume a scarce resource, see Section~\ref{sec: sybil protection}. 
This resource endows every node with a certain \emph{weight} that is used for the implicit voting procedure.

\subsection{Sybil Protection}\label{sec: sybil protection}

A common problem in permissionless distributed systems is that it is easy to spawn a significant number of nodes, also known as the Sybil attack. Thus, any critical component must ensure that the action of nodes is limited, otherwise, it would be trivial for an attacker to gain a disproportionately large influence and corrupt the protocol.

To limit or prevent Sybil attacks, we assume that each node can be associated with a particular reputation or \textit{weight}  attributing them an equivalent proportion of voting power in the applied voting mechanism.

\begin{definition}[Weight]\label{def:weight}
For a given node $i\in\nodeset$ there is an associated weight $\weight(i)$, given by a function $\weight: \nodeset \to [0,1]$. The weights are assumed to be normalised, i.e.   
$$
\sum_{i\in \nodeset} \weight(i)=1.
$$
\end{definition}
The above weight function plays a crucial role in the validation process, see Sections~\ref{sec: witness weight}-\ref{sec: Approval Weight}.

\begin{remark}
We make use of the same weights as a control for the writing access  in Section~\ref{sec: rate control}. Note, however, that the weight for writing and validation could be different.
\end{remark}

A common way to implement such a weight is the so-called resource testing, where each identity has to prove the ownership of specific difficult-to-obtain resources.
Since in the cryptocurrency world, users own a certain amount of a scarce resource, i.e.  tokens, a practical Sybil protection mechanism can be based on proving the ownership of tokens and, thus, a certain amount of collateral.

Another way of implementing the weights is through delegation methods. The owners of source tokens, from which the  weights are derived, can then delegate these weights to any node of their choosing. This brings several key advantages. For example, fund owners can delegate weight to nodes that provide good service or revoke it when the node does not behave as expected, thus enabling the implementation of a ``reputation'' system. In the extreme case, this even allows decoupling the weights from the token distribution and incorporate real-world trust models.

Generally, the weight distribution in our system may change over time due to changes in the weights or inevitable churns (nodes join and leave). Due to the asynchronous nature of the protocol, the perception of the weights may then differ from node to node. 
The protocol design considers this effect and allows a certain divergence in the weight vector. This tolerance to different perceptions provides for some additional features of the protocol. 
However, a more detailed discussion of a divergence in the nodes' view on the weight vector is out of the scope of this paper. Thus, for simplicity, we make the following assumption.

\begin{ass}[Agreement on stability of weights]
All nodes in the network perceive the weight of node $i$ to be precisely $\weight(i)$. This weight is assumed to remain constant over time. 
\end{ass}

\subsection{Writing Access}\label{sec: rate control}

The distributed nature of the protocol and the Byzantine environment within which it operates puts several constraints on the writing access.  These constraints are even more critical for our protocol since it is not leader-based and does not rely on the intermediary of miners and block creators. Similar to~\cite{cullen2021} we require the following conditions:
\begin{enumerate}
    \item \textbf{Consistency:} if a block that is issued by an honest node is written to the (distributed) database by one honest node, it should eventually be written by all honest nodes.
    \item \textbf{Fairness:}  given a weight function and a maximum bandwidth, nodes can issue blocks at a rate proportional to their weight.
    \item \textbf{Security:} the above constraints are guaranteed in a Byzantine environment.  
\end{enumerate}

Consequently, the protocol should ensure that in  congested scenarios only a limited amount of blocks are propagated, i.e.  the block rate is capped by a certain throughput. Furthermore, this should happen fairly. These requirements prevent nodes from becoming overloaded and from inconsistencies in the ledger being created.
In principle, this could be enabled through fees and PoW, or more novel alternatives as the access control algorithm presented in  \cite{cullen2021}.

For the safe operation of the consensus mechanism, we assume the availability of such a mechanism. The required tool should provide guarantees on the constraints mentioned above. We make the following assumption. 

\begin{ass}[Writing access]
The writing access is controlled such that consistency, security, and fairness  in writing access are guaranteed for a given weight function $\weight$.
\end{ass}

\section{Block Structure and Witness Weight}\label{sec: tangle}
In this section, we introduce our protocol's  data structure concepts. To replicate a certain content over the distributed network, a node must wrap this content in a block.\footnote{In prior works, we refer to this object as a message.}
However, when the content is simply transactions, we require a block to contain only one transaction in its payload. This assumption is made for sake of a better presentation and can be relaxed, such that blocks contain more than one transaction. 
Moreover, each block has to refer to at least two blocks issued in the past. The latter requirement is motivated by the leaderless architecture of our protocol, in which each node can issue blocks independently of others. In addition, we discuss a particular metric on blocks, called the Witness Weight, that allows nodes to reliably understand when a significant fraction of the network has seen a given block.

In Section~\ref{sec: blocks}, we formally define a block. Section~\ref{sec: Tangle} discusses the Tangle, a DAG formed by blocks and their references. The local version of the Tangle seen by a specific node is introduced in Section~\ref{sec: local tangles}. Using the weight function for nodes introduced in Section~\ref{sec: sybil protection}, we formally define the Witness Weight of a given block in the local Tangle in Section~\ref{sec: witness weight} and show how to use this metric as a confirmation rule for blocks in Section~\ref{sec: WW confirmation rule}. The analysis of the growth of the Witness Weight is provided in Section~\ref{sec: heuristics WW}.

\subsection{Blocks}\label{sec: blocks}

The protocol's goal is to replicate certain content between the nodes in the network reliably. For example, this content could be the atomic updates of balances of fund owners.

This content is wrapped into an object that we call \textit{block}. A node that would like to initiate the addition of certain content to the Tangle across the network assembles such a block, which includes the content, $k$ references to previous blocks and the signature of the node (see Figure~\ref{fig:blockLayout}). We call the process of assembling and initial broadcasting the \textit{issuance} of a block. 
Each node that receives a new block forwards it to its neighbours.

\begin{figure}[t]
    \centering
    \definecolor{superlightgray}{RGB}{224, 224, 224}
\definecolor{mygray}{RGB}{240, 240, 240}
\definecolor{mygreen}{RGB}{0, 204, 153}
\definecolor{myred}{RGB}{255, 194,150}
\definecolor{lightsteelblue}{RGB}{176,196,222}
\definecolor{navajo}{RGB}{255,222,173}
\begin{tikzpicture}
\def\widthBlock{1.5}
\def\heightBlock{0.8}
\def\shiftdown{3*\heightBlock}

\tikzstyle{rounded_block}=[draw, rectangle, thick, minimum height=\heightBlock cm, minimum width = \widthBlock cm, text centered, rounded corners, draw=darkgray, font = \scriptsize]
\node [rounded_block, fill=white, minimum height=10*\heightBlock cm, minimum width = 3.6*\widthBlock cm] at (0,-\shiftdown) {} ;
\node [rounded_block, draw = lightgray, fill=superlightgray, minimum height=6.9*\heightBlock cm, minimum width = 3.2*\widthBlock cm] at (0,-0.6*\heightBlock-0.5*\shiftdown+0.3) {} ;
\node[font = \small] at (0,1.2*\heightBlock) {\textbf{Transaction}};
\node [rounded_block, dashed, fill=mygray, minimum height=2.8*\heightBlock cm, minimum width = 2.8*\widthBlock cm] at (0, -0.6*\heightBlock) {} ; 
\node [rounded_block, drop shadow, fill=mygreen] at (1.4*\heightBlock, 0) {Input $n$} ; 
\node [rounded_block, drop shadow, fill=mygreen] at (-1.4*\heightBlock, 0) {Input $1$} ; 
\node[font = \large] at (0,0) {...};
\node [rounded_block, drop shadow, fill=lightgray] at (1.4*\heightBlock, -1.2*\heightBlock) {\begin{tabular}{l}
    Unlock \\
    block $n$
\end{tabular}} ;
\node [rounded_block, drop shadow, fill=lightgray] at (-1.4*\heightBlock, -1.2*\heightBlock) {\begin{tabular}{l}
    Unlock \\
    block $1$
\end{tabular}} ;
\node[font = \large] at (0,-1.2*\heightBlock) {...};
\node [rounded_block, dashed, fill=mygray, minimum height=2.8*\heightBlock cm, minimum width = 2.8*\widthBlock cm] at (0, -0.6*\heightBlock-\shiftdown) {} ; 
\node [rounded_block, drop shadow, fill=myred] at (1.4*\heightBlock, -\shiftdown) {Output $m$} ; 
\node [rounded_block, drop shadow, fill=myred] at (-1.4*\heightBlock, -\shiftdown) {Output $1$} ; 
\node[font = \large] at (0,-\shiftdown) {...};
\node [rounded_block, drop shadow, fill=lightgray] at (1.4*\heightBlock, -1.2*\heightBlock-\shiftdown) {\begin{tabular}{l}
    Unlock \\
    block $m$
\end{tabular}} ;
\node [rounded_block, drop shadow, fill=lightgray] at (-1.4*\heightBlock, -1.2*\heightBlock-\shiftdown) {\begin{tabular}{l}
    Unlock \\
    block $1$
\end{tabular}} ;
\node[font = \large] at (0,-1.2*\heightBlock-\shiftdown) {...};

\node [rounded_block, drop shadow, fill=lightsteelblue] at (-1.4*\heightBlock, -2*\shiftdown) {Reference $1$} ; 
\node [rounded_block, drop shadow, fill=lightsteelblue] at (1.4*\heightBlock, -2*\shiftdown) {Reference $k$} ; 
\node[font = \large] at (0,-2*\shiftdown) {...};
\node [rounded_block, draw = lightgray, minimum width = 3.2*\widthBlock cm, fill=navajo] at (0, -1.3*\heightBlock-2*\shiftdown) {Signature of issuing node};

\end{tikzpicture}
    \caption{Simplified block layout with a transaction as content. The fund owner provides the node with the transaction. The node  wraps the transaction into a block and signs the block.
}
    \label{fig:blockLayout}
\end{figure}
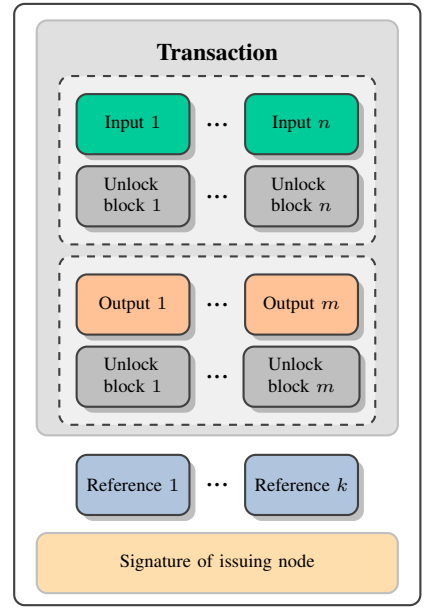

\begin{definition}[Block]\label{def: block}
A \textit{reference $\reference(x)$ of block $x$} is a pair $(r_y, v)$, where $r_y=\hash(y)$ is a  unique value that corresponds to a previously issued block $y$ and  $v$ is the value of a label. We define a block $x$ as an object with content
$$
x=(\{\reference_1(x),\ldots,\reference_k(x)\}, \tx{x}, \mathrm{nodeID(x)} ),
$$ where the $\reference_i(x)$'s are references, $\tx{x}$ is a transaction and $\mathrm{nodeID(x)}$ identifies the issuing node. 
\end{definition}
\begin{remark}\label{rem: hash function}
A collision-resistant hash function is used to map data of arbitrary size to a fixed-size binary sequence, i.e.  $\hash: \{0,1\}^*\to \{0,1\}^h$. Moreover, it is required that it is practicably impossible to find for a given sequence $x$ another sequence $x'$ such that $\hash(x) =\hash(x')$.
Throughout the remainder of the paper, we assume that a particular hash function is fixed and used by all participants.
\end{remark}
\begin{remark}
The label $v$ indicates the reference or voting type, as we will see later in Section \ref{sec: voting and voting DAG}.
\end{remark}

The issuing node obtains the content through a service-client relationship with the issuer of the content, which can be facilitated through an application programming interface (API) call. Alternatively, the node itself may also be the issuer of the content. An essential application for the content is the transfer of funds, i.e.  the consumption and creation of outputs. We call this type of content a \textit{transaction}. In this paper, for the sake of presentation, we will assume that each block contains exactly one transaction in its payload. However, in general, blocks are not limited to this use case.

As blocks will also be used to propagate votes, keeping track of the issuing nodes is crucial.
\begin{definition}[Issuer of a block]
For a block $x$, the node that issued $x$ is denoted as $\iss(x)$, where $\iss(x)\in\nodeset$.
\end{definition}

\subsection{The Tangle}\label{sec: Tangle}

The Tangle is a data structure built in accordance with the following rule as stated in the original paper~\cite{popov2015} of the Tangle: ``\textit{In order to issue a [block]~\footnote{The term used in the original whitepaper is \textit{transaction}, however, in this work we distinguish between the block and its contained  transaction.}, a node
chooses two other  [blocks] to approve''.}

More generally,  we modify this by allowing a block to reference up to  $k$ existing blocks.
The data structure takes the form of a DAG, where the blocks correspond to the vertices, and the references form the edges.

Let us define this data structure more formally. We denote the set of blocks by $\blockset$. 
There is a special block, called the \textit{genesis} and denoted by $\genesis$. This block does not contain any references.  Any other block has to directly refer to at least two (not necessarily distinct) blocks. Thereby, the reference relationship can be encoded into a DAG.

\begin{definition}[The Tangle]\label{def:tangle}
The Tangle $\TangleDAG$ is a DAG whose vertex set is  the set of blocks $\blockset$. There is a directed edge from  $y$ and $x$ in $\TangleDAG$ if and only if  $y$  directly refers to $x$.
\end{definition}

Using the notation from Section~\ref{sec: graph-theoretical notions}, we write $\le_{\blockset}$ to denote the partial order on the set of blocks induced by $\TangleDAG$. For a block $x\in \blockset$, the Tangle past and future cone of $x$ are denoted as $\cone{p}{\blockset}{x}$ and $\cone{f}{\blockset}{x}$, respectively. The parents and children of $x$ are written as $\parent{\blockset}{x}$ and $\child{\blockset}{x}$. If $x<_{\blockset} y$ we say that block $x$ \textit{approves} or \textit{references} block $y$. Specifically, if $x\in \child{\blockset}{y}$, then $x$ \textit{directly references} $y$; if $x\not \in \child{\blockset}{y}$ and $x<_{\blockset}$, then $x$ \textit{indirectly references} $y$. A leaf in the Tangle DAG is said to be a \textit{tip}.

\begin{example}
We refer to Figure~\ref{fig:conesInTangle} for an illustration of the Tangle and the Tangle future and past cones of block $x$.
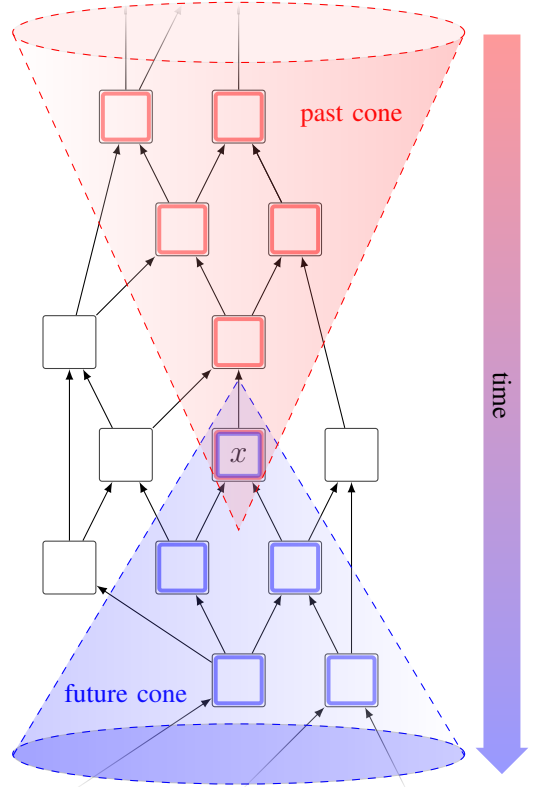
\begin{figure}[t]
    \centering
\begin{tikzpicture}
\def\xCoordinate{0.0}
\def\yCoordinate{0.0}
\def\yAdd{1.5}
\def\xAdd{0.75}
\def\innerSepar{-1.5}
\def\lineWidth{1.5}
\def\roundedCorners{1}
\def\opacityInternal{0.5}
\def\minHeight{20}
\def\minWidth{20}
\tikzstyle{block}=[draw, rectangle, minimum height=\minHeight pt, minimum width = \minWidth pt, text centered, rounded corners=\roundedCorners pt, draw=darkgray, font=\large]
\def\arrowStyle{-latex}
\node[block] (central_block) at (\xCoordinate,\yCoordinate) {$x$};
 \node [draw=red,opacity=\opacityInternal,rounded corners= \roundedCorners pt,line width = \lineWidth pt,
        inner sep=\innerSepar pt,fit=(central_block)] (red_central_block) {};
 \node [draw=blue, opacity=\opacityInternal, rounded corners =  \roundedCorners pt,line width = \lineWidth pt,
        inner sep=\innerSepar pt,fit=(red_central_block)] (blue_central_block) {};
        
\node[block] (blue_below_1_1) at (\xCoordinate-\xAdd,\yCoordinate-\yAdd) {};
 \node [draw=blue,opacity=\opacityInternal,rounded corners= \roundedCorners pt,line width = \lineWidth pt,
        inner sep=\innerSepar pt,fit=(blue_below_1_1)] (blue_blue_below_1_1) {};
 \draw[\arrowStyle] (blue_below_1_1) -- (central_block);  

\node[block] (blue_below_1_2) at (\xCoordinate+\xAdd,\yCoordinate-\yAdd) {};
 \node [draw=blue,opacity=\opacityInternal,rounded corners= \roundedCorners pt,line width = \lineWidth pt,
        inner sep=\innerSepar pt,fit=(blue_below_1_2)] (blue_blue_below_1_2) {};
 \draw[\arrowStyle] (blue_below_1_2) -- (central_block);

 
 \node[block] (blue_below_2_1) at (\xCoordinate,\yCoordinate-2*\yAdd) {};
 \node [draw=blue,opacity=\opacityInternal,rounded corners= \roundedCorners pt,line width = \lineWidth pt,
        inner sep=\innerSepar pt,fit=(blue_below_2_1)] (blue_blue_below_2_1) {};
 \draw[\arrowStyle] (blue_below_2_1) -- (blue_below_1_2); 
 \draw[\arrowStyle] (blue_below_2_1) -- (blue_below_1_1); 
 
 \node[block] (blue_below_2_2) at (\xCoordinate+2*\xAdd,\yCoordinate-2*\yAdd) {};
 \node [draw=blue,opacity=\opacityInternal,rounded corners= \roundedCorners pt,line width = \lineWidth pt,
        inner sep=\innerSepar pt,fit=(blue_below_2_2)] (blue_blue_below_2_2) {};
 \draw[\arrowStyle] (blue_below_2_2) -- (blue_below_1_2); 
 \draw[\arrowStyle, path fading = south]  (\xCoordinate,\yCoordinate-3*\yAdd) -- (blue_below_2_2) ; 
 \draw[\arrowStyle, path fading = south]  (\xCoordinate-3*\xAdd,\yCoordinate-3*\yAdd) -- (blue_below_2_1) ; 
 
 \draw[\arrowStyle, path fading = south]  (\xCoordinate+3*\xAdd,\yCoordinate-3*\yAdd) -- (blue_below_2_2) ;

 \node[block] (red_above_1) at (\xCoordinate,\yCoordinate+\yAdd) {};
 \node [draw=red,opacity=\opacityInternal,rounded corners= \roundedCorners pt,line width = \lineWidth pt,
        inner sep=\innerSepar pt,fit=(red_above_1)] (red_red_above_1) {};
 \draw[\arrowStyle] (central_block) -- (red_above_1);

  \node[block] (red_above_2_1) at (\xCoordinate+\xAdd,\yCoordinate+2*\yAdd) {};
  \draw[\arrowStyle] (red_above_1) -- (red_above_2_1);
   \node [draw=red,opacity=\opacityInternal,rounded corners= \roundedCorners pt,line width = \lineWidth pt,
        inner sep=\innerSepar pt,fit=(red_above_2_1)] (red_red_above_2_1) {};
  
  \node[block] (red_above_2_2) at (\xCoordinate-\xAdd,\yCoordinate+2*\yAdd) {};
  \node [draw=red,opacity=\opacityInternal,rounded corners= \roundedCorners pt,line width = \lineWidth pt,
        inner sep=\innerSepar pt,fit=(red_above_2_2)] (red_red_above_2_2) {};
  \draw[\arrowStyle] (red_above_1) -- (red_above_2_2);

  \node[block] (red_above_3) at (\xCoordinate,\yCoordinate+3*\yAdd) {};
   \node [draw=red,opacity=\opacityInternal,rounded corners= \roundedCorners pt,line width = \lineWidth pt,
        inner sep=\innerSepar pt,fit=(red_above_3)] (red_red_above_3) {};
  \draw[\arrowStyle] (red_above_2_2) -- (red_above_3);
  \draw[\arrowStyle] (red_above_2_1) -- (red_above_3);

 \node[block] (red_above_3_0) at (\xCoordinate-2*\xAdd,\yCoordinate+3*\yAdd) {};
   \node [draw=red,opacity=\opacityInternal,rounded corners= \roundedCorners pt,line width = \lineWidth pt,
        inner sep=\innerSepar pt,fit=(red_above_3_0)] (red_red_above_3_0) {};
  \draw[\arrowStyle] (red_above_2_2) -- (red_above_3_0);
  \draw[\arrowStyle] (red_above_2_1) -- (red_above_3);
 \draw[\arrowStyle,path fading=north] (red_above_3) -- (\xCoordinate,\yCoordinate+4*\yAdd);
  \draw[\arrowStyle,path fading=north] (red_above_3_0) -- (\xCoordinate-\xAdd,\yCoordinate+4*\yAdd);
  \draw[\arrowStyle,path fading=north] (red_above_3_0) -- (\xCoordinate-2*\xAdd,\yCoordinate+4*\yAdd);

\node[block] (central_grey_left) at (\xCoordinate-2*\xAdd,\yCoordinate) {};
 \draw[\arrowStyle] (central_grey_left) -- (red_above_1);  
 
\node[block] (central_grey_right) at (\xCoordinate+2*\xAdd,\yCoordinate) {}; 

 \draw[\arrowStyle] (central_grey_right) -- (red_above_2_1);
 \draw[\arrowStyle] (blue_below_1_2) -- (central_grey_right);

 \node[block] (above_grey_1) at (\xCoordinate-3*\xAdd,\yCoordinate+\yAdd) {};
 \draw[\arrowStyle] (above_grey_1) -- (red_above_2_2);
\draw[\arrowStyle] (central_grey_left) -- (above_grey_1);
\draw[\arrowStyle] (blue_below_1_1) -- (central_grey_left);
\draw[\arrowStyle] (above_grey_1) -- (red_above_3_0);

 \node[block] (below_grey_1) at (\xCoordinate-3*\xAdd,\yCoordinate-\yAdd) {};
\draw[\arrowStyle] (below_grey_1) -- (above_grey_1);
\draw[\arrowStyle] (below_grey_1) -- (central_grey_left);
\draw[\arrowStyle] (blue_below_2_1) -- (below_grey_1);
\draw[\arrowStyle] (blue_below_2_2) -- (central_grey_right);

  \def\xCoordinate{0.0}
  \def\yCoordinate{\yAdd * 2/3}
  \def\x{3.0}
  \def\y{5.0}
  \def\R{\x+0.004}
  \def\yc{\y+0.02}
  \def\e{0.4}
  \def\opacityCones{0.2}

  \node at (\xCoordinate-\x/2,\yCoordinate-\y*5/6)  {\color{blue} future  cone};
  \begin{scope}[rotate=0]
    \shade[right color=white,left color=blue,opacity=\opacityCones]
      (\xCoordinate-\x,\yCoordinate-\yc) arc (180:360:{\R} and \e) -- (\xCoordinate+\x,\yCoordinate-\yc) -- (\xCoordinate,\yCoordinate) -- cycle;
    \draw[fill=blue,opacity=\opacityCones]
      (\xCoordinate,\yCoordinate-\yc) circle ({\R} and \e);
    \draw[blue, dashed]
      (\xCoordinate-\x,\yCoordinate-\y) -- (\xCoordinate,\yCoordinate) -- (\xCoordinate+\x,\yCoordinate-\y);
    \draw[blue, dashed]
      (\xCoordinate,\yCoordinate-\yc) circle ({\R} and \e);
  \end{scope}

 \def\yCoordinate{-\yAdd * 2/3}
  \def\x{3.0}
  \def\y{6.6}
  \def\R{\x+0.004}
  \def\yc{\y+0.02}
   \node at (\xCoordinate+\x*1/2,\yCoordinate+\y*5/6)  {\color{red} past cone};
  \begin{scope}[rotate=0]
    \shade[left color=white, right color=red,opacity=\opacityCones]
    (\xCoordinate-\x,\yCoordinate+\yc) arc (180:360:{\R} and \e) -- (\xCoordinate+\x,\yCoordinate+\yc) -- (\xCoordinate,\yCoordinate) -- cycle;
     \draw[fill=red,opacity=0.3*\opacityCones]
      (\xCoordinate+0,\yCoordinate+\yc) circle ({\R} and \e);
    \draw[red, dashed]
      (\xCoordinate-\x,\yCoordinate+\y) -- (\xCoordinate,\yCoordinate) -- (\xCoordinate+\x,\yCoordinate+\y);
    \draw[red, dashed]
      (\xCoordinate,\yCoordinate+\yc) circle ({\R} and \e);
  \end{scope}
  
\node [single arrow,top  color=red, bottom color=blue,
single arrow head extend=3pt,transform shape, opacity = 2*\opacityCones, minimum height=280pt, text opacity=1, rotate = 270, anchor=west] 
at (\xCoordinate+\x*7/6,\yCoordinate+\y){time
};
 
\end{tikzpicture}
\caption{Future and past cones of a block $x$ in the Tangle}
    \label{fig:conesInTangle}
\end{figure}
\end{example}

\subsection{Local Tangles}\label{sec: local tangles}

Due to the distributed nature of the network, nodes can receive blocks at differing times or even out of order. The time at which a node first receives a block is called \textit{arrival time}.

Blocks can also be lost during their broadcast. While, generally, this could be problematic, the Tangle DAG allows for an elegant solution to remedy the loss by a process called \textit{solidification}. If a node receives a block for which the parents are unknown, it requests the missing block from its peers. Upon receipt of the missing parent block, the past cone is now complete (unless their parents are missing - in which case the node has to repeat this procedure recursively). Once a block's past cone is completed, the node flags the block  as \textit{solid}. The time of solidification of a block $x$ in node $i$ is denoted by  $\tau_{s,i}(x)$. We only consider blocks included in the Tangle after they are flagged solid.

As a consequence of the above, we can argue that there is no such thing as one Tangle in the network, as every node may have a different perception of it. Hence, at time $t$ a node $i$ is aware only of the block $x$ that satisfy $\tau_{s,i}(x)\leq t$.  
We denote by  $\blockset_{i,t}$ and $\LocalTangleDAG{i,t}$ the local perception of the block set and the Tangle DAG perceived from node $i$ at (local) time $t$.  Past and future cones then are also given in their \emph{local} forms $\cone{f}{\blockset_{i,t}}{x}$ and $\cone{p}{\blockset_{i,t}}{x}$.  
We omit subscripts and simply write $\TangleDAG =\LocalTangleDAG{i,t}$ if the dependence on $i$ and $t$ is clear from the context. 

\subsection{Witness Weight and Weighted Local Tangles}\label{sec: witness weight}

In the original Tangle whitepaper~\cite{popov2015} the cumulative weight of a block plays a crucial  role in the consensus finding. This cumulative weight is  the number of blocks referencing a given block. In case of a conflict, nodes follow the part of the Tangle that contains the largest cumulative weight.  

We adopt this fundamental idea to the setting where each node carries some weight. In this way, the nodes' weight replaces the PoW in the block creation as a Sybil protection mechanism.
The nodes' signature in each block links the issuing node to the block (see Section~\ref{sec: blocks}). Thus, a node can be associated with the set of blocks on the Tangle issued by that node, and the node's weight can be mapped to the blocks. 

\begin{definition}[Block Supporter and Witness Weight]
Let  $x\in\blockset_{i,t}$ be a block. Denote by $\supporter_{\blockset_{i,t}}(x)$ the set of nodes that issues a block in the future cone of $x$: 
$$
\supporter_{\blockset_{i,t}}(x)=\left\{j\in\nodeset: \ \exists y \in \cone{f}{\blockset_{i,t}}{x}, \ j=\iss(y)\right\}.
$$
We call nodes from $\supporter_{\blockset_{i,t}}(x)$ \textit{supporters} of $x$.
We define the function $\WW_{i,t}: \blockset_{i,t} \to [0,1]$ which is called the \emph{Witness Weight (WW)} of a block seen by node $i$ at time $t$ as follows
\begin{align}\label{eq: WW block}
\WW_{{i,t}} (x) :=& \sum_{j \in \supporter_{\blockset_{i,t}}(x)} \weight(j).
\end{align}
\end{definition}
As the total weight is normalised to $1$ the WW describes the percentage of weight approving a given block. Whenever it is clear from the context, we omit indices $i$ and $t$.

\begin{example}
In Figure~\ref{fig: approval weight and supporters}, we give an example of the set of nodes approving given blocks $x$, $y$ and $z$. We use unique colours in the bottom of blocks to represent signatures of different issuing nodes. One can readily check that $\supporter_{\blockset}(x)$ consists of nodes corresponding to brown, cyan and gray colours.
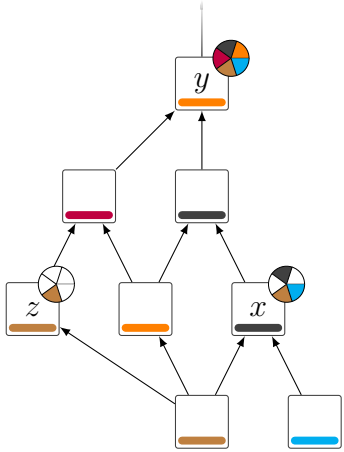
\begin{figure}[h]   
\centering
\begin{tikzpicture}
\def\xCoordinate{0.0}
\def\yCoordinate{0.0}
\def\yAdd{1.5}
\def\xAdd{0.75}
\def\innerSepar{-1.5}
\def\lineWidth{1.5}
\def\roundedCorners{1}
\def\opacityInternal{0.5}
\def\minHeight{20}
\def\fracyAdd{1/6}
\def\lineWidthBelow{3}
\def\minWidth{20}
\def\scaleFactorSupp{0.12}
\def\minWidthCM{\minWidth*0.0352778*3/4}
\tikzstyle{block}=[draw, rectangle, minimum height=\minHeight pt, minimum width = \minWidth pt, text centered, rounded corners=\roundedCorners pt, draw=darkgray, font=\large]
\tikzstyle{block_colored}=[draw, rectangle, minimum height= 2 pt, minimum width = \minWidth * 5 / 6 pt, rounded corners=\roundedCorners  pt, draw=darkgray]
\def\arrowStyle{-latex}
\node[block] (central_block) at (\xCoordinate,\yCoordinate) {};

 \draw[line width = \lineWidthBelow, line cap = round, darkgray] (\xCoordinate - \minWidthCM/2 ,\yCoordinate- \fracyAdd*\yAdd)--(\xCoordinate + \minWidthCM/2 ,\yCoordinate - \fracyAdd*\yAdd);
        
\node[block] (blue_below_1_1) at (\xCoordinate-\xAdd,\yCoordinate-\yAdd) {};
 \draw[\arrowStyle] (blue_below_1_1) -- (central_block);  
 \draw[line width = \lineWidthBelow, line cap = round, orange] (\xCoordinate-\xAdd - \minWidthCM/2 ,\yCoordinate-\yAdd - \fracyAdd*\yAdd)--(\xCoordinate-\xAdd + \minWidthCM/2 ,\yCoordinate-\yAdd - \fracyAdd*\yAdd);

\node[block] (blue_below_1_2) at (\xCoordinate+\xAdd,\yCoordinate-\yAdd) {$x$};
\begin{scope}[transform canvas={scale=\scaleFactorSupp},shift={(12.5*\xAdd,-6.5*\yAdd)}]
\foreach \i/\j/\k in {0/72/white,72/144/darkgray,144/216/white, 216/288/brown, 288/360/cyan}
{
\draw[fill=\k] (0,0)  -- (\i:2) arc (\i:\j:2);
}
\end{scope}
\draw[line width = \lineWidthBelow, line cap = round, darkgray] (\xCoordinate +\xAdd- \minWidthCM/2 ,\yCoordinate-\yAdd- \fracyAdd*\yAdd)--(\xCoordinate +\xAdd+ \minWidthCM/2 ,\yCoordinate -\yAdd- \fracyAdd*\yAdd);

 \draw[\arrowStyle] (blue_below_1_2) -- (central_block);

 
 \node[block] (blue_below_2_1) at (\xCoordinate,\yCoordinate-2*\yAdd) {};
 \draw[line width = \lineWidthBelow, line cap = round, brown] (\xCoordinate- \minWidthCM/2 ,\yCoordinate-2*\yAdd- \fracyAdd*\yAdd)--(\xCoordinate + \minWidthCM/2 ,\yCoordinate -2*\yAdd- \fracyAdd*\yAdd);
 
 \draw[\arrowStyle] (blue_below_2_1) -- (blue_below_1_2); 
 \draw[\arrowStyle] (blue_below_2_1) -- (blue_below_1_1); 
 
 \node[block] (blue_below_2_2) at (\xCoordinate+2*\xAdd,\yCoordinate-2*\yAdd) {};
  \draw[line width=\lineWidthBelow, line cap = round, cyan] (\xCoordinate + 2*\xAdd- \minWidthCM/2 ,\yCoordinate-2*\yAdd- \fracyAdd*\yAdd)--(\xCoordinate + 2*\xAdd+ \minWidthCM/2 ,\yCoordinate -2*\yAdd- \fracyAdd*\yAdd);
 
 \draw[\arrowStyle] (blue_below_2_2) -- (blue_below_1_2);

 

 \node[block] (red_above_1) at (\xCoordinate,\yCoordinate+\yAdd) {$y$};
 \draw[\arrowStyle] (central_block) -- (red_above_1);
  \draw[line width = \lineWidthBelow, line cap = round, orange] (\xCoordinate - \minWidthCM/2 ,\yCoordinate+\yAdd- \fracyAdd*\yAdd)--(\xCoordinate + \minWidthCM/2 ,\yCoordinate +\yAdd- \fracyAdd*\yAdd);
  \begin{scope}[transform canvas={scale=\scaleFactorSupp},shift={(4.3*\xAdd,10.2*\yAdd)}]
\foreach \i/\j/\k in {0/72/orange,72/144/darkgray,144/216/purple, 216/288/brown, 288/360/cyan}
{
\draw[fill=\k] (0,0)  -- (\i:2) arc (\i:\j:2);
}
\end{scope}

   \draw[\arrowStyle, path fading = north]  (red_above_1) -- (\xCoordinate,\yCoordinate+7/4*\yAdd); 

\node[block] (central_grey_left) at (\xCoordinate-2*\xAdd,\yCoordinate) {};
 \draw[\arrowStyle] (central_grey_left) -- (red_above_1);  \draw[line width = \lineWidthBelow, line cap = round, purple] (\xCoordinate - 2*\xAdd- \minWidthCM/2 ,\yCoordinate- \fracyAdd*\yAdd)--(\xCoordinate - 2*\xAdd + \minWidthCM/2 ,\yCoordinate - \fracyAdd*\yAdd);

\draw[\arrowStyle] (blue_below_1_1) -- (central_grey_left);

 \node[block] (below_grey_1) at (\xCoordinate-3*\xAdd,\yCoordinate-\yAdd) {$z$};
 \draw[line width = \lineWidthBelow, line cap = round, brown] (\xCoordinate - 3*\xAdd - \minWidthCM/2 ,\yCoordinate-\yAdd- \fracyAdd*\yAdd)--(\xCoordinate -3*\xAdd + \minWidthCM/2 ,\yCoordinate -\yAdd- \fracyAdd*\yAdd);
  
\draw[\arrowStyle] (below_grey_1) -- (central_grey_left);
\draw[\arrowStyle] (blue_below_2_1) -- (below_grey_1);

   \begin{scope}[transform canvas={scale=\scaleFactorSupp},shift={(-21.5*\xAdd,-6.5*\yAdd)}]
\foreach \i/\j/\k in {0/72/white,72/144/white,144/216/white, 216/288/brown, 288/360/white}
{
\draw[fill=\k] (0,0)  -- (\i:2) arc (\i:\j:2);
}
\end{scope}

\end{tikzpicture}
\label{fig: approval weight and supporters}
 \caption{A Tangle DAG, where the issuing node of a block can be identified with a unique colour shown in the bottom of the block. The colors of the supporters of blocks $x,y,z$ are depicted in the top-right corners. 
 }
\label{fig: approval weight and supporters}
\end{figure}
\end{example}

We proceed with two trivial statements saying that the WWs of blocks are monotonically increasing toward the genesis and the WW of a block can only grow over time. 

\begin{lemma}[Monotonicity of the WW]\label{lem: monotonicity of WW}
For any two blocks $x,y\in\blockset$ such that $x\le_{\blockset} y$, it holds that $\supporter_{\blockset}(x)\subseteq \supporter_{\blockset}(y)$ and, hence, $\WW(x)\le \WW(y)$.
\end{lemma}
\begin{lemma}[Growth of the WW]\label{lem: growth of WW}
For any block $x\in\blockset$, node $i\in\nodeset$ and time instants $t_1$ and $t_2$ such that $t_1<t_2$, it holds that $\supporter_{\blockset_{i,t_1}}(x)\subseteq \supporter_{\blockset_{i,t_2}}(x)$ and, hence, $\WW_{{i,t_1}} (x)\le \WW_{{i,t_2}} (x)$.
\end{lemma}
A more delicate analysis of the growth of the WW under certain assumptions is provided in Section~\ref{sec: heuristics WW}.

\subsection{Confirmation Rule for Blocks}\label{sec: WW confirmation rule}

The block stream is controlled by the writing access control, see Section~\ref{sec: rate control}. A priori, this control alone may not be sufficient to guarantee that all nodes see all blocks in the network.
However, to guarantee the safety of the system, nodes must have consensus on which blocks should  permanently be accepted in the data set $\blockset$, otherwise, inconsistencies between the nodes could arise.
If such a consensus is achieved, we consider a block \textit{confirmed}. 
Furthermore, to maintain consistency in the data structure $\TangleDAG$, a block $x$ can only be confirmed if all blocks in $\cone{p}{\blockset}{x}$ are confirmed.

Tools that provide information about the confirmation status of blocks, with specific safety and liveness considerations, are generally referred to as  \textit{confirmation rule}. 
We design such a tool based on the concept of  WWs of the blocks. 
The WW allows the nodes and users to create their subjective confirmation criterion. The larger the WW of a block, the higher the probability that the block will be in the ledger forever. This idea is similar to the  ``depth'' of a transaction in a blockchain. Therefore, the actual confirmation criterion may depend on the protocol environment and the underlying use case. 

\begin{definition} [Confirmed block] \label{def: confirmed block}
Let $\theta\in(0.5,1]$ be a fixed threshold. We say that a block $x\in\blockset$ is \emph{confirmed} for a node $i\in\nodeset$ at time $t$ if  $\WW_{{i,s}} (x)\ge \theta$, for some $s\leq  t$.
\end{definition}

Once a block is confirmed for a node, it remains confirmed forever. This irreversibility of the confirmation status places some strong requirements on the convergence of this status. More specifically, once a single node reaches the threshold for a given block, all nodes should reach this threshold eventually with a very high probability. 

In an honest scenario, this assumption can be easily satisfied since a high WW also represents that a large proportion of nodes have ``seen'' a given block and issued a block approving it.  If the default tip selection algorithm is suitably chosen and followed by sufficiently many nodes all nodes will attach blocks eventually to the future cone of that block with a very high probability (for more details, see 
Section~\ref{sec: heuristics WW}).
In Section~\ref{sec: security} we discuss the liveness and safety of the protocol in detail.

\subsection{Growth of Witness Weight}\label{sec: heuristics WW}

In this section, we model the block issuance and discuss the growth of the WW and its dependencies on the protocol environment. 

We consider the following assumption.

\begin{ass}[Issuing rate]\label{ass:issuingRate}
Each node $i\in\nodeset$ issues blocks at a Poisson rate $\lambda_i$ (per second). 
The rate~$\lambda_i$ is proportional to the corresponding weights $\weight(i)$ (see Definition~\ref{def:weight}), i.e.  $\lambda_i=\lambda\weight(i)$ for some constant $\lambda>0$. We assume that every node issues blocks independently of the other nodes. The rate of issuance for all nodes is then  
$$
\lambda=\sum_{i\in\nodeset}\lambda_i.
$$
\end{ass}

\begin{remark}
Under Assumption~\ref{ass:issuingRate} the times between two successive blocks from a node $i\in\nodeset$ are independent and exponentially distributed with parameter $\lambda_i$.
\end{remark}

To develop a heuristic for the WW we use the following approach. We assume that there is an ``omniscient observer'', that is instantly aware of all blocks issued by all nodes. The observer's perception of the state may differ from the perception of a given node, however, these differences have no substantial influence on the heuristic result. We refer to \cite{Kusmierz2019, PKS:18} where this method has already been proven to lead to good heuristics. 
This view is reflected in the notation by omitting the index $i$. For instance, $\blockset_t$ denotes the set of blocks perceived by this omniscient observer at time $t$ and $\WW_{t}(x)$ denotes the corresponding WW of a block $x$ at time $t$. 

Let $x$ be a block issued at time $t_0$ and denote by $E_i(\delta,x)$ the event that node~$i$ issues a block in the time interval $[t_0,t_0+ \delta]$ in the future cone of $x$.  We write $\mathbf{1}\{E_i(\delta,x)\}$ for the indicator function of this event; it is equal to $1$ if the event occurred and $0$ otherwise.

For $t= t_0+\delta$, the WW of block $x$ perceived by the omniscient observer satisfies 
\begin{equation}\label{eq:AWbound1}
\WW_{t}(x) = \sum_{i=1}^N \weight(i) \mathbf{1}\{E_i(\delta,x)\}.
\end{equation}
Node $i$ issues blocks with rate $ \lambda\weight(i)$ and, thus, we have that
\begin{equation}\label{eq:upperBoundRate}
    \P (E_i(\delta,x )) \leq 1- \exp(-  \delta\lambda\weight(i) ).
\end{equation}
Note that the equality does not necessarily hold since not all new incoming blocks have to witness block $x$.
Taking the expectation in Equation~\eqref{eq:AWbound1} and applying Inequality \eqref{eq:upperBoundRate} we obtain
\begin{equation}\label{eq:AWbound2}
    \E [\WW_{t}(x)] \leq \sum_{i=1}^N \weight(i) \left(1- \exp(- \delta\lambda\weight(i))\right).
\end{equation}

The formula given in~\eqref{eq:AWbound1} holds in the very general setting. For the analysis of the protocol, it is, however, important to consider a specific weight distribution. 
Probably the most appropriate modelings of weight distributions rely on universality phenomena. The most famous example of this universality phenomenon is the central limit theorem. While the central limit theorem is suited to describe statistics where values are of the same order of magnitude, it is not appropriate to model more heterogeneous situations where the values might differ in several orders of magnitude. These heterogeneous situations are frequently  described by a Zipf law and appear in  many fields; e.g. city populations, internet traffic data, the formation of P2P communities, company sizes, and science citations. We refer to~\cite{Li2002ZipfsLE} for a brief introduction and more references, and to~\cite{Adamic2002ZipfsLA,kondor2014dotherich,com_sel} for the appearance of Zipf's law on the internet, computer networks, and DLTs. 

We consider a situation with $N$ elements or nodes. Zipf's law  predicts that the (normalised) weight of the node of rank $r$  is given by \begin{equation}\label{eq:Zipf_law}
	\weight(r) = \frac{r^{-s}}{ \sum_{j = 1}^N j^{-s}},
\end{equation}
where $s\in [0,\infty)$ is the Zipf parameter. Since the weights $\weight(\cdot)$ in~\eqref{eq:Zipf_law} only depends on two parameters, $s$ and $N$, this provides a convenient model to investigate the performance of the protocol in a wide range of network situations. For instance, a homogeneous network with $N$ nodes having equal weight can be modeled by choosing $s = 0$. With increasing value of $s$ the network becomes increasingly centralised.

\begin{figure*}[t]
    \centering
\includegraphics[width=0.98\textwidth]{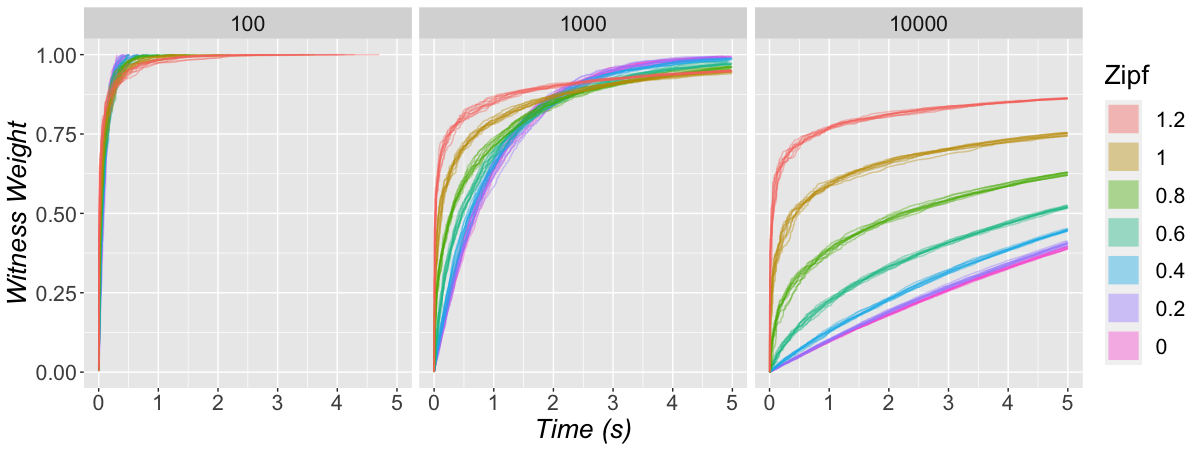}
    \caption{Growth of the issued WW in Example~\ref{ex:growthWW} with $1000$ blocks per second. We see the different behaviour for $100$ nodes (left), $1000$ nodes (middle) and $10.000$ nodes (right). The growth depends essentially on the chosen Zipf parameter $s$ (in colour) and the number of nodes.}
    \label{fig: growth of the approval weight}
\end{figure*}

\begin{example}\label{ex:growthWW}
We refer to Figure~\ref{fig: growth of the approval weight}. The growth of the WW depends on several factors, notably the issuing rate $\lambda$ and the distribution of the nodes' weight.  In the case of a Zipf distribution the weight depends on two parameters, the number of nodes $N$ and the Zipf parameter $s$. The upper bound  ~\eqref{eq:AWbound2} is a convex monotone function in $\delta$ and $\lambda$. The dependence on the parameters $N$ and $s$ is not so obvious. For this reason, we perform some Monte-Carlo simulations for $N\in\{100, 1000, 10000\}$ and $s\in\{0, 0.2, 0.4, 0.6, 0.8, 1, 1.2\}$, and $\lambda=1000$.\footnote{Typical values for the Zipf parameter found in  popular cryptocurrencies for the top 100 addresses are in the range between $s=0.7$ and $s=1.2$ \cite{com_sel}.} For a given $t_0$ we approximate the WW at time $t_0+t$ to be the sum of the weights of all nodes having issued a block during the time interval $[t_0, t_0+t]$. This provides a lower bound estimate for the WW of a block that is issued at $t_0$.
In Figure~\ref{fig: growth of the approval weight} every line corresponds to one realisation of the growth of the issued WW (for $t_0=0$).
\end{example}

\subsection{Estimates on Time to Confirmation}\label{sec:TTC}

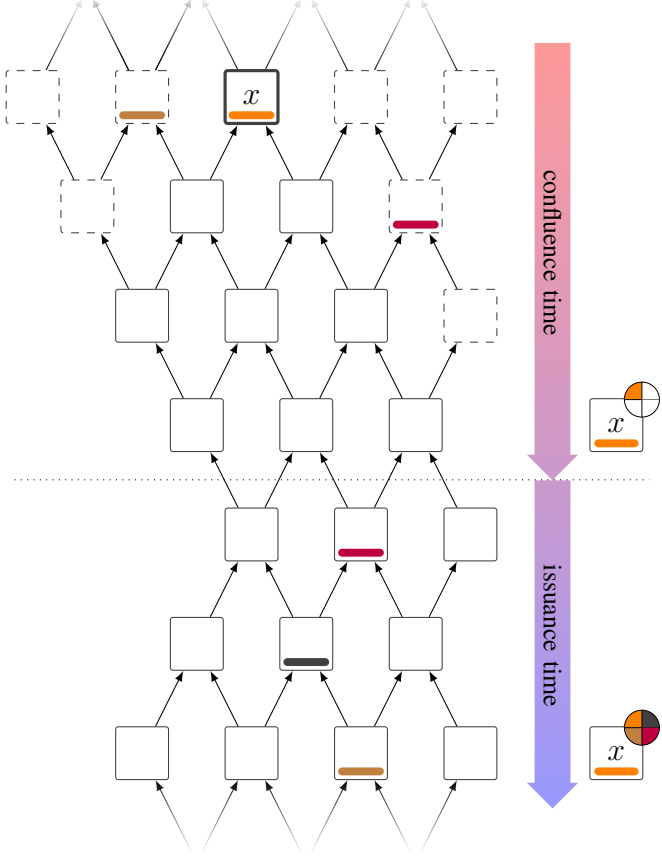
\begin{figure}
    \centering
\long\def\ifnodedefined#1#2#3{%
    \@ifundefined{pgf@sh@ns@#1}{#3}{#2}%
}
\makeatother
\pgfmathsetseed{\number\pdfrandomseed}

\begin{tikzpicture}[scale=0.97]
\def\xCoordinate{0.0}
\def\yCoordinate{0.0}
\def\yAdd{1.5}
\def\xAdd{0.75}
\def\innerSepar{-1.5}
\def\lineWidth{1.5}
\def\roundedCorners{1}
\def\opacityInternal{0.5}
\def\minHeight{20}
\def\fracyAdd{1/6}
\def\lineWidthBelow{3}
\def\minWidth{20}
\def\scaleFactorSupp{0.12}
\def\minWidthCM{\minWidth*0.0352778*3/4}
\tikzstyle{block}=[draw, rectangle, minimum height=\minHeight pt, minimum width = \minWidth pt, text centered, rounded corners=\roundedCorners pt, draw=darkgray, font=\large]
\tikzstyle{block_colored}=[draw, rectangle, minimum height= 2 pt, minimum width = \minWidth * 5 / 6 pt, rounded corners=\roundedCorners  pt, draw=darkgray]
\def\arrowStyle{-latex}

\node (node_0_0) at (0, 0) {};
\node (node_0_1) at (1.5, 0) {};
\node (node_0_2) at (3, 0) {};
\node (node_0_3) at (4.5, 0) {};
\node[block, dashed] (node_1_0) at (0.75, -1.5) {}; 

\draw[\arrowStyle, path fading = north] (node_1_0) -- (node_0_0); 
\draw[\arrowStyle, path fading = north] (node_1_0) -- (node_0_1); 

\node[block, dashed] (node_1_m_1) at (-0.75, -1.5) {}; 
\draw[\arrowStyle, path fading = north] (node_1_m_1) -- (node_0_0); 

\draw[\arrowStyle, path fading = north] (node_1_0) -- (node_0_0); 
\draw[\arrowStyle, path fading = north] (node_1_0) -- (node_0_1); 

\draw[line width = 3, line cap = round, brown](0.49, -1.75)--(1.01, -1.75);
\node[block, very thick] (node_1_1) at (2.25, -1.5) {$x$}; 
\draw[line width = 3, line cap = round, orange](1.99, -1.75)--(2.51, -1.75);

\draw[\arrowStyle, path fading = north] (node_1_1) -- (node_0_1); 
\draw[\arrowStyle, path fading = north] (node_1_1) -- (node_0_2); 
\node[block, dashed] (node_1_2) at (3.75, -1.5) {}; 

\draw[\arrowStyle, path fading = north] (node_1_2) -- (node_0_2); 
\draw[\arrowStyle, path fading = north] (node_1_2) -- (node_0_3); 
\node[block, dashed] (node_1_3) at (5.25, -1.5) {}; 
\draw[\arrowStyle, path fading = north] (node_1_3) -- (node_0_3); 

\node[block] (node_2_1) at (1.5, -3) {}; 

\draw[\arrowStyle] (node_2_1) -- (node_1_0); 
\draw[\arrowStyle] (node_2_1) -- (node_1_1); 
\node[block] (node_2_2) at (3, -3) {}; 

\node[block, dashed] (node_2_0) at (0, -3) {}; 
\draw[\arrowStyle] (node_2_0) -- (node_1_0); 
\draw[\arrowStyle] (node_2_0) -- (node_1_m_1);

\draw[\arrowStyle] (node_2_2) -- (node_1_1); 
\draw[\arrowStyle] (node_2_2) -- (node_1_2); 
\node[block, dashed] (node_2_3) at (4.5, -3) {}; 

\draw[\arrowStyle] (node_2_3) -- (node_1_2); 
\draw[\arrowStyle] (node_2_3) -- (node_1_3); 

\draw[line width = 3, line cap = round, purple](4.24, -3.25)--(4.76, -3.25);
\node[block] (node_3_0) at (0.75, -4.5) {}; 

\draw[\arrowStyle] (node_3_0) -- (node_2_0); 
\draw[\arrowStyle] (node_3_0) -- (node_2_1); 

\node[block] (node_3_1) at (2.25, -4.5) {}; 

\draw[\arrowStyle] (node_3_1) -- (node_2_1); 
\draw[\arrowStyle] (node_3_1) -- (node_2_2); 
\node[block] (node_3_2) at (3.75, -4.5) {}; 

\draw[\arrowStyle] (node_3_2) -- (node_2_2); 
\draw[\arrowStyle] (node_3_2) -- (node_2_3); 
\node[block, dashed] (node_3_3) at (5.25, -4.5) {}; 

\draw[\arrowStyle] (node_3_3) -- (node_2_3); 
\node[block] (node_4_1) at (1.5, -6) {}; 

\draw[\arrowStyle] (node_4_1) -- (node_3_0); 
\draw[\arrowStyle] (node_4_1) -- (node_3_1); 
\node[block] (node_4_2) at (3, -6) {}; 

\draw[\arrowStyle] (node_4_2) -- (node_3_1); 
\draw[\arrowStyle] (node_4_2) -- (node_3_2); 
\node[block] (node_4_3) at (4.5, -6) {}; 

\draw[\arrowStyle] (node_4_3) -- (node_3_2); 
\draw[\arrowStyle] (node_4_3) -- (node_3_3); 
\node[block] (node_5_1) at (2.25, -7.5) {}; 

\draw[\arrowStyle] (node_5_1) -- (node_4_1); 
\draw[\arrowStyle] (node_5_1) -- (node_4_2); 
\node[block] (node_5_2) at (3.75, -7.5) {}; 

\draw[\arrowStyle] (node_5_2) -- (node_4_2); 
\draw[\arrowStyle] (node_5_2) -- (node_4_3); 
\draw[line width = 3, line cap = round, purple](3.49, -7.75)--(4.01, -7.75);
\node[block] (node_5_3) at (5.25, -7.5) {}; 

\draw[\arrowStyle] (node_5_3) -- (node_4_3); 
\node[block] (node_6_1) at (1.5, -9) {}; 

\draw[\arrowStyle] (node_6_1) -- (node_5_1); 
\node[block] (node_6_2) at (3, -9) {}; 

\draw[\arrowStyle] (node_6_2) -- (node_5_1); 
\draw[\arrowStyle] (node_6_2) -- (node_5_2); 
\draw[line width = 3, line cap = round, darkgray](2.74, -9.25)--(3.26, -9.25);
\node[block] (node_6_3) at (4.5, -9) {}; 

\draw[\arrowStyle] (node_6_3) -- (node_5_2); 
\draw[\arrowStyle] (node_6_3) -- (node_5_3); 
\node[block] (node_7_0) at (0.75, -10.5) {}; 

\draw[\arrowStyle] (node_7_0) -- (node_6_1); 
\node[block] (node_7_1) at (2.25, -10.5) {}; 

\draw[\arrowStyle] (node_7_1) -- (node_6_1); 
\draw[\arrowStyle] (node_7_1) -- (node_6_2); 
\node[block] (node_7_2) at (3.75, -10.5) {}; 

\draw[\arrowStyle] (node_7_2) -- (node_6_2); 
\draw[\arrowStyle] (node_7_2) -- (node_6_3); 
\draw[line width = 3, line cap = round, brown](3.49, -10.75)--(4.01, -10.75);
\node[block] (node_7_3) at (5.25, -10.5) {}; 

\draw[\arrowStyle] (node_7_3) -- (node_6_3); 
\node (node_8_1) at (1.5, -12) {};
\draw[\arrowStyle, path fading = south] (node_8_1) -- (node_7_0); 
\draw[\arrowStyle, path fading = south] (node_8_1) -- (node_7_1); 
\node (node_8_2) at (3, -12) {};
\draw[\arrowStyle, path fading = south] (node_8_2) -- (node_7_1); 
\draw[\arrowStyle, path fading = south] (node_8_2) -- (node_7_2); 
\node (node_8_3) at (4.5, -12) {};
\draw[\arrowStyle, path fading = south] (node_8_3) -- (node_7_2); 
\draw[\arrowStyle, path fading = south] (node_8_3) -- (node_7_3); 
\node[single arrow, top  color = red, bottom color = violet, single arrow head extend = 3pt, transform shape, opacity = 0.4, minimum height = 6cm, text opacity = 1, rotate = 270, anchor = west] at ( 6.375, -0.75) {confluence time}; 
\node[single arrow, top  color = violet, bottom color = blue, single arrow head extend = 3pt, transform shape, opacity = 0.4, minimum height = 4.5cm, text opacity = 1, rotate = 270, anchor = west] at ( 6.375, -6.75) {issuance time}; 
\node[block] (node_x_1) at (7.25, -6) {$x$}; 
\draw[line width = 3, line cap = round, orange](7, -6.25)--(7.5, -6.25);
\begin{scope}[transform canvas = {scale = \scaleFactorSupp}, shift = {(7.6 / \scaleFactorSupp, -5.65/\scaleFactorSupp)}]
\foreach \i/\j/\k in{0/90/white, 90/180/orange, 180/270/white, 270/360/white}
{
\draw[fill = \k](0, 0)  -- (\i:2) arc(\i:\j:2);
}
\end{scope}

\node[block] (node_x_2) at (7.25, -10.5) {$x$}; 
\draw[line width = 3, line cap = round, orange](7, -10.75)--(7.5, -10.75);
\begin{scope}[transform canvas = {scale = \scaleFactorSupp}, shift = {(7.6 / \scaleFactorSupp, -10.15/\scaleFactorSupp)}]
\foreach \i/\j/\k in{0/90/darkgray, 90/180/orange, 180/270/brown, 270/360/purple}
{
\draw[fill = \k](0, 0)  -- (\i:2) arc(\i:\j:2);
}
\end{scope}

\draw [dotted] (-1,-6.75) -- (7.7,-6.75);

\end{tikzpicture}
    \caption{Illustration of the Tangle to display the confluence time and issuance time. The colours in the bottom of the blocks represents the issuing nodes with significant weight. We demonstrate the colours of the ``heavy'' supporters of block $x$ on the right after each time period. The dashed blocks correspond to blocks that are not in the future cone of $x$.}
    \label{fig:confirmationTime}
\end{figure}

As discussed in Section~\ref{sec: WW confirmation rule}, a confirmation rule is essential for many use cases, and time to confirmation (TTC) is undoubtedly a vital performance measure of every consensus protocol. As a thorough analysis of the TTC is out of the scope of this paper, we give a first ``heuristic'' upper bound in this section. 

\begin{definition}[Time to confirmation]\label{def:TTC}
We define the time to confirmation of a block $x$ (at level $\theta$) by a node $i$ as
\begin{equation}
    \tau_{f,i}=\tau_{f,i}(x) := \inf\{ t>0: \WW_{{i,t}}(x) \ge \theta\} - \tau_{s,i}(x),
\end{equation}
where $\tau_{s,i}(x)$ is the solidification time of $x$ (see Section~\ref{sec: local tangles}). 
\end{definition}

In the remainder of this section, we omit index $i\in\nodeset$ since the provided analysis is relevant to all nodes.  We divide the TTC into two periods. During the first period, we  wait for the confluence time $\tau_c=\tau_c(x)$ until a given block $x$ is contained in the past cone of (almost) all current tips. During the second period, the issuance time $\tau_{iss}$, we let the WW grow until it reaches the threshold $\theta$. The TTC $\tau_f$ is then bounded above by 
\begin{equation}\label{eq:TTC approx}
    \tau_f \leq \tau_c + \tau_{iss}.
\end{equation}
Estimates for $\tau_{iss}$ are obtained from~\eqref{eq:AWbound1} and this formula can be simplified for specific choices of the weights (see Example~\ref{rem:TTCexample}).

\begin{example}\label{ex:confluence issuing time}
We demonstrate the confluence time and the issuance time with the help of Figure~\ref{fig:confirmationTime}. Blocks with a solid frame are in the future cone of block $x$. After the confluence time, all blocks approve $x$. The yellow, green and purple colours represent blocks by nodes that hold significant weight, i.e.  the nodes have a large influence on the confirmation. Once the cumulative weight of the nodes issued in the future cone of $x$ reaches the threshold $\theta$, the block becomes confirmed.
\end{example}

With some additional assumptions,
we can obtain estimates for the confluence time $\tau_c$ similarly to~\cite{popov2015}.
Our first assumption is that the delay between block creation and the moment that other nodes in the network receive this block is constant. 

\begin{ass}[Constant network delay]\label{ass:constantNetworkDelay}
We assume that the time between the block creation and until any other node receives this block equals some constant $h$.
\end{ass}

\begin{definition}[Number of tips]
Let $L(t)$ be the total number of tips of the Tangle at time $t$.
\end{definition}

As mentioned in Section~\ref{sec: local tangles}, there is no ``objective Tangle,'' and every node has its own perception. Nevertheless, previous work~\cite{Kusmierz2019} showed that the approximation made in this section leads to reasonable approximations for some quantitative properties of the Tangle, such as the number of tips and confluence times. For this reason, we omit the subscript ``$i$'' and work with a unique objective Tangle in this section.
Similar to~\cite{popov2015} we assume the number of tips to be in a stationary regime.
 
\begin{ass}[Constant Tangle width]\label{ass:constantWidth}
We assume that the number of tips $L(t)$ of the Tangle is stationary and has mean $L_0$. 
\end{ass}

Using Assumptions~\ref{ass:issuingRate},~\ref{ass:constantNetworkDelay}, and~\ref{ass:constantWidth} we follow the heuristics described in~\cite[Section 3] {popov2015}. A first observation is that at any given time $t$ there are on average $\lambda h$ \emph{hidden tips}, those blocks that have been issued after $t-h$ but are not yet visible to the network. As in~\cite{popov2015} we assume that typically there are $r$ \emph{revealed tips}, those that have been attached before $t-h$ but are still tips. Hence, we can write the total (average) number of tips as $L_0=r+\lambda h$. By Assumption~\ref{ass:constantWidth} we consider that the number of tips $L(t)$ is roughly stationary. This implies that since $\lambda h$ tips join the tip pool, during the same time, roughly $\lambda h$ blocks that have been tips at time $t-h$ became referenced and are no longer tips. Hence, the tip pool of size $L_0$ can be divided into $r$ revealed tips and $\lambda h$ blocks that are no longer tips. This division leads to the crucial observation that a new block (with $k$ parents) approves on average $k r / (r+\lambda h)$ (revealed) tips. Moreover, in the stationary situation where the tip pool size $L_0$ stays approximately constant, the mean number of chosen tips should be equal to $1$; otherwise, the number of tips would change. Solving 
$k r / (r+\lambda h)=1$ leads to
\begin{equation}\label{eq:L0}
    L_0 = L_0^{(k)} = \frac{ k \lambda h}{k-1}.
\end{equation}
This result, predicted in~\cite{popov2015},  has been confirmed through simulation studies in~\cite{PKS:18, Kusmierz2019} and theoretical results in \cite{king:21}.

A first consequence of \eqref{eq:L0} is that, if $L_0$ is large, the expected time for a block to be approved for the first time is approximately 
\begin{equation}\label{eq:timeToApproval}
    h+L_0/(k\lambda)=h+ \frac{h}{(k-1)}.
\end{equation}
The size of the tip pool is naturally linked to the growth of the WW of a given block; the larger the tip pool the slower the growth of the WW.

\begin{remark}
For any given $\lambda$ and $h$ we can  choose $k$ sufficiently large such that $k >  L_0^{(k)}$. In this case, blocks are referenced essentially immediately after they become \emph{visible}, however, at the cost of a larger block size. 
\end{remark}

We can proceed similar to~\cite{popov2015} to obtain that
\begin{equation}\label{eq:tauC}
    \tau_c \approx \frac{h}{W\left(\frac{(k-1)^2}{k} \right)} \left( \log L_0 + \log \varepsilon \right),
\end{equation}
where $\log$ denotes the natural logarithm function and $W$ is the principal branch of the Lambert $W$-function, which is defined as the inverse function to $z=we^{w}$, i.e.  $w=W(z)$. For large $k$, we can use the approximation 
$$
W\left(\frac{(k-1)^2}{k} \right)\approx 2 \log (k-1) - \log k \approx \log k
$$
and, hence, obtain
\begin{equation}\label{eq:tauCklarge1}
    \tau_c \approx \frac{h}{\log k} \log(L_0) \approx \frac{1}{\log k} h \log(\lambda h).
\end{equation}
In Section~\ref{sec:confluenceTime}, we will give  more details on the derivation of the confluence time.

\begin{example}\label{rem:TTCexample}
The behaviour of the issuing time $\tau_{iss}$ heavily depends on the actual weight distribution, e.g. see Figure \ref{fig: growth of the approval weight}.
However, the extreme case of all nodes having the same weight can be treated  more analytically. Extreme is meant here in the sense that the growth of the WW is to some extent the smallest. 
Hence, let $\weight(i)=1/N$ for all nodes $i\in\nodeset$ and assume that we want to get a bound on the confirmation time, i.e.  the first time a given block $x$ reaches  $\WW_{t}(x) \ge \theta$. Denote by $X_i$ the first time a block was sent from node $i\in\nodeset$. The vector of these times $(X_1, \ldots, X_N)$ can be ordered in increasing order and we obtain the so-called order statistics $X_{(1)}, \ldots, X_{(N)}.$ In the case where all $X_i$ follow the same exponential distribution $\mathrm{Exp}(\gamma)$ the distribution of the $i$th order statistic is given by
\begin{equation}
    X_{(i)} \sim \frac{1}{\gamma} \sum_{j=1}^i \frac{Z_j}{N-j+1},
\end{equation}
where the $Z_j$ are i.i.d.~exponential random variables with parameter $1.$
Eventually, the time it takes that $\lceil\theta N\rceil$ nodes issued a block is distributed as $X_{(\lceil \theta \cdot N\rceil)}$.
The expectation is given by
$$
   \E[ X_{(i)}] = \frac{1}{\gamma} \sum_{j=1}^i \frac{1}{N-j+1},
$$
with $i =\lceil \theta \cdot N\rceil$.
Using a standard integral approximation for the above sum, we obtain for large $N$ that
$$
   \E[ X_{(i)}] \approx  \frac{N}{\lambda} \left(\log(N) - \log(N-i)  \right).
$$
Hence, for $i=\theta N$, 
$$
   \tau_{iss} \approx \E[ X_{(i)}] \approx  \frac{N}{\lambda} \left( - \log(1-\theta) \right).
$$
Combing this result with the bound~\eqref{eq:tauCklarge1}  on the confluence time in (\ref{eq:TTC approx})  we obtain the following asymptotic upper bound on the TTC for large $k$ (and the other parameters fixed):
$$
    \tau_f \lesssim   \frac{1}{\log k} h \log(\lambda h) + \frac{N}{\lambda} \left( - \log(1-\theta) \right) .
$$
\end{example}

\section{The Ledger}\label{sec: ledger}
This section introduces several novel concepts to represent transactions and their interrelationships. Recall that in the standard UTXO conflict-free model, transactions specify the outputs of previous transactions as inputs and create new outputs by spending (or consuming) the inputs. No two transactions are consuming the same input. Such a conflict-free data structure can be implemented in a network where a consensus mechanism filters transactions. The latter is typically done by choosing a ``leader'' among the participants, and the leader adds a block of transactions to the conflict-free ledger. To bypass this ``centralised'' bottleneck, we propose the concept of the Reality-based UTXO Ledger, an augmented version of the standard conflict-free UTXO Ledger that allows more than one output spend. We refer the reader to the parallel work~\cite{RealitiesLedger2022}, where we discuss all concepts in detail.

In Section~\ref{sec: utxo and transactions}, we recall the definition of a transaction in the UTXO model and the ledger, which is a set of all transactions. In Section~\ref{sec:ledgerState}, we introduce definitions of conflicting transactions, conflicts and branches, which represent proper subsets of ``non-conflicting conflicts''. A reality is a maximal possible branch, and restricting a ledger to a reality results in the conflict-free UTXO Ledger. Finally, in Section~\ref{sec: reality selection algorithm} we discuss how nodes could choose a reality given an abstract weight function defined on the set of conflicts. The selected reality allows a node to express its opinion when issuing new blocks and validating transactions.

\subsection{UTXO Model and Transactions}\label{sec: utxo and transactions}

In the Unspent Transaction Output (UTXO) model transactions specify the outputs of previous transactions as inputs and spend them by creating new outputs.

Thus, a transaction consists of a list of inputs and a list of outputs, see Figure~\ref{fig:blockLayout}. Note that outputs must be unique. The uniqueness is typically achieved by creating the output ID with the involvement of a hash function. For example, the output ID could be the concatenation of the index of an output and the hash of a transaction's content.
Every output represents a specific amount of the underlying cryptocurrency. The value of all inputs, i.e.  spent outputs, must equal the value of all outputs of a transaction.
With each output comes a declaration by whom and under which conditions it can be spent. Under unlock conditions, e.g. a signature proving ownership of a given input's address, the transaction issuer is allowed to spend the inputs. We refer to Figure~\ref{fig:blockLayout} for a general transaction layout.  

As said in Section~\ref{sec: blocks}, blocks contain transactions in their payload. Hereafter, we write $\tx{x}$ to denote the transaction contained in the payload of a block $x$.

Let us define the transactions and ledger model more formally. We follow the approach of~\cite{IOTASC}.
\begin{definition}[Output and input]
An \textit{output} is a pair of a value $\val\in \mathbb{R}^+$ and an unlock condition $\cond$. We write $o=(\val, \cond)$ to denote the output. An \textit{input} $i$ is a reference to an output. We say the input \textit{consumes} the output. 
\end{definition}
\begin{definition}[Transaction]
A transaction $\tx{x}$ is a collection of inputs $\inputs(\tx{x}),$ outputs $\outputs(\tx{x})$, and unlock proofs $\unlock(i)$, $i\in \inputs(\tx{x})$, where
\begin{enumerate}[leftmargin=*]
    \item $\inputs(\tx{x})=(i_1,\ldots, i_n)$ is a list of inputs, i.e.  references to unconsumed outputs. We say that those outputs are spent or consumed by transaction $x$;
    \item $\outputs(\tx{x})=(o_1,\ldots, o_m)$ is a list of new outputs produced by transaction $\tx{x}$;
    \item $\unlock(i)$ is a proof which performs verification of the unlock conditions of each input $i$ of transaction $\tx{x}$. This is usually done by cryptographic proof of authorization that ensures that the issuer of the transaction satisfies the condition $\cond$ of the consumed outputs.
\end{enumerate}
\end{definition}
\begin{definition}[Ledger]
The \emph{ledger} is a set of transactions and denoted as $\ledger$. 
\end{definition}

The UTXO ledger starts at the so-called \textit{genesis} which contains outputs and no inputs. We emphasize that we use the same term for the ultimate predecessor of all blocks and all transactions. Recall that the genesis-block is written as $\genesis$, whereas the genesis-transaction  will be denoted as $\tx{\genesis}$.

Typically every output can be consumed by at most one transaction and, hence,  the value of all unspent outputs is conserved overall. Specifically, in the standard conflict-free UTXO model, the ledger can not contain a so-called \textit{double spend}, i.e.  two transactions that consume the same output of a transaction.

In the following section, we alleviate this conflict-free restriction and allow the Ledger to contain conflicting transactions.

\subsection{Reality-based Ledger}\label{sec:ledgerState}

In this section, we propose an augmented version of the standard conflict-free UTXO ledger model that allows containing double spends. We suggest different structures that can be used for tracking conflicting transactions without the need for consensus.  

First, we explain how the transactions and their in- and outputs result in a DAG structure. 
The information contained in the Ledger DAG is split into the Conflict Graph, which keeps track of the conflicting transactions only. Then we introduce the concept of branches. A branch forms a possible non-conflicting state of the ledger. We will then derive a concept, called a \textit{reality}, which allows us to reduce $\ledger$ to a maximal subset of transactions that yield a conflict-free (\textit{Reality-based}) ledger.

\begin{definition}[Ledger DAG]\label{def: ledgerDAG}
We define the Ledger DAG $\LedgerDAG$ to be a DAG whose vertex set is the ledger $\ledger$. There is a directed edge $(\tx{x},\tx{y})$ in the edge set of  $\LedgerDAG$ if and only if an input of $\tx{x}$ references an output of $\tx{y}$.
\end{definition}

We refer to Appendix~\ref{sec: toy example}, where we demonstrate this graph together with many other core concepts. Using the notation from Section~\ref{sec: graph-theoretical notions}, we write $\le_{\ledger}$ to denote the partial order on the set of transactions induced by $\LedgerDAG$. The past cone of a transaction $\tx{x}$ is denoted by $\cone{p}{\ledger}{\tx{x}}$, i.e.  transaction $\tx{x}$ spends value directly or indirectly  from transactions in $\cone{p}{\ledger}{\tx{x}}\setminus\{\tx{x}\}$.

Typically, the addition of transactions to this type of data structure is such that only transactions, which create no conflict with any previously recorded transactions are allowed to be added, i.e. the Ledger DAG is conflict-free. However, this requires a consensus mechanism that pre-selects transactions.

Now we introduce a new design for a ledger, where this constraint is replaced by a relaxed one -- namely, a new transaction $\tx{x}$ can be added to the ledger if $\inputs(x)$ are references to outputs which are not already consumed in $\cone{p}{\ledger}{\tx{x}} \setminus \{ \tx{x}\}$. In the following, we provide an overview of some of the most important concepts of the proposed solution that allows conflicting transactions to co-exist. Thereby, we start with a formal definition of conflicts and conflicting transactions.
\begin{definition}[Conflicts]\label{def:conflicts}
Two distinct transactions $\tx{x},\tx{y}\in \ledger$ are \emph{directly conflicting} if they have at least one input in common. 
A transaction $\tx{x}\in \ledger$ is called a \textit{conflict} if and only if there exists a transaction $\tx{y}\in \ledger\setminus\{\tx{x}\}$ such that $\tx{x}$ and $\tx{y}$ are directly conflicting.
\end{definition}
\begin{definition}[Conflicting transactions]\label{def:conflictingTx}
Two distinct transactions $\tx{x}_1,\tx{y}_1\in \ledger$ are said to be \textit{conflicting}  if there exist distinct $\tx{x}_2,\tx{y}_2\in \ledger$ with $\tx{x}_1 \le_{\ledger} \tx{x}_2$ and $\tx{y}_1\le_{\ledger} \tx{y}_2$ such that $\tx{x}_2$ and $\tx{y}_2$ are directly conflicting. 
\end{definition}
The interrelations between conflicts can be encoded with the help of the Conflict DAG and the Conflict Graph.
\begin{definition}[Conflict DAG and Conflict Graph]
 The set of all conflicts is denoted by $\conflictset$ and dubbed the \emph{set of conflicts} of the ledger $\ledger$. We define the Conflict DAG $\ConflictDAG$ to be the minimal subDAG of the Ledger DAG induced by $\conflictset \cup \tx{\genesis}$ (cf. Definition~\ref{def: minimal subDAG}). We define the Conflict Graph $\ConflictGraph$ to be the graph whose vertex set is $\conflictset$ and two conflicts are connected by an edge if and only if these conflicts are conflicting (as transactions).
\end{definition}

We can group transactions based on whether they conflict with each other or not.
\begin{definition}[Conflict-free set and conflicting sets]\label{def: conflicting sets}
A subset of transactions $S\subseteq \ledger$ is called \textit{conflict-free} if it does not contain any two  conflicting transactions. We also say that $S_1\subseteq \ledger$ is \textit{conflict-free with respect to} $S_2\subseteq \ledger$ if there is no $\tx{x}_1\in S_1$ and $\tx{x}_2\in S_2$ such that $\tx{x}_1$ and $\tx{x}_2$ are conflicting. Alternatively, $S_1$ is \textit{conflicting} with $S_2$ if $S_1$ is not conflict-free with respect to $S_2$.
\end{definition}
We further specialise conflict-free sets and introduce the notion of branches.

\begin{definition}[Branch and set of branches]\label{def:branch}
A set of conflicts ${B}\subseteq \conflictset$ is called a \emph{branch} if and only if the two properties hold:
\begin{enumerate}[leftmargin=*]
    \item ${B}$ is conflict-free (cf. Definition~\ref{def: conflicting sets});
    \item ${B}$ is $\ConflictDAG$-past-closed (cf. Definition~\ref{def: past closed}).
\end{enumerate}
Define $\branchset$ to be the set of all branches. A branch that represents the empty set is called the \emph{main branch}. 
\end{definition}

We now introduce the concept of a reality which can be defined as a maximal possible branch or, equivalently, a maximal independent set in the Conflict Graph. In other words, a reality aggregates the maximal number of conflicts while preserving non-conflicting nature.
\begin{definition}[Maximal branch and reality]\label{def: reality}
A branch ${B}\in\branchset$ is \emph{maximal} if there exists no other branch $A\in\branchset$ such that  ${B}\subset A$. 
A maximal branch is called a \textit{reality}.
\end{definition}
Next, we describe the notion of the maximal contained branch of a given transaction which  consists of the set of conflicting transactions in the past cone of the given transaction.
\begin{definition}[Maximal contained branch] \label{def: maximal contained branch}
Let $\branchset$ be the set of all branches, and $\branch_{\ledger}^{(p)}: \ledger \to \branchset$ be a function that for a given transaction $\tx{x}\in \ledger$ returns the maximal branch contained in $\cone{p}{\ledger}{\tx{x}}$. 
\end{definition}

We note that there could not be two maximal branches in the ledger past cone of a transaction. Indeed, the past cone of any transaction is conflict-free and, thus, if there would be two maximal branches, we could consider the union of two branches, which has to be also a branch. 
\begin{definition}[Ledger of a reality]\label{def: ledger of a reality}
Let $R\in \branchset$ be a reality. Define the $R$\textit{-ledger}, written as $\ledger(R)$, to be the set of all transactions $\tx{x}\in \ledger$ such that $\branch_{\ledger}^{(p)}(\tx{x})\subseteq R$.
\end{definition}
Recall that a maximal contained branch of a transaction from the $R$-ledger is a subset of $R$. Thus, the past cones of any two transactions are conflict-free and so is the $R$-ledger.

\begin{remark}[Local Ledger]
As discussed in Section~\ref{sec: local tangles}, there could be subjective versions of the Tangle DAG. Similarly, every node has its own perception of the Ledger. Thereby, we will use subscripts $i,t$ in  $\ledger$, $\LedgerDAG$ and other related notions if we talk about the point of view of node $i\in\nodeset$ at moment $t$.
\end{remark}

\subsection{Reality selection algorithm}\label{sec: reality selection algorithm}

To issue new blocks and validate transactions, each node in the network has to choose a conflict-free part of the ledger that it prefers. For this purpose, it suffices for a node to choose a preferred reality. Once a reality $R$ is chosen, the node can make different operations on the $R$-ledger.

\begin{definition}[Preferred reality]\label{def: preferred reality}
Node $i\in\nodeset$ at time $t$ chooses a specific reality $R=R_{i,t}\in\branchset$ which is called the \textit{preferred reality} for node $i$.
\end{definition}

There could be different ways to choose the preferred reality. We provide a natural reality selection algorithm that takes as an input the Conflict Graph and an abstract weight function $\mathbf{w}: \conflictset\to [0,1]$ satisfying two
properties:
\begin{enumerate}[leftmargin=*]
\item monotonicity: for any two conflicts  $x, y\in \conflictset$ such that $x \le_{\conflictset} y$, it holds that 
    $$
    \mathbf{w} (x) \le \mathbf{w}(y);
    $$
\item consistency: let $x_1,\ldots, x_s$ be pairwise conflicting conflicts.\footnote{We say that transactions  $S\subseteq \ledger$ are pairwise conflicting if any pair of transactions $x, y\in S$ are conflicting. } Then it holds that 
$$
\sum_{i=1}^{s}\mathbf{w}(x_i) \le 1.
$$
\end{enumerate}

\begin{remark}\label{rem: reality selection algorithm}
In Section~\ref{sec: Approval Weight}, we introduce the Approval Weight function defined on the set of all transactions, i.e. $\AW: \ledger\to[0,1]$. Then the required weight function can be obtained as the restriction of the Approval Weight to the set of conflicts, i.e. $\mathbf{w}=\AW|_{\conflictset}$.
\end{remark}
In Algorithm~\ref{alg:selectionConflictGraph} we describe the proposed procedure. In this algorithm, we initialize $R$ as  the genesis and $U$ as the set of conflicts. Then we iteratively construct a subset $R$ of conflicts and prune transactions conflicting with $R$ from $U$. Specifically, we add a conflict to $R$ if this conflict is not conflicting with this set and attains the highest value of the weight function among all  $\ConflictDAG$-maximal elements that remain in $U$. By construction, Algorithm~\ref{alg:selectionConflictGraph}  leads to a maximal independent set in the Conflict Graph or a reality. The number of iterations in the while-loop is bounded by $|\conflictset|$ and the number of $\ConflictGraph$-neighbours is also bounded by $|\conflictset|$. Thus, it is possible to implement this algorithm with complexity $O(|\conflictset|^2)$. 
\begin{algorithm}[t]
\caption{Reality selection in Conflict Graph}
\label{alg:selectionConflictGraph}
\KwData{Conflict Graph $\ConflictGraph=(\conflictset,E)$}
\KwResult{reality ${R}\in \branchset$}
${R}\gets \emptyset$\\
$U\gets \conflictset$\\
\While{$|U|\neq 0$}{
$c^*\gets \argmax \{ \mathbf{w}(c): c\in \max_{\conflictset}(U)\}$ {\scriptsize{\Comment*[r]{use $\min\hash(c)$ for breaking ties}}}
${R} \gets {R} \cup \{c^*\}$\\ $U \gets U  \setminus \{N_{\conflictset}(c^{*}) \cup\{c^{*}\}\}$
}

\end{algorithm}

We refer to Appendix~\ref{sec: toy example}, where we apply the algorithm as part of an illustrated example.

\section{On Tangle Voting}\label{sec: Voting}

In this section, we present a voting mechanism based on the Tangle and the Ledger DAG. This mechanism allows for selecting realities in the Reality-based Ledger.

In  Section~\ref{sec: tamper proof DAGs} we give an overview of two suitable DAG structures, which can be utilised to enable voting on the realities. Section~\ref{sec: voting and voting DAG} combines these two structures into a Voting DAG and introduces basic concepts that follow from it. We also address how voting on two DAGs increases the liveness of the protocol. 
Section~\ref{sec: Approval Weight} defines a metric called Approval Weight which is utilised in Section~\ref{sec: tip selection algorithms} to  identify a preferred reality and vote for it using a suitable tip selection algorithm.

\subsection{Extension of Witness Weight and Liveness Problems}\label{sec: liveness of transactions}

In Section~\ref{sec: tangle} we introduced the Witness Weight, which is a metric used for the confirmation of blocks. In this section, we seek a similar tool for the confirmation of transactions. 

The Witness Weight has the property that it is monotonically increasing since it expresses the percentage of the weight that has witnessed a block's existence. The situation is different for transactions where we want to leverage the node's weight to decide between conflicting transactions. To ensure liveness, nodes must have the possibility to change their votes and withdraw their weights from the approval weight of a given transaction.\footnote{In contrast, if weights are added but  not withdrawn, it is possible that two conflicting transactions gain precisely the same weight which would result in an impasse.} 
However, changing the opinions might imply that blocks that reference (and vote for) blocks with rejected transactions might never be confirmed.

This situation creates a negative incentive to reference new tips. More precisely, nodes may be incentivized to either reference only blocks from trusted entities, tips of a certain age, or in the worst case, ancient and already confirmed blocks. The last behaviour may eventually lead to no new blocks being confirmed anymore.

The problems above were until now a significant concern of DAG-based consensus protocols, e.g. \cite{popov2015}. We propose to solve these by using the Reality-based Ledger and extending the reference scheme.

\subsection{Immutable  DAGs}\label{sec: tamper proof DAGs}

Blocks are the primary information carriers of the network, i.e. they contain transactions and express the opinion of the issuing nodes. The references in the blocks, together with the signature of the nodes and the unlock proofs for the inputs, form two immutable data structures, similar to a blockchain.

First, the Tangle  $\TangleDAG$ is constructed on the set of blocks $\blockset$. The interrelations are defined by the references contained in the blocks, which are selected and signed by the issuing nodes (for more details, see Section~\ref{sec: tangle}).

Second, the Ledger  DAG $\LedgerDAG$ is constructed on the set of transactions $\ledger$. Their interrelations  are defined by the consumption of inputs, which are the outputs of previous transactions. The consumption and creation of outputs are cryptographically verified by the signature of the fund owner (for more details, see Section~\ref{sec: ledger}).

For nodes to objectively agree on a partial order of events, we require the following assumption. 

\begin{ass}[Past cone completeness]\label{ass: pastcone completeness}
For a transaction $\tx{x}$ that spends an output created in a transaction $\tx{y}$, it holds that the block $x$ is contained in the Tangle future cone of $y$, i.e.  $x\in\cone{f}{\blockset}{y}$. 
\end{ass}

In other words, we have the natural assumption that the spending of the output should happen in the future cones the blocks ``creating'' these outputs.

\begin{lemma}\label{lem: UTXO Tangle subgraph}
Under Assumption~\ref{ass: pastcone completeness},  the partial order $\le_{\ledger}$ induced by $\LedgerDAG$ is consistent with the partial order $\le_{\blockset}$ induced by $\TangleDAG$. More specifically, if for some blocks $x,y\in\blockset$, we have that the corresponding transactions satisfy $\tx{x}\le_{\ledger} \tx{y}$, then it holds that $x\le_{\blockset} y$.
\end{lemma}

\begin{proof}
The statement can be shown trivially by induction on the length of the shortest path between $\tx{x}$ and $\tx{y}$ in $\LedgerDAG$. The base case, when the length of the path is one, is implied by Assumption~\ref{ass: pastcone completeness}. 
\end{proof}

\subsection{Voting and Voting DAG}\label{sec: voting and voting DAG}

As a consequence of Lemma~\ref{lem: UTXO Tangle subgraph} both, the Tangle and the  Ledger DAG, are suitable for nodes to express their opinions about which transactions they prefer among any conflicting transactions. 
More specifically by creating and attaching new blocks, nodes have an implicit way of voting for the ``preferred'' branches and conflicts. Let us define this more precisely.

We utilise the references contained in a block, which constitute the edges of the Tangle, see Section~\ref{sec: tangle}, to express a node's opinion. As by Definition~\ref{def: block} a reference  contains two fields: $r_x$, which is a reference to block $x$ and $v$, which is the value of a label. We call the label $v$ the \textit{vote type} that can take values in $\{v_\mathcal{T}, v_\mathcal{L}\}$. This label gives additional meaning to the reference to $x$ in the Tangle and defines the following two specialised references.

\begin{definition}[Block reference]\label{def: Block reference}
We say a reference $\reference(y)=(r_x,v)$ from a block $y$ to a block $x$ is a \textit{block reference} if $y$ references $x$. In this case, we set the label $v=v_\mathcal{T}$. 
\end{definition}

To overcome the liveness issues described in Section~\ref{sec: liveness of transactions} we additionally add a reference that bypasses the block and directly addresses the contained transaction.

\begin{definition}[Transaction reference]\label{def: Transaction reference}
We say a reference $\reference(y)=(r_x,v)$ from a block $y$ to a block $x$ is a \textit{transaction reference}  if $y$ references $\tx{x}$. In this case, we set the label $v=v_\mathcal{L}$. 
\end{definition}

\begin{remark}
Naturally, a block references the transaction that is the content of the block. As such, an honest node would not issue a block with a transaction that is not in its preferred reality (see Section~\ref{sec: tip selection algorithms}).
\end{remark}

\begin{example}
Consider Figure~\ref{fig: inheritance}. Blocks $y$ and $y'$ contain the same transaction $\tx{y}$, but $y$ refers to the transaction $\tx{x}$ in block $x$ and, thus, issues a transaction reference, while block $y'$ refers to the block $x$ and, thus, issues a block reference, instead.
\begin{figure}
    \centering
    \definecolor{superlightgray}{RGB}{224, 224, 224}
\definecolor{mygray}{RGB}{240, 240, 240}
\definecolor{mygreen}{RGB}{0, 204, 153}
\definecolor{myred}{RGB}{255, 153,153}
\definecolor{lightsteelblue}{RGB}{176,196,222}
\definecolor{navajo}{RGB}{255,222,173}
\begin{tikzpicture}[scale = 0.92]
\def\widthBlock{1.5}
\def\heightBlock{0.8}
\def\shiftdown{3*\heightBlock}
\def\arrowStyle{-latex}

\tikzstyle{rounded_block}=[draw, rectangle, thick, minimum height=\heightBlock cm, minimum width = \widthBlock cm, text centered, rounded corners, draw=darkgray, font = \scriptsize]
\node [rounded_block, fill=white, minimum height= 1.7 cm, minimum width = 4.1 cm] at (0,0) {} ;

\node [rounded_block, fill=white, minimum height= 1.4 cm, minimum width = 1.8 cm, fill = superlightgray] at (-1,0) {} ;

\node [rounded_block, fill=white, minimum height= 0.8 cm, minimum width = 1.2 cm, fill = lightsteelblue] (blue_block) at (-1,-0.2) { Branch $B_{\tx{x}}$} ;

\node [rounded_block, fill=white, minimum height= 0.8 cm, minimum width = 1.2 cm, fill = myred] (red_block) at (1,-0.2) { Branch $B_{{x}}$} ;

\node[font = \footnotesize] at (-1,0.44) {Transaction $\tx{x}$};

\node[font = \footnotesize] at (1,0.44) {Block ${x}$};

\node [rounded_block, fill=white, minimum height= 1.7 cm, minimum width = 4.1 cm] (block_bottom_left) at (-2.5,-3) {} ;

\node [rounded_block, minimum height= 1.4 cm, minimum width = 1.8 cm, fill = superlightgray] at (-3.5,-3) {} ;

\node [rounded_block,  minimum height= 0.8 cm, minimum width = 1.2 cm, fill = navajo] at (-3.5,-3.2) { Branch $B_{\tx{y}}$} ;

\node [rounded_block, fill=white, minimum height= 0.8 cm, minimum width = 1.2 cm, fill = lightsteelblue] at (-1.5,-3.2) { Branch $B_{\tx{x}}$} ;

\node[font = \footnotesize] at (-3.5,-2.56) {Transaction $\tx{y}$};

\node[font = \footnotesize] at (-1.5,-2.56) {Block ${y}$};

\node [rounded_block, fill=white, minimum height= 1.7 cm, minimum width = 4.1 cm] (block_bottom_right) at (2.5,-3) {} ;

\node [rounded_block, minimum height= 1.4 cm, minimum width = 1.8 cm, fill = superlightgray] at (1.5,-3) {} ;

\node [rounded_block,  minimum height= 0.8 cm, minimum width = 1.2 cm, fill = navajo] at (1.5,-3.2) { Branch $B_{\tx{y}}$} ;

\node [rounded_block, fill=white, minimum height= 0.8 cm, minimum width = 1.2 cm, fill = myred] at (3.5,-3.2) { Branch $B_{\tx{x}}$} ;

\node[font = \footnotesize] at (1.5,-2.56) {Transaction $\tx{y}$};

\node[font = \footnotesize] at (3.5,-2.56) {Block ${y'}$};

\draw[\arrowStyle, dashed]  (block_bottom_left) -- (blue_block);

\draw[\arrowStyle]  (block_bottom_right) -- (red_block);

\node[font = \footnotesize] at (-2.6, -1.5){\begin{tabular}{r}
Transaction \\
reference
\end{tabular}};

\node[font = \footnotesize] at (2.6, -1.5){\begin{tabular}{l}
Block \\
reference
\end{tabular}};

\end{tikzpicture}

    \caption{Inheritance of branches: we consider two potential blocks $y$ and $y'$ that contain the same transaction $\tx{y}$, but have either a transaction, or a message reference to  block $x$. Thus, a node can vote in two  ways. Specifically, block can approve a previous block via a transaction reference or a block reference, and inherit the branch of the referenced transaction or the referenced block, respectively.}
    \label{fig: inheritance}
\end{figure}
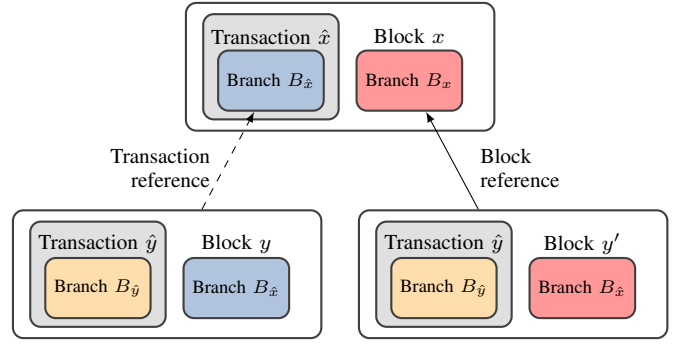
\end{example}

\begin{remark}
The distinction into the sub-categories (transaction reference and block reference) is only relevant for the purpose of voting; the definition of the Witness Weight, see Section \ref{sec: witness weight}, remains unaffected.
\end{remark}

We define a data structure that combines the two immutable data structures in Section~\ref{sec: tamper proof DAGs} into one single DAG  used for propagating the votes.

\begin{definition}[Voting DAG]
The Voting DAG $\VotingDAG$ is a DAG whose vertex set $\votingset$ is the union of the set of blocks $\blockset$ and the set of transactions $\ledger$, i.e.   $\votingset=\blockset\cup\ledger$.
Let $v$ and $u$ be two vertices in $\votingset$. There exists a directed edge from  $u$ to $v$ in $\VotingDAG$ if and only if one of the following properties holds:
\begin{enumerate}[leftmargin=*]
    \item $u,v\in \blockset$ and  $u$ contains a block reference to $v$;
    \item  $u\in\blockset, v\in\ledger$ and $u$ contains a transaction reference to transaction $v$;
    \item $u\in\blockset$ and $v=\tx{u}\in\ledger$, i.e.  $v$ is a transaction in block $u$;
    \item $u,v\in\ledger$ and  transaction $u$ spends the output from transaction $v$, i.e.  $v\in\parent{\ledger}{u}$.
\end{enumerate}
\end{definition}

So far we described how references between blocks are given additional meaning to construct the voting DAG. This DAG allows nodes to express their opinions, recursively. Following Definition~\ref{def: future and past cones} we define $\cone{p}{\votingset}{x}$ as the \textit{voting past cone} of block or transaction $x$ in the Voting DAG.

\begin{definition}[Voting]\label{def: vote}
A node $i$ expresses a \emph{direct vote} for a vertex $x\in \votingset$ in the voting DAG $\VotingDAG$ by referencing $x$ in a block $y\in\blockset$,  where $\iss(y)=i$. 
We say node $i$ \textit{indirectly votes} for any vertex in  $\cone{p}{\votingset}{x}$.
\end{definition}

\begin{example}
We illustrate the concept of a Voting DAG in Figure~\ref{fig: voting DAG}. The Voting DAG assembles information from the Tangle and the Ledger DAG. We assume a situation where the node that issues block $x$ does not approve transaction $\tx{y}$ and, thus, can vote neither for blocks $y$ nor $z$. However, it can vote for transaction $\tx{z}$ by using a transaction vote. More precisely, by creating a transaction reference to block $z$, the vote of block $x$ avoids vertices $z,y,\tx{y}$ shown in grey.
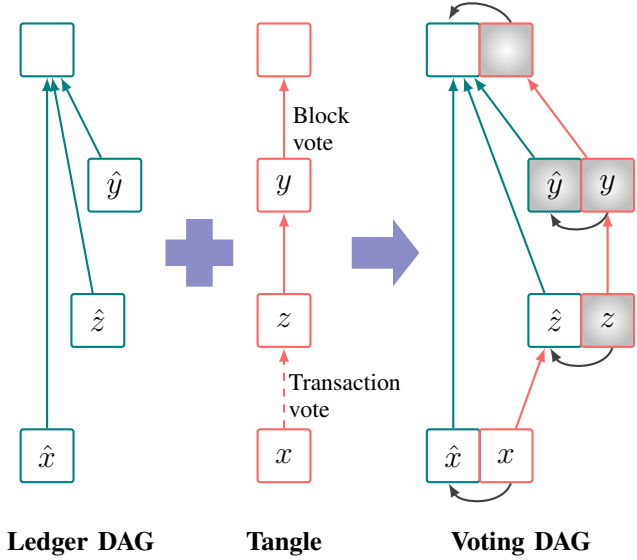
\begin{figure}
    \centering
    \definecolor{megalightgray}{RGB}{244, 244, 244}
\definecolor{mybluegray}{RGB}{140, 140, 200}
\definecolor{myblue}{RGB}{102, 178, 255}
\definecolor{myred}{RGB}{255, 102, 102}
\tikzstyle{rounded_block}=[draw, rectangle, thick, minimum height=\heightBlock cm, minimum width = \widthBlock cm, text centered, rounded corners, draw=darkgray, font = \small]

\begin{tikzpicture}[use Hobby shortcut, scale = 0.9]
\def\xCoordinate{0.0}
\def\yCoordinate{0.0}
\def\yAdd{1.5}
\def\xAdd{0.75}
\def\innerSepar{-1.5}
\def\lineWidth{1.5}
\def\roundedCorners{1}
\def\opacityInternal{0.5}
\def\minHeight{20}
\def\fracyAdd{1/6}
\def\lineWidthBelow{3}
\def\minWidth{20}
\def\scaleFactorSupp{0.12}
\def\minWidthCM{\minWidth*0.0352778*3/4}
\tikzstyle{block}=[draw, rectangle, minimum height=\minHeight pt, minimum width = \minWidth pt, text centered, rounded corners=\roundedCorners pt, draw=darkgray, font=\large]
\tikzstyle{block_colored}=[draw, rectangle, minimum height= 2 pt, minimum width = \minWidth * 5 / 6 pt, rounded corners=\roundedCorners  pt, draw=darkgray]
\def\arrowStyle{-latex}


\node[block, thick, draw = teal] (node_1_0) at (0, 0) {}; 

\node[block, thick,draw = teal] (node_1_1) at (1, -2) {$\tx{y}$};

\draw[\arrowStyle, thick, draw = teal] (node_1_1) -- (node_1_0); 

\node[block, thick, draw = teal] (node_1_2) at (0.75, -4) {$\tx{z}$};

\draw[\arrowStyle, thick,draw = teal] (node_1_2) -- (node_1_0);

\node[block, thick,draw = teal] (node_1_3) at (0.0, -6) {$\tx{x}$};

\draw[\arrowStyle, thick, draw = teal] (node_1_3) -- (node_1_0);

\node[block, thick, draw = myred] (node_2_0) at (3.5, 0) {}; 

\node[block, thick,draw = myred] (node_2_1) at (3.5, -2) {${y}$};

\draw[\arrowStyle, thick, draw = myred] (node_2_1) -- (node_2_0); 

\node[block, thick, draw = myred] (node_2_2) at (3.5, -4) {${z}$};

\draw[\arrowStyle, thick,draw = myred] (node_2_2) -- (node_2_1);

\node[block, thick,draw = myred] (node_2_3) at (3.5, -6) {${x}$};

\draw[\arrowStyle, thick, dashed, draw = myred] (node_2_3) -- (node_2_2);

\node[font = \small] at (4.05, -1.2) {\begin{tabular}{l}
Block \\
vote
\end{tabular}
};

\node[font = \small] at (4.4, -5.15) {\begin{tabular}{l}
Transaction \\
vote
\end{tabular}
};

\node[block, thick, draw = teal] (node_3_0_0) at (6, 0) {}; 
\node[block, thick, draw = myred, even odd rule,outer color=lightgray,inner color=white] (node_3_0_1) at (6.78, 0) {}; 
\draw[\arrowStyle, thick, draw = darkgray] ([xshift = 0.2 em, yshift=0.0em] node_3_0_1.north) to [out=120,in=60] ([xshift = -0.2 em,yshift=0.1em ]node_3_0_0.north)
    ;

\node[block, thick,draw = teal, even odd rule,outer color=lightgray,inner color=white] (node_3_1_0) at (7.5, -2) {$\tx{y}$};
\node[block, thick,draw = myred, even odd rule,outer color=lightgray,inner color=white] (node_3_1_1) at (8.28, -2) {${y}$};
\draw[\arrowStyle, thick, draw = teal] (node_3_1_0) -- (node_3_0_0);
\draw[\arrowStyle, thick, draw = myred] (node_3_1_1) -- (node_3_0_1);
\draw[\arrowStyle, thick, draw = darkgray] ([xshift=0.0em, yshift=-0.0em]node_3_1_1.south) to [out=250,in=300] ([xshift=-0.2em,yshift=-0.0em]node_3_1_0.south);

\node[block, thick,draw = teal] (node_3_2_0) at (7.5, -4) {$\tx{z}$};
\node[block, thick,draw = myred, even odd rule,outer color=lightgray,inner color=white] (node_3_2_1) at (8.28, -4) {${z}$};
\draw[\arrowStyle, thick, draw = myred] (node_3_2_1) -- (node_3_1_1);
\draw[\arrowStyle, thick, draw = teal] (node_3_2_0) -- (node_3_0_0);
\draw[\arrowStyle, thick, draw = darkgray] ([xshift = 0.2em, yshift=-0.0em]node_3_2_1.south) to [out=250,in=300] ([xshift = -0.2em, yshift=-0.0em]node_3_2_0.south);

\node[block, thick,draw = teal] (node_3_3_0) at (6, -6) {$\tx{x}$};
\node[block, thick,draw = myred] (node_3_3_1) at (6.78, -6) {${x}$};

\draw[\arrowStyle, thick, draw = teal] (node_3_3_0) -- (node_3_0_0);
\draw[\arrowStyle, thick, draw = myred] (node_3_3_1) -- (node_3_2_0);
\draw[\arrowStyle, thick, draw = darkgray] ([xshift = 0.2em, yshift=-0.0em]node_3_3_1.south) to [out=250,in=300] ([xshift = -0.2em, yshift=-0.0em]node_3_3_0.south);

\draw[mybluegray,line width=1em] (1.75,-3) -- (2.75,-3)(2.25,-2.5) -- (2.25,-3.5);
\draw[mybluegray, fill=green, -{Triangle[width = 25pt, length = 10pt]}, line width = 12pt] (4.5, -3) -- (5.5, -3);

\node at (0.5,-7.3) {\textbf{Ledger DAG}
};
\node at (3.5,-7.3) {\textbf{Tangle}
};
\node at (7,-7.3) {\textbf{Voting DAG}
};
 
\end{tikzpicture}
    \caption{Illustration of how the Voting DAG is assembled from the Tangle and the Ledger DAG. By creating a transaction reference to block $z$, the vote of block $x$ avoids vertices $z,y,\tx{y}$ shown in grey.}
    \label{fig: voting DAG}
\end{figure}
\end{example}

We can also describe the voting past cone in terms of a recursive equation.

\begin{proposition}\label{prop: recursive tructure of voting}
Suppose a given block $x\in\blockset$ has block references to block $y_1,\ldots,y_s$ and transaction references to blocks $z_1,\ldots,z_r$ with $s+r=k$. 
Then the voting past cone of $x$ can be written in a recursive way
$$
\cone{p}{\votingset}{x}= x \cup   \cone{p}{\ledger}{\tx{x}}\cup C_{\ledger}(x) \cup C_{\votingset}(x),
$$
where 
\begin{align*}
C_{\ledger}(x)&:=\cone{p}{\ledger}{\tx{z}_1}\cup \ldots\cup \cone{p}{\ledger}{\tx{z}_r},\\
C_{\votingset}(x)&:=\cone{p}{\votingset}{y_1}\cup \ldots\cup \cone{p}{\votingset}{y_s}.\\
\end{align*}
\end{proposition}

The Reality-based Ledger introduces the concept of branches, see Section~\ref{sec: ledger}. The consumption of more than one output from different branches creates a new branch, which is the union of the branches of the consumed outputs. 
Now we extend this concept to blocks, which can combine branches by voting for previous blocks or transactions. 
More precisely we can relate a given reference in a block with a branch. The branch of the block is then defined as follows.

\begin{definition}[Voting Branch]\label{def: voting branch}
Given a block $x\in\blockset$, we define the \textit{voting branch} of $x$ to be 
$$
\branch^{(p)}_{\votingset}(x):= \cone{p}{\votingset}{x} \cap \conflictset,
$$
where $\conflictset$ is the set of conflicts.
\end{definition}

\begin{remark}
We highlight that for the correctness of the protocol, a node has to create references for a new block $x$ in such a way that $\branch^{(p)}_{\votingset}(x)$ is indeed a branch as defined in Definition~\ref{def:branch}. The property that $\branch^{(p)}_{\votingset}(x)$ is $\ConflictDAG$-past-closed trivially follows from the fact that $\cone{p}{\votingset}{x}\cap\ledger$ is $\LedgerDAG$-past-closed. However, the conflict-free property of $\branch^{(p)}_{\votingset}(x)$ is not necessarily true in general and has to be checked. We address this issue when we discuss tip selection algorithms in Section~\ref{sec: tip selection algorithms}.
\end{remark}

Recall Definition~\ref{def: maximal contained branch} that introduces the maximal contained branch  of a transaction $\tx{x}$, written as $\branch^{(p)}_{\ledger}(\tx{x})$. Using Proposition~\ref{prop: recursive tructure of voting}, we relate the voting branch of block $x$ and  the maximal contained branch of transaction $\tx{x}$ in the following statement.

\begin{proposition}[Inheritance of branches]\label{prop:inheritance of branches}
Suppose a given block $x\in\blockset$ has block references to block $y_1,\ldots,y_s$ and transaction references to blocks $z_1,\ldots,z_r$ with $s+r=k$. A block $x$ inherits the union of the branches that are associated with these votes, i.e.  the voting branch can be decomposed as follows 
$$
\branch^{(p)}_{\votingset}(x) =  \branch^{(p)}_{\ledger}(x) \cup B_{\ledger}(x) \cup B_{\votingset}(x),
$$
where
\begin{align*}
B_{\ledger}(x)&:=\branch^{(p)}_{\ledger}(\tx{z}_1)\cup \ldots\cup \branch^{(p)}_{\ledger}(\tx{z}_r),\\
B_{\votingset}(x)&:=\branch^{(p)}_{\votingset}(y_1)\cup \ldots\cup \branch^{(p)}_{\votingset}(y_s),\\
\end{align*}
\end{proposition}

\begin{example}
We follow the same example as shown in Figure~\ref{fig: inheritance}. We assume the maximal contained branch of the transaction in block $y$ is the main branch, i.e. $B_{\tx{y}}=\emptyset$. Block $y$ votes for the transaction contained in block $x$ and, thus, inherits the branch $B_{\tx{x}}$. Since $B_{\tx{y}}=\emptyset$  the voting branch of block $y$ is $B_{\tx{x}}$. Similarly, block $y'$ votes for the block itself and inherits the voting branch $B_{x}$. We highlight that the branch of the transaction contained in block $y$ (and $y'$) is not affected by the choice of the vote.
\end{example}

We can associate a given block $x$ with a branch $B_\blockset=\branch^{(p)}_{\votingset}(x)$. Furthermore, the content of $x$, which is a transaction $\tx{x}$, also can be associated with a branch $B_\ledger=\branch^{(p)}_{\votingset}(\tx{x})$. Due to Lemma~\ref{lem: UTXO Tangle subgraph} we have that $B_\ledger \subseteq B_\blockset$.
Since a node may change its opinion about a conflict and vote for a conflicting transaction to $\tx{x}$, the vote is only valid from the point-of-view of the referencing block $y$. A later change of the node's vote is possible by issuing another block that votes for a conflicting transaction to $\tx{x}$. 

\begin{definition}[Change of vote and current vote]\label{def: change and current vote}
Let $\tx{x}$ be a transaction for which node $i\in\nodeset$ voted for. Let transaction $\tx{y}$ be conflicting with $\tx{x}$. If node $i$ votes for $\tx{y}$ after it voted for $\tx{x}$, node $i$ is no longer approving $\tx{x}$. We say $i$ \textit{revokes} its vote from $\tx{x}$. 
If $i$'s most recent vote is approving $\tx{x}$, i.e.  the vote is also not revoked, we say $i$'s \textit{current vote} is approving $\tx{x}$.
\end{definition}

\begin{remark}
The notion of ``time'' and its implications on the meaning of ``after'' in Definition \ref{def: change and current vote} are crucial. Natural choices are the timestamp inside a transaction or the solidification time of a block that contains a given transaction. 
\end{remark}

\begin{example}
The principle of Definition~\ref{def: change and current vote} is demonstrated in Figure~\ref{fig:octopus_example}. Specifically, transactions $\tx{x}$ and $\tx{y}$ are conflicting; block references are depicted with solid edges, whereas transaction references are depicted with dashed edges. Initially, brown and purple nodes voted for $\tx{x}$. However, after a while, green nodes have revoked their votes from $\tx{x}$ and their latest votes are approving $\tx{y}$. The supporters of $\tx{x}$ and $\tx{y}$ are shown in the top-right corners of blocks  for each of the two periods.
\end{example}

\subsection{Approval Weight and Confirmation Rule for Transactions}\label{sec: Approval Weight}

Nodes must be able to track the progress of the acceptance of a transaction. We extend the concepts of Witness Weight, introduced in Section~\ref{sec: witness weight}, to the \textit{Approval Weight} (AW) of transactions. The  objective  is then to define a parameterisable confirmation condition for transactions similar to the one discussed for blocks in Section \ref{sec: WW confirmation rule}.

\begin{definition}[Transaction  Supporters and Approval Weight]\label{def: AW and supporters}
Let  $\tx{x}\in\ledger$ be a transaction. 
Denote by $\supporter_{\ledger}(\tx{x})$ the set of nodes that has a current vote for $\tx{x}$. These nodes are called \textit{supporters} of $\tx{x}$.
We define the function $\AW: \ledger \to [0,1]$ which is called the \emph{Approval Weight} (AW) \begin{align}\label{eq: AW transaction}
    \AW (\tx{x}) :=& \sum_{j \in \supporter_{\ledger}(\tx{x})} \weight(j)
\end{align}
\end{definition}
Clearly, the AW describes the percentage of the network approving a given transaction. 

\begin{remark}
The WW of a block $x$ and the AW of its contained transaction $\tx{x}$ are related, but there is no ``monotonicity property''. 
If $\overline x$ is only contained in the block $x$ we have that $\AW(\tx{x})\leq \WW(x)$. If $\overline x$ is contained in more than one block~\footnote{A fund owner may request several nodes to broadcast a transaction $\tx{x}$, or a node may issue several different blocks with $\tx{x}$ contained.}, we can have that the AW of the transaction is even larger than the WW of each enveloping block.
\end{remark}

The supporter of transactions can be updated using propagation of the supporter information through the voting DAG. More precisely on arrival of a block $x$ we traverse $\cone{p}{\votingset}{x}$. We propose Algorithm~\ref{alg: update supporters} to update transactions supporters when a new block is processed. The AWs are then updated using Equation~\eqref{eq: AW transaction}.

Similar to Definition~\ref{def: confirmed block} we define the confirmation of a transaction. We will use subscripts $i$ and $t$ such as $\AW_{i,t}$ if we talk about the perception of the $i$th node at moment $t$.

\begin{definition} [Confirmed transaction]\label{def: confirmed transaction}
Let $\theta\in(0.5,1]$ be a fixed threshold. We say that a transaction $\tx{x}\in\ledger$ is \emph{confirmed} for a node $i\in\nodeset$ at time $t$ if  $\AW_{{i,s}} (\tx{x})\ge \theta$, for some $s\leq  t$.
\end{definition}

We also define the AW of a branch, which will form the base for the algorithm in the next section. The supporters of a branch are equal to the intersection of the supporters of the conflicts in the branch. More formally we have the following.
\begin{definition}[Branch Supporters and Approval Weight]\label{def: AW of branch}
Let $ B \in \branchset$ be a branch. 
We define $\supporterAW_{i,t}(B)$ to be the set of nodes that issued blocks that approve all conflicts in $B$. Similarly we define the AW for  $B$ as
\begin{align}\label{eq: AW Branch}
\AW  (B) :=& \sum_{j \in \supporterAW (B)} \weight(j).
\end{align}
We define the AW of the main branch, i.e. the empty set $\emptyset$, to be $1$.
\end{definition}

\begin{algorithm}[t]
\caption{Updating transaction supporters when new block arrives}
\label{alg: update supporters}
\KwData{Tangle DAG $\TangleDAG$,  Ledger DAG $\LedgerDAG$, new block $x$ issued by node $j$, $\{\supporter_{\ledger}(y)\}_{y\in\ledger}$}
\KwResult{updated $\{\supporter_{\ledger}(y)\}_{y\in\ledger}$}
\For{$\forall z\in\cone{p}{\votingset}{x} \cap \ledger$}{
$\supporter_{\ledger}(z)\gets\supporter_{\ledger}(z)\cup\{j\}$
}
\For{$\forall z\in\ledger$ that are
conflicting with $\cone{p}{\votingset}{x}$}{
$\supporter_{\ledger}(z)\gets\supporter_{\ledger}(z)\setminus\{j\}$
}
\end{algorithm}

\subsection{Tip Selection Algorithm}\label{sec: tip selection algorithms}

The consensus protocol relies substantially on an implicit voting mechanism. Nodes express their opinions and votes by choosing the references in their newly issued blocks. The process that determines the references is called the \textit{Tip Selection Algorithm} (TSA) and is discussed in this section.

With every block, a node can vote on  which parts of the Tangle and the Ledger DAG it prefers by  using block or transaction references.  
The preferred parts of the  Tangle and the Ledger DAG are defined by the preferred reality. 
Following the algorithm described in Section~\ref{sec: reality selection algorithm}, a node $i$ at moment $t$ keeps its preferred reality $R=R_{i,t}$ up to date.

We now describe a tip selection mechanism that considers both block and transaction votes. 
Note that due to Lemma~\ref{lem: UTXO Tangle subgraph} the Ledger DAG induces a partial order consistent with the one induced by the Tangle and, thus, voting on the Ledger DAG allows expressing a more selective, albeit less efficient vote than on the Tangle.

Let us define some reality-dependent tip sets on the Tangle DAG and the Ledger DAG.

Denote by  $\tips_{{\blockset}}(R)\subset \mathcal{T}$  the tips in the Tangle DAG whose Tangle past cones contain  only transactions in reality $R$. More precisely, for any $x\in \tips_{{\blockset}}(R)$, there is no $y\in\cone{p}{\blockset}{x}$  such that $\tx{y}\in\conflictset\setminus R$.

Denote by  $\tips_{{\ledger}}(R) \subset \ledger$ the tips in the Ledger DAG whose past cones contain conflicts from $R$ only. In other words, it holds that for any $\tx{x}\in \tips_{{\ledger}}(R)$, there is no $\tx{y}\in\cone{p}{\ledger}{\tx{x}}$ such that $\tx{y}\in\conflictset\setminus R$.
A node should apply the following tip selection and reference setting.

\begin{definition}[Uniform Random Tip Selection on a reality]
To issue a new block, node $i$ chooses tips  to approve uniformly at random from all tips in the Tangle DAG until $k$ references are created. For a randomly chosen tip, the node proceeds with the following steps:
\begin{enumerate}
\item if the selected block is in the set $\tips_{\blockset}(R)$, a block reference is created; 
\item otherwise, if the selected tip contains a transaction that is in the set $\tips_{\ledger}(R)$, a transaction reference is created; 
\item if neither of the above apply, the block is discarded instead. 
\end{enumerate}
We call this algorithm the \textit{Uniform Random Tip Selection Algorithm restricted on the reality $R$} (or $R$-URTS for short) and refer to Algorithm~\ref{alg: URTS} for a pseudo code.
\end{definition}
We refer to Appendix~\ref{sec: toy example}, where we demonstrate this algorithm as part of an illustrated example. 

A node may have voted previously for a branch that is no longer its preferred branch. It has therefore to change its vote. With the above tip selection nodes are allowed to vote for branches they previously did not ``prefer'' (by voting for a conflicting transaction) and vote ``against'' branches they previously voted for. Every node must therefore keep the supporters for each branch and their AW up to date. An important consequence is that the AW of certain branches may increase in time while for others it may decrease in time. 

\begin{algorithm}[t]
\caption{Uniform Random Tip Selection restricted on reality $R$}
\label{alg: URTS}
\KwData{Tangle DAG $\TangleDAG$,  Ledger DAG $\LedgerDAG$, preferred reality $R\in\branchset$,  number of references $k$} \KwResult{tips  $L_{\blockset} \cup L_{\ledger}$
}
$L_{\blockset}\gets \emptyset$\\
$L_{\ledger}\gets \emptyset$\\
$cnt\gets 0$\\
\While{$cnt<k$}{
Choose tip $x$ uniformly at random in $\TangleDAG$\\
Set $Q_{\mathV}$ to be conflicts contained in {$\cone{p}{\mathV}{x}$ }\\ 
Set $Q_{\ledger}$ to be conflicts in $\cone{p}{\ledger}{\tx{x}}$\\
\eIf{$Q_{\mathV}\subseteq R$}{
$cnt\gets cnt+1$\\
$L_{\blockset}\gets L_{\blockset}\cup \{x\}$
}{
\If{$Q_{\ledger}\subseteq R$}{
$cnt\gets cnt+1$\\
$L_{\ledger}\gets L_{\ledger}\cup \{x\}$
}
}
}
\end{algorithm}
The addition of the transaction vote demonstrates that solutions for the Tip Selection Algorithm can be found that mitigate or reduce liveness issues  and that transactions eventually will be considered for tip selection. Thus, in the following, we work under the following assumption.

\begin{ass}[Block inclusion]\label{ass: TSA}
Let $R$ be the preferred reality. The tip selection satisfies that for every transaction $\tx{x}\in\ledger(R)$ (see Definition~\ref{def: ledger of a reality})  we have that there is at least one element in $\cone{f}{\votingset}{x}$ that the tip selection algorithm can pick up. In particular, there exists at least one available tip in the union $\tips_{\blockset}(R)\cup \tips_{\ledger}(R).$
\end{ass}

We also refer to Section \ref{sec: outlook} for a more detailed discussion.

\section{Communication and Adversary Models}\label{sec: comm Adversary Model}

Before stating the security requirements of the protocol, we have to make assumptions about the underlying communication model. It is common to describe the uncertainty related to the communication by an attacker that controls the delays of the blocks.  The communication model defines the limits the adversary can delay the communication between the nodes.
As a model, it is only a simplification, but it allows a systematic study of the most critical components. 

For simplicity, we also analyse the voting mechanism without details such as the TSA. We want to emphasise that our modelling can also be applied to other consensus protocols, thus, providing a framework for comparing different DLTs.

\subsection{Communication Model}\label{sec:communicationModel}

The participating nodes communicate over a peer-to-peer (P2P) protocol or network. In this  P2P protocol, nodes send their signed blocks to their neighbouring peers. Neighbours forward blocks from other nodes in the overlay network only if they have verified its validity; if a transaction is invalid, the propagation stops. The transmission of a block between two nodes is done by sending a \emph{package} containing the block.

There are three basic (or classic) models for the P2P communication between the nodes: the synchronous model, the asynchronous model, and the partial synchronous model, e.g.  see~\cite{DLL:88} and~\cite{ChTo:96}.

In the \emph{synchronous model}, there exists some known finite time bound $\Delta$ by which an adversary can delay the delivery of a package. 
In the asynchronous model, an adversary can delay the delivery of a package by an unknown finite amount of time. There is no bound on the time to deliver a block but each package must eventually be delivered. In the partial synchronous model, we assume that there is some finite unknown upper bound $\Delta$ on block delivery. This bound is not known in advance and can be chosen by the adversary. 

A partially synchronous system can  be seen as  initially asynchronous that becomes eventually synchronous. The time at which the system becomes synchronous is called the Global Stabilisation Time (GST). 

We also consider a probabilistic synchronous model,  see~\cite{muller2021}. In this  model we assume that for every $\varepsilon > 0$ and $\delta \in [0,1]$, a proportion $\delta$ of the blocks is delivered within a bounded (and known) time $\Delta=\Delta(\varepsilon, \delta)$, that depends on $\varepsilon$ and $\delta$, with probability of at least $1-\varepsilon$. The probabilistic synchronous model is similar  to the asynchronous model with crash failure faults, see~\cite{ChTo:96}.

The specific implementations for a consensus mechanism depend heavily on the underlying synchronicity assumption. 
It also seems appropriate to distinguish between consensus protocols that find consensus on one data set and consensus protocols that find consensus on a growing number of decisions. The latter allows to ``strengthen the synchronicity'' between the nodes if the data are related by references.

\subsection{The Tangle, Solidification, and Synchronicity}\label{sec:solidicationSynchronicity}
The references that form the Tangle are essential for the consistency of information every node has. Consider that a package propagates to only part of the network, e.g. lost during some of the propagation processes on the communication layer. However, nodes that have received the block start building on it and gossip their blocks to the network. These new blocks contain references  to the partially missing block. Since nodes must know the past cone of any block to have a complete Tangle history from that blocks' point of view, we use a mechanism called the \textit{solidification process}. In this mechanism, nodes that receive a given block only process it if its past cone is complete or, otherwise, ask their peers  for the missing referenced block (for more details, see Section~\ref{sec: local tangles}).
In other words, the solidification process is a mechanism to recover lost blocks and, hence, strengthens the ``synchronicity'' of the communication model. We think that this, to some extent, supports the assumption that all blocks are delivered within a bound time $\Delta$ with high probability.

\subsection{Adversary Model}

We distinguish between three types of nodes: \textit{honest}, \textit{faulty}, and \textit{malicious}. Honest nodes follow the protocol, faulty nodes are not working properly (e.g. not sending any transactions), and malicious nodes are trying to disturb the protocol by not following the rules actively. In most scenarios, we assume that the malicious nodes are controlled by an abstract entity that we call the \textit{attacker}. We assume that the attackers are computationally limited and cannot break the signature schemes or the cryptographic hash functions involved.  However, we assume that the attacker is omniscient and ``knows immediately'' about all state changes of the honest nodes.

In classic consensus protocols, the communication model already covers the adversary behaviours, as delaying blocks is essentially the only way an attacker can influence the system. This is no longer true for our consensus protocol. Here, adversarial strategies can be divided into two main categories: attacks on the protocol level and attacks on the voting layer.

\subsection{Configuration Graph and Schedule}
How events in distributed systems are triggered depends on some external causes that are often referred to as the environment. We follow~\cite{BBBG:15} and model this environment using the abstraction of a scheduler.  

To this end, we consider a communication network on which all communications between the nodes are carried out. These networks are often referred to as P2P networks. We model them using a directed graph whose vertex set is the set of participating nodes. There is a directed edge from $i$ to $j$, if node $i$ can send packages directly to node $j$.
\begin{figure}[t]
    \centering
\begin{tikzpicture}
\begin{scope}[every node/.style={circle,very thick, draw=teal}]
    \node (1) at (0,0.5) {$1$};
    \node (2) at (0,3) {$2$};
    \node (3) at (2.5,4) {$3$};
    \node (4) at (2.5,1) {$4$};
    \node (5) at (2.5,-2) {$5$};
    \node (6) at (5,3) {$6$} ;
\end{scope}

\begin{scope}[>={Stealth[black]},
              every node/.style={fill=white,circle},
              every edge/.style={draw=darkgray,thick}]
    \path [->] (1) edge node {$x$} (2);
    \path [->] (2) edge node {$z$} (3);
    \path [->] (1) edge node {$x$} (4);
    \path [->] (4) edge node {$y$} (3);
    \path [->] (1) edge node {$x$} (5);
    \path [->] (4) edge node {$y$} (5);
    \path [->] (4) edge node {$y$} (6);
    \path [->] (3) edge node {$w$} (6);
    \path [->] (5) edge node {$z$} (6); 
    \path [->] (6) edge[bend left=30] node {$u$} (5); 
    \path [->] (2) edge[bend right=60] node {$z$} (5); 
\end{scope}
\end{tikzpicture}
\caption{Illustration of a network of 6 nodes. Packages $u,w,x,y,z$ are send along directed edges representing communication channels.}
\label{fig:network}
\end{figure}
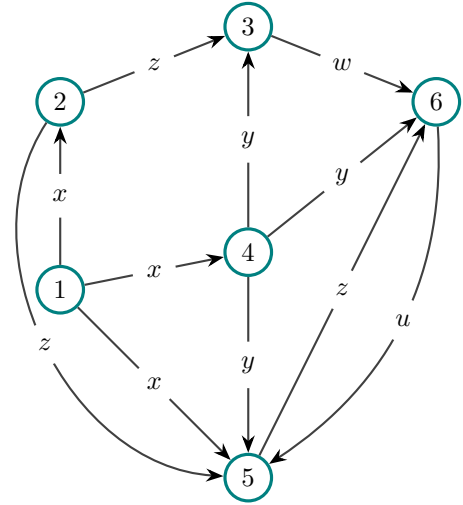

We assume this graph to be connected. Along the directed edges of this graph \textit{packages} are exchanged by the nodes. In our case these packages contain blocks.

\begin{definition}[Packages and communication graph]\label{def:packages}
For each block $x$ sent from node $i$ to node $j$ we add a directed edge, called \emph{package}, from $i$ to $j$ that is labelled by the vector
\begin{equation}\label{eq:edgeCommunication}
e(x,i,j):=(x, i,j, t(x), \delta_{i,j}(x)).    
\end{equation}
This label indicates that package $x$ was sent at time $t(x)$ from node $i$ and arrives at node $j$ at time $t(x)+ \delta_{i,j}(x)$. The state space of all packages is denoted by $\M$. We dub the resulting graph the \emph{communication graph} $G$ of the protocol.
\end{definition}

Essentially a node $i$ does the following once it receives a package $e$ from node $j$. It checks if the block $x$ that is contained in the package $e$ was already treated. If this is the case, the node's status remains unchanged, and no new package is issued. If the block is new, the node checks its validity and adds the block to its local Tangle. If applicable, it updates the supporters of branches and conflicts and, if the transaction contained in the block is conflicting, adds a new conflict to the set of conflict. After this step, the node forwards the block in new packages to all its neighbours from which the node has not received the block. 

A node may also create blocks. Once it creates a block $x$, it attaches $x$ to its local Tangle. Then, it creates a package $e$ containing block $x$ and sends copies of this package to all of its neighbours.\footnote{In this example, we use a flooding protocol for the dissemination of blocks.}

\begin{example}
We illustrate the concept of networks and packages in Figure~\ref{fig:network}. In this figure, the network consists of six nodes. Directed edges exist between some of them and show the communication channels. We point out that these communication does not necessarily have to be symmetric. Packages containing blocks can be sent along the edges.
\end{example}

Every node keeps a local version of the Tangle $\LocalTangleDAG{i}$ that we consider as the (local) configuration $\omega_i$ of node $i$.  
For ease of presentation, we consider the following simplified version of the OTV that does not keep track of where the actual blocks are attached in the Tangle but only keeps track of the supporters of the branches or conflicts. The (local) state space is therefore given by $\Q = \underbrace{2^{\nodeset}\times \ldots \times 2^{\nodeset}}_{|\conflictset|}$, where $\conflictset$ is a fixed set of conflicts and $2^{\nodeset}$ is the set of all possible subsets of $\nodeset=\{1,\ldots, N\}$.   

\begin{remark}
The simplified version described above allows a more accessible analysis of the voting on conflicting transactions. This comes with the cost of not describing the confirmation of non-conflicting transactions. We give more details on the ``liveness'' of these transactions in Section~\ref{sec:livenessHonest}.
\end{remark}

We interpret the packages a node $i$ receives as input assignments with values in the space of all packages $\M$. Each input assignment $e$ yields an update of the current configuration, and each configuration $\omega_{i}$ leads to an output assignment. We therefore consider two functions 
\begin{equation*}
I(e, \omega_{i}) :  \M \times \Q \mapsto \Q, 
\end{equation*}
and 
\begin{equation*}
O(\omega_i) :  \Q \mapsto \M^{|\nodeset_{i}|-1},
\end{equation*}
where $\nodeset_{i}$ are the neighbours of node $i$ in the communication graph $G$.
A node $i$ runs an algorithm $\algo= (I, O)$ that reacts to incoming packages by  updating its internal state $\omega_i$ and eventual sending outgoing packages indicating its state update. 
We also consider the configuration of the whole system that takes values  $\overline{\omega}=\{\omega_1,..,\omega_N\}  \in \Q^{N}$. 
The corresponding algorithm is denoted by  $\overline{\algo}$.

The creation of blocks uses randomness (by design) through the TSA. Moreover, issuing times of blocks may depend on the interactions of the node with the environment of our system. For this reason, we model the time between two successive blocks of one given node by random variables. oreover, the latency between packages of two given nodes is described by random variables.
This randomness turns our protocol into a \emph{random} protocol, and the randomness is described by the probability measure $\P$. As we consider a simplified model and are only interested in the supporters of given branches, the randomness enters only in the ``time components'' of the packages or edges.  Consequently,  edges become random variables.

\begin{definition}[Configuration graph]
Let  $e$ be a package sent from $i$ to $j$ in $G$ and $\overline{\omega}, \overline{\omega}' \in \mathcal{Q}^N$ be  two (global) configurations. We write $\overline{\omega} \stackrel{e}{\rightarrow} \overline{\omega}'$ if and only if
\begin{equation*}
    \P(I(e, \omega_i) =  \omega'_i )>0
\end{equation*}
for some $i$
and, $\P(\omega'_l=\omega_l)=1$ for all $l\neq i$. 
We say that~$\overline{\omega}'$ is \emph{accessible} from $\overline{\omega}$ by $e$. The notation of accessibility defines a directed graph on the set of (global) configurations that we dub the \emph{configuration graph} of the algorithm. 
\end{definition}

\begin{definition}[Valid packages]
A package (or edge) $e$ from node $i$ to node $j$ is called \emph{valid} given a global configuration $\overline{\omega}$ if and only if 
\begin{equation*}
    \P( O(\omega_{i}) \ni e)>0.
\end{equation*}
In other words, any valid edge must be the outcome of the algorithm $\overline{\algo}$. A sequence of edges $e_1, e_2, \ldots$ is called \emph{valid} given an (initial) configuration $\overline{\omega}(0)$ if and only if  $e_1$ is valid given $\overline{\omega}(0)$ and inductively
$e_\ell$ is valid given $\overline{\omega}{(\ell-1)}$, where $\overline{\omega}{(\ell-1)}$ is such that $\overline{\omega}(\ell-2) \stackrel{e_{\ell-1}}{\rightarrow} \overline{\omega}(\ell-1)$. 
\end{definition}

In the following, we assume that honest nodes only issue valid packages.

\begin{definition}[Communication of configurations]
We say that the (global) configuration $\overline{\omega}'$ is \emph{accessible} from a configuration $\overline{\omega}$ if and only if there exists a finite valid path  from $\overline{\omega}$ to $\overline{\omega}'$ in the configuration graph. In this case, we write $\overline{\omega} \rightarrow \overline{\omega}'$ 
We define that a configuration is accessible from itself. Two configurations $\overline{\omega}$ and $\overline{\omega}'$ are said to \emph{communicate} if and only if they are accessible from each other. In this case, we write $\overline{\omega} \leftrightarrow \overline{\omega}'$. 
\end{definition}
 
The relation $\leftrightarrow$  defines an equivalence relation on the set of configurations. 
\begin{definition}[Communication classes]
The equivalence classes of the equivalent relation $\leftrightarrow$ are called the \emph{communication classes} of the set of configurations.  A (communication) class is called \emph{closed} if and only if it has no outgoing edges, and \emph{open} otherwise.
\end{definition}

The closed communication classes play a vital role as they describe the outcome of the protocol. Let $R\in\branchset$ be a reality. Then the configuration with $\supporter_{\ledger_i}(R)=\nodeset$ for all $i\in \nodeset$ is a closed class. Let us note here that we are still assuming all nodes to be honest and behave according to the protocol.

\begin{definition}[Consensus state]\label{def:consensusState}
A state $\omega$ is called a \emph{consensus state} if and only if
\begin{equation}
    \supporter_{\ledger_i}(R)=\nodeset, \quad \forall i\in \nodeset,
\end{equation}
for some reality $R$.
\end{definition}

\begin{remark}
Let us stress that the definition of ``consensus state'' is only about agreeing on the preferred reality. It does not take into account the meaning of confirmation; see Definition~\ref{def:TTC}. Liveness and safety with respect to confirmation are discussed in the following sections.
\end{remark}

We make a crucial assumption about the communication layer.

\begin{ass}[Random block issuance and package delay]\label{ass: random Package Delay}
Block issuances and package delays are random and satisfy:
\begin{enumerate}
    \item Nodes issue new blocks independently and distributed according to some probability distribution $\mu_{\mathrm{iss}}$. 
    \item The delays of packages between two nodes are independent and distributed according to some probability distribution $\mu_{\mathrm{pack}}$.
    \item Block issuances and package delays are independent.
    \item With a positive probability packages are delivered faster than new blocks are issued. More precisely, if $X\sim\mu_{\mathrm{iss}}$ and $Y\sim\mu_{\mathrm{pack}}$, then $\P(Y<X)>0$.
    \end{enumerate}
\end{ass}

\begin{lemma}\label{lem:consensusIsPossible}
Under Assumption~\ref{ass: random Package Delay}, 
for every given configuration $\overline{\omega}$ there exists a consensus state $\overline{\omega_c}$ such that $\overline{\omega} \stackrel{e}{\rightarrow}  \overline{\omega}_c$.
\end{lemma}
\begin{proof}
Let $\overline{\omega}$ be a configuration. 
We wait until all existing packages and corresponding changes of votes are sent to all other nodes. During this time, no new block is issued with a positive (non-zero) probability. 
After every node has seen all current blocks every node has the same perception of the supporters of the different realities.  In other words, nodes agree on the AWs of the different branches. Now, every node changes its opinion to its preferred reality, issues transactions indicating their change of vote, and gossips them using packages on the communication graph. Once all these packages are seen by all nodes a consensus state is reached. 
\end{proof}

There are two immediate consequences of Lemma~\ref{lem:consensusIsPossible}.
\begin{cor}
Under Assumption~\ref{ass: random Package Delay}, a communication class is closed if and only if it consists of one consensus state.
\end{cor}

\begin{cor}
Under Assumption~\ref{ass: random Package Delay} (and in absence of an adversary), the protocol  converges ($\P$-almost surely)  to a consensus state.
\end{cor}

\begin{definition}[Schedule]
A \emph{schedule} on {the communication graph} $G$ is a sequence of (finite or infinite) valid edges $e_{1}, e_{2},\ldots$. 
A (finite or infinite) \emph{execution} of a sequence of edges $e_{1}, e_{2},\ldots$ by $\overline{\algo}$ on $G$ is a sequence of configurations $\overline{\omega}({0}) \stackrel{e_{1}}{\rightarrow} \overline{\omega}{(1)} \stackrel{e_{2}}{\rightarrow} \cdots$, where $\overline{\omega}{(0)}$ is the initial (global) configuration. 
\end{definition}

The above definitions can naturally extend to models that distinguish between honest and adversary nodes. We assume that adversary nodes do not have to follow the algorithm~$\algo$ but can produce messaging voting for non-preferred realities. On the communication level, adversary nodes may be more potent than honest nodes, i.e.  issuing blocks more frequently, and may delay the relaying of honest packages. Nevertheless, we assume that Assumption \ref{ass: random Package Delay} holds for all honest and malicious nodes. We say that the protocol reaches a consensus state if all honest nodes eventually prefer the same reality. Let us denote by $\nodeset_h$ and $\nodeset_a$ the set of honest and malicious nodes.  In analogy to the above, we obtain the following result.

\begin{theorem}[{Eventual consistency} - random blocks]\label{thm: security Random Blocks}
Assume Assumption~\ref{ass: random Package Delay} to hold for the blocks and packages of honest and malicious nodes and let $q$ be the weight of the adversary. Then, all honest nodes will ($\P$-almost surely) eventually prefer the same reality if $q<1/2$.
\end{theorem}

\begin{proof}
Since $q<1/2$ a consensus state is reached if all honest nodes have the same preferred reality and all nodes know about it, i.e. 
\begin{equation}
    \supporter_{\ledger_i}(R) \supset \nodeset_h, \quad \forall i\in \nodeset_h,
\end{equation}
for some reality $R$. We have to prove that for every given configuration $\omega$ there exists an available consensus state. This is proven similar to Lemma \ref{lem:consensusIsPossible} together with the situation where an adversary is neither issuing a block nor can delay the honest packages,  which occurs with a positive probability under Assumption \ref{ass: random Package Delay}.
\end{proof}

\section{Liveness and Safety}\label{sec: security}

In the previous section, we were interested in the eventual convergence and proved an optimal result in Theorem \ref{thm: security Random Blocks} under the assumption of random blocks issuance and random package delay. This section adds  the confirmation status of transactions into our considerations. We  divide security into liveness and safety to allow a more detailed and quantitative analysis.

From a general point of view, liveness means that eventually, good things will happen, and safety means that nothing wrong will ever happen. In our situation, this translates into the following.
The safety condition is that any two honest nodes should always reach an agreement and that this decision satisfies the specified validity conditions. Furthermore, no two nodes should ever confirm conflicting transactions. The liveness property is that each honest node should eventually make a  decision on the confirmation status of a transaction, i.e.  in our case all nodes reach the confirmation threshold $\theta$, see Definitions~\ref{def: confirmed block} and \ref{def: confirmed transaction}, eventually.  

\begin{remark}
In general, one requires in addition that the consensus protocol satisfies integrity.  Integrity requires that the eventual outcome of the consensus protocol was initially proposed by at least one node. Since in OTV honest nodes always pick a maximal branch, the integrity property is satisfied once the protocol terminates. 
\end{remark}

\subsection{Non-Conflicting Transactions}\label{sec:livenessHonest}

Liveness of a non-conflicting transaction is the property that it will eventually be included in the ledger state. In the strongest form, it means that every non-conflicting transaction will be confirmed, see Definitions~\ref{def: confirmed block} and \ref{def: confirmed transaction}.  
Therefore, the security threshold for liveness is at most a proportion $(1-\theta)$ of the weight, as an attacker or faulty nodes holding a proportion $(1-\theta)$ can stop the confirmation by not issuing any blocks anymore. 

Liveness is inherently linked with the TSA and the orphanage problem. 
We assume the following Assumption on the TSA\footnote{Formally it is more a requirement on the definition of ``tips''.} and we refer to Section~\ref{sec: liveness of transactions} for a discussion.

\begin{proposition}[Liveness and safety  of non-conflicting transactions]\label{prop:Liveness}
We assume in the asynchronous model that the tip pool size is stationary, and that Assumption~\ref{ass: TSA} is satisfied.  The weight of the malicious nodes is $q.$ Then, eventually every non-conflicting transaction is confirmed for all honest nodes if $q < 1- \theta$. 
\end{proposition}

\begin{proof} 
Let $x$ be a block containing a non-conflicting transaction $\tx{x}$ and consider an arbitrary honest node $i$. Each time this node issues a new block, the probability that it refers to (and votes for) $x$ is positive, due to Assumption~\ref{ass: TSA}.  At this point, it is important to have the second type of reference that allows to only  vote for the transaction $\tx{x}$ and not the whole Tangle past cone of $x$.
Let us denote by $p_j$ this last probability for the $j$th issued block. Then, due to the assumption on the stationarity of the tip pool size, there exists some $\varepsilon>0$ such that $p_j \geq \varepsilon$ for infinitely many indices $j$. Assumption~\ref{ass: TSA} guarantees the independence of these events, and the Lemma of Borel-Cantelli implies that node $i$ eventually votes for block $x$. Then, since the number of nodes is finite, all nodes eventually vote for $x$.
\end{proof}

Some discussion on the validity of the stationary tip pool size assumption is appropriate. This kind of assumption was also made throughout Section~\ref{sec:TTC} as Assumption~\ref{ass:constantWidth}.
Let us review this assumption in the light of the communication and the adversary model. 

An attacker can delay blocks with honest transactions such that the network delay $h$ increases. This, in turn, will inflate the tip pool size and the time to confirmation~\cite{penzkofer2021}. In the asynchronous model, this could lead to memory overflow of the nodes or halt confirmation of certain transactions.
While this attack is theoretically possible in this model, it is more of a theoretical interest than a practical issue. We also want to note here that nodes do have an efficient way to ``synchronize'' their perceptions of the Tangle due to the solidification process; see Section~\ref{sec:solidicationSynchronicity}. 

On the Tangle layer, the ``worst-case scenario'' seems to be the following. The adversary issues blocks, referencing already referenced blocks, not removing any tips from the tip pool. Under the assumption that nodes can issue blocks proportionally to their weight, we obtain that $q/(1-q)$ malicious blocks are issued for each honest block. Honest nodes can increase the number of references to keep the Tangle width stationary. More precisely, it is sufficient that the honest nodes' blocks, on average, remove $K:=(q/(1-q))+1$ tips. In other words, we can choose the number of references $k>K$ to guarantee robustness against this attack. For instance, $q=1/2$ leads to $K=2$.

\subsection{Conflicting Transactions}

Theoretical results on the liveness and safety of conflicting transactions rely heavily on the assumptions of the underlying communication and adversary model. Moreover, the analysis of the OTV protocol is complex: it requires modeling of the networking part, modeling of the weight distribution, and various (even an infinite number of) adversarial strategies. The following section shows that an adversary can hinder consensus finding in specific situations or edge cases. However, we want to emphasize that this interference only influences the liveness of conflicting transactions and that an appropriate TSA guarantees liveness of non-conflicting transactions; see Proposition ~\ref{prop:Liveness}. In Section~\ref{sec: SOTV}, we add a feature to the protocol that allows us to obtain theoretical results on the liveness of conflicting transactions.

\section{Impossibility Results and Metastability}\label{sec: impossibility Metastability}

Impossibility results play an essential role in the theory of consensus protocols, as they emphasize the limitations and critical edge cases. The most famous impossibility result is the FLP-result, \cite{FLP}, which states that achieving consensus in the asynchronous communication model is in general impossible for deterministic protocols.
From a general point of view, this impossibility is due to the possible delay of packages in the P2P communication and the resulting ``symmetric'' situation that hinders consensus finding.

We will consider the situation of two or more directly conflicting transactions. It is the role of the consensus mechanism to reach an agreement on which transaction should eventually be accepted. One may consider that keeping conflicting transactions in an undecided state, i.e.  violating the liveness, is acceptable. However, this is problematic for several reasons. For example, if nodes keep transactions indefinitely undecided, this could drastically inflate the communication required on the voting layer and prevent the pruning capability of the ledger. Transactions that are undecided for a long time can also harm safety. There is always a chance that some node confirms an ``undecided'' transaction. While  the probability of this event might be small, it is still positive, and hence this unlikely event will happen at some point in time.
We also note that simply rejecting malicious transactions does not provide a solution since this would allow delayed cancellation of transactions, thus, violating the system's safety.

In this section, we give examples where the liveness and safety of conflicting transactions are not satisfied; more complicated examples can be constructed following the same principles. They constitute an impossibility result in the sense that the proposed protocol does not guarantee liveness or safety under the asynchronous communication model. 
These situations rely on strong assumptions about the attackers. We distinguish between attacks on the communication level and those on the voting level. By requiring both levels we give a theoretical result when safety cannot be guaranteed, Lemma \ref{lem: negative result safety}.

\subsection{Communication Level}
We start with an example where an attacker does not take part directly in the voting but only controls the schedule of the honest nodes' blocks. Let us point out that the attacker does not need to control any weight in this scenario.

The first adversary attack is dubbed a \emph{metastability} attack since it tries to keep the honest nodes in an undecided situation. We refer to~\cite{PoMu:21} for more details and analysis of these kinds of attacks. On a conceptual level, these kinds of attacks exploit a situation where the system is kept in a roughly symmetric condition between two incompatible options. Once the symmetric scenario is broken, nodes likely converge quickly on one of the options.

\begin{example}[Metastability Attack I]\label{ex:meta1}
We consider $N=4$ participating nodes $\{1,2,3,4\}$ that communicate directly; the communication graph is the complete graph with four vertices. We assume that every node has the same weight, i.e.  $m_i=1/4$ for all $i\in \nodeset$. We consider the scenario of a simple double spend. The set of conflicts is, therefore, $\{\tx{x}, \tx{y}\}$. We assume for the sake of simplicity that a node prefers its own opinion if both conflicts have $50\%$ of AW.
Nodes $1$ and $2$ starts with an initial like of conflict $\tx{x}$ and nodes  $3$ and $4$ prefer $\tx{y}$. At the time $t_0$, every node $i$ communicates its vote to each of its neighbors by attaching a block $x_i$. The attacker delays these blocks (more precisely, the corresponding packages) either by some $\delta>0$ or $\gamma>\delta$. 
More precisely, we have the following edges, as defined in Equation (\ref{eq:edgeCommunication}), in our communication graph: 
\begin{equation*}
    (x_i, i,3, t_0, \delta),  (x_1, i,4, t_0, \delta), \quad i\in\{1,2\},
\end{equation*}
\begin{equation*}
    (x_j, j,1, t_0, \delta),  (x_j, j,2, t_0, \delta), \quad j\in\{3,4\},
\end{equation*}
and
\begin{align*}
    (x_1, 1,2, t_0, \gamma), &  (x_2, 2,1, t_0, \gamma), \cr (x_3, 3,4, t_0, \gamma), & (x_4, 4,3, t_0, \gamma).
\end{align*}
At time $\gamma$ this schedule leads to an inversion of the preferred conflicts, see Figure~\ref{fig:ex:meta1}. An attacker that controls the communication level could therefore delay consensus finding arbitrarily. 
To make the description of the former attacker more formal, we must specify the assumption on the issuance of block and the communication model.   For instance, in the synchronous model, with a known upper bound $\Delta$ on the network delay, such an attack is successful if the $\delta,\gamma < \Delta$ and the honest nodes issue blocks periodically.  In the asynchronous setting, an attacker can adjust the delays $\delta$ and $\gamma$ even if the honest nodes do not continuously issue their transactions simultaneously. 
\end{example}

\begin{figure}
    \centering
    \definecolor{megalightgray}{RGB}{244, 244, 244}
\definecolor{mygray}{RGB}{240, 240, 240}
\definecolor{myblue}{RGB}{102, 178, 255}
\definecolor{myred}{RGB}{255, 102, 102}
\definecolor{myorange}{RGB}{255, 178, 102}
\tikzstyle{rounded_block}=[draw, rectangle, thick, minimum height=\heightBlock cm, minimum width = \widthBlock cm, text centered, rounded corners, draw=darkgray, font = \small]

\begin{tikzpicture}[use Hobby shortcut, scale = 0.9]
\def\xCoordinate{0.0}
\def\yCoordinate{0.0}
\def\yAdd{1.5}
\def\xAdd{0.75}
\def\innerSepar{-1.5}
\def\lineWidth{1.5}
\def\roundedCorners{1}
\def\opacityInternal{0.5}
\def\minHeight{20}
\def\fracyAdd{1/6}
\def\lineWidthBelow{3}
\def\minWidth{20}
\def\scaleFactorSupp{0.12}
\def\minWidthCM{\minWidth*0.0352778*3/4}
\tikzstyle{block}=[draw, rectangle, minimum height=\minHeight pt, minimum width = \minWidth pt, text centered, rounded corners=\roundedCorners pt, draw=darkgray, font=\large]
\tikzstyle{block_colored}=[draw, rectangle, minimum height= 2 pt, minimum width = \minWidth * 5 / 6 pt, rounded corners=\roundedCorners  pt, draw=darkgray]
\def\arrowStyle{-latex}
\tikzstyle{block_circle}=[draw, circle, minimum height= 0.6 cm,  draw=teal]
\newcommand{\localRecord}[6]{
\node [rounded_block, fill=white, minimum height=0.4 cm, minimum width = 1.5cm] at (#5,#6) {} ;
\node[block_circle, thick, minimum size= 1 pt,  fill = #1] at (#5-0.5,#6) {};
\node[block_circle, thick, minimum height= 0.08, fill = #2] at (#5-0.17,#6) {};
\node[block_circle, thick, minimum height= 0.08, fill = #3] at (#5+0.17,#6) {};
\node[block_circle, thick, minimum height= 0.08, fill = #4] at (#5+0.5,#6) {};
}
\node [rounded_block, fill=megalightgray, minimum height=0.75 cm, minimum width = 8cm] at (0,0) {} ;

\node[block_circle, very thick, fill = myorange]  at (-3.6,0) {};
\localRecord{myorange}{white}{white}{white}{-3.6}{-0.75}

\node[block_circle, very thick, fill = myorange]  at (-1.2,0) {};
\localRecord{white}{myorange}{white}{white}{-1.2}{-0.75}

\node[block_circle, very thick, fill = myblue] at (1.2,0) {};
\localRecord{white}{white}{myblue}{white}{1.2}{-0.75}

\node[block_circle, very thick, fill = myblue] at (3.6,0) {};
\localRecord{white}{white}{white}{myblue}{3.6}{-0.75}

\node [rounded_block, fill=megalightgray, minimum height=0.75 cm, minimum width = 8cm] at (0,-2) {} ;

\node[block_circle, very thick, fill = myblue] at (-3.6,-2) {};
\localRecord{myorange}{white}{myblue}{myblue}{-3.6}{-2.75}

\node[block_circle, very thick, fill = myblue]  at (-1.2,-2) {};
\localRecord{white}{myorange}{myblue}{myblue}{-1.2}{-2.75}

\node[block_circle, very thick, fill = myorange] at (1.2,-2) {};
\localRecord{myorange}{myorange}{myblue}{white}{1.2}{-2.75}

\node[block_circle, very thick, fill = myorange] at (3.6,-2) {};
\localRecord{myorange}{myorange}{white}{myblue}{3.6}{-2.75}

\node [rounded_block, fill=megalightgray, minimum height=0.75 cm, minimum width = 8cm] at (0,-4) {} ;

\node[block_circle, very thick, fill = myblue] at (-3.6,-4) {};
\localRecord{myorange}{myorange}{myblue}{myblue}{-3.6}{-4.75}

\node[block_circle, very thick, fill = myblue]  at (-1.2,-4) {};
\localRecord{myorange}{myorange}{myblue}{myblue}{-1.2}{-4.75}

\node[block_circle, very thick, fill = myorange] at (1.2,-4) {};
\localRecord{myorange}{myorange}{myblue}{myblue}{1.2}{-4.75}

\node[block_circle, very thick, fill = myorange] at (3.6,-4) {};
\localRecord{myorange}{myorange}{myblue}{myblue}{3.6}{-4.75}

\node [rounded_block, fill=megalightgray, minimum height=0.75 cm, minimum width = 1.2cm] at (-3.6,-6) {} ;
\node[block_circle, very thick, fill = none] at (-3.6,-6) {};
\node[font = \small] at (-1.6,-6) {\begin{tabular}{l}
    Current opinion \\
    of a node
\end{tabular}};

\localRecord{white}{white}{white}{white}{1.2}{-6}
\node[font = \small] at (3.2,-6) {\begin{tabular}{l}
    Local record \\
    of votes
\end{tabular}};

\node [single arrow,top  color=red, bottom color=violet,
single arrow head extend=3pt,transform shape, opacity = 0.4, minimum height=85, text opacity=1, rotate = 270, anchor=west] 
at (4.9, 0.5){$\delta$};
\node [single arrow,top  color=violet, bottom color=blue,
single arrow head extend=3pt,transform shape, opacity = 0.4, minimum height=55, text opacity=1, rotate = 270, anchor=west] 
at (4.9, -2.5){$\gamma-\delta$};

\end{tikzpicture}
    \caption{Illustration of Example~\ref{ex:meta1}.  Nodes are voting for transaction $\tx{x}$ (blue) or $\tx{y}$ (orange). Each node ends up with the opposite opinion it started with, thus, creating a deadlock.}
    \label{fig:ex:meta1}
\end{figure}
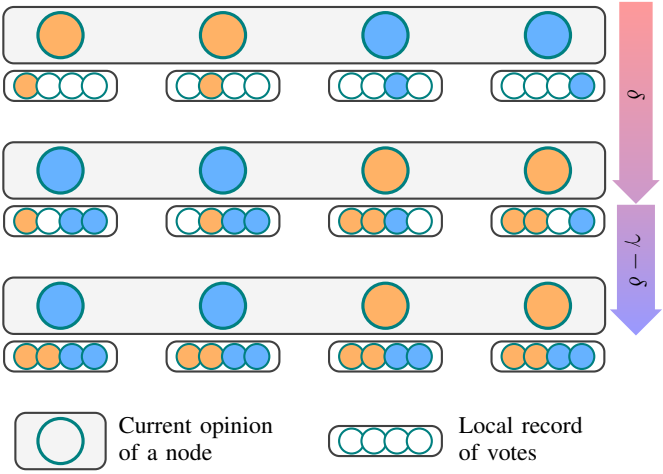

\begin{remark}
The situation described above is undoubtedly a special case and mainly of theoretical interest. However, it raises the question under which conditions such schedules exist and how realistic they appear in real applications.
\end{remark}

\subsection{Voting Level}
In this section, we describe situations, where an attacker can successfully interfere in the consensus finding by using the voting layer. We do not need conditions to control communication between honest nodes but relatively strong assumptions about the adversary's ability to issue new blocks and reliably forward them to the honest nodes.

\begin{example}[Metastability Attack II]\label{ex:meta2}
We again consider the situation of one double spend, i.e.  a set of conflicts $\{\tx{x}, \tx{y}\}$. In this attack, the adversary votes for the minority, i.e.  the conflict that has less AW. The attack is supposed not to influence the communication layer, and we work under the assumption of the synchronous communication model. 
We assume that the propagation of blocks happens fast, i.e.  each block causes a state update in all other nodes. Furthermore, we assume that the adversary can issue at a high rate, such that for every other honest node's block, the adversary can issue a block.

We consider an even number $N_h$ of honest nodes and three malicious nodes, and where each node holds the same weight. 
We say if a node votes for $\tx{x}$ or $\tx{y}$ it is in set $X$ and $Y$, respectively. The protocol starts with $\sfrac{1}{2}N_h$ honest nodes initially voting for $\tx{x}$ and $\sfrac{1}{2} N_h$ honest nodes voting for $\tx{y}$. We refer to Figure~\ref{fig:ex:meta2} for an illustration.
Next, the adversary votes for $\tx{x}$ (with all three nodes), resulting in a vote of $|X|/|Y| = ( \sfrac{1}{2}N_h+3) / ( \sfrac{1}{2}N_h)$. Nodes in $X$ will continue to vote in favour of $\tx{x}$. 
On the other hand, an honest node in $Y$ will eventually change its vote and issue a transaction in favor of $\tx{x}$, thus, changing from set $Y$ to $X$. Now, before any other honest nodes can express their vote, the attacker switches its vote to $\tx{y}$. Hence, in total we have $|X|/|Y| = ( \sfrac{1}{2}N_h + 1) / ( \sfrac{1}{2}N_h+2)$. 
Honest nodes will now vote for $\tx{y}$. However, as soon as a node from $X$ changes its vote, the resulting situation is symmetric to the initial condition. Thus, the adversary can repeat this ad infinitum.

\end{example}
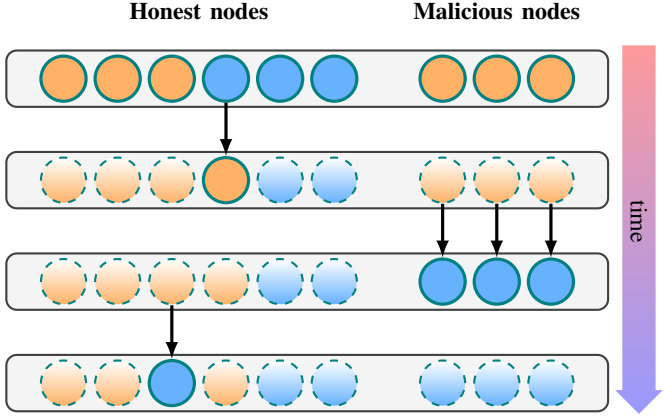
\begin{figure}
    \centering
    \definecolor{megalightgray}{RGB}{244, 244, 244}
\definecolor{mygray}{RGB}{240, 240, 240}
\definecolor{myblue}{RGB}{102, 178, 255}
\definecolor{myred}{RGB}{255, 102, 102}
\definecolor{myorange}{RGB}{255, 178, 102}
\tikzstyle{rounded_block}=[draw, rectangle, thick, minimum height=\heightBlock cm, minimum width = \widthBlock cm, text centered, rounded corners, draw=darkgray, font = \small]

\begin{tikzpicture}[use Hobby shortcut, scale = 0.9]
\def\xCoordinate{0.0}
\def\yCoordinate{0.0}
\def\yAdd{1.5}
\def\xAdd{0.75}
\def\innerSepar{-1.5}
\def\lineWidth{1.5}
\def\roundedCorners{1}
\def\opacityInternal{0.5}
\def\minHeight{20}
\def\fracyAdd{1/6}
\def\lineWidthBelow{3}
\def\minWidth{20}
\def\scaleFactorSupp{0.12}
\def\minWidthCM{\minWidth*0.0352778*3/4}
\tikzstyle{block}=[draw, rectangle, minimum height=\minHeight pt, minimum width = \minWidth pt, text centered, rounded corners=\roundedCorners pt, draw=darkgray, font=\large]
\tikzstyle{block_colored}=[draw, rectangle, minimum height= 2 pt, minimum width = \minWidth * 5 / 6 pt, rounded corners=\roundedCorners  pt, draw=darkgray]
\def\arrowStyle{-latex}
\tikzstyle{block_circle}=[draw, circle, minimum height= 0.6 cm,  draw=teal]

\node [rounded_block, fill=megalightgray, minimum height=0.75 cm, minimum width = 8cm] at (0,0) {} ;

\node[block_circle, very thick, fill = myorange] at (-3.6,0) {};
\node[block_circle, very thick, fill = myorange] at (-2.8,0) {};
\node[block_circle, very thick, fill = myorange] at (-2.0,0) {};
\node[block_circle, very thick, fill = myblue] (node_1_4) at (-1.2,0) {};
\node[block_circle, very thick, fill = myblue] at (-0.4,0) {};
\node[block_circle, very thick, fill = myblue] at (0.4,0) {};

\node[block_circle, very thick, fill = myorange] at (2,0) {};
\node[block_circle, very thick, fill = myorange] at (2.8,0) {};
\node[block_circle, very thick, fill = myorange] at (3.6,0) {};

\node [rounded_block, fill=megalightgray, minimum height=0.75 cm, minimum width = 8cm] at (0,-1.5) {} ;

\node[block_circle,  thick, top color=white, bottom color=myorange, dashed] at (-3.6,-1.5) {};
\node[block_circle,  thick, top color=white, bottom color=myorange, dashed] at (-2.8,-1.5) {};
\node[block_circle,  thick, top color=white, bottom color=myorange, dashed] at (-2.0,-1.5) {};
\node[block_circle, very thick, fill = myorange] (node_2_4) at (-1.2,-1.5) {};
\draw[\arrowStyle, very thick] (node_1_4) -- (node_2_4); 
\node[block_circle,  thick, top color=white, bottom color=myblue, dashed] at (-0.4,-1.5) {};
\node[block_circle,  thick, top color=white, bottom color=myblue, dashed] at (0.4,-1.5) {};

\node[block_circle,  thick, top color=white, bottom color=myorange, dashed] (node_2_7) at (2,-1.5) {};
\node[block_circle,  thick, top color=white, bottom color=myorange, dashed] (node_2_8) at (2.8,-1.5) {};
\node[block_circle,  thick, top color=white, bottom color=myorange, dashed] (node_2_9) at (3.6,-1.5) {};

\node [rounded_block, fill=megalightgray, minimum height=0.75 cm, minimum width = 8cm] at (0,-3) {} ;

\node[block_circle,  thick, top color=white, bottom color=myorange, dashed] at (-3.6,-3) {};
\node[block_circle,  thick, top color=white, bottom color=myorange, dashed] at (-2.8,-3) {};
\node[block_circle,  thick, top color=white, bottom color=myorange, dashed] (node_3_3)  at (-2.0,-3) {};
\node[block_circle, thick, top color=white, bottom color=myorange, dashed] at (-1.2,-3) {};
\node[block_circle,  thick, top color=white, bottom color=myblue, dashed] at (-0.4,-3) {};
\node[block_circle,  thick, top color=white, bottom color=myblue, dashed] at (0.4,-3) {};

\node[block_circle,  very thick, fill = myblue] (node_3_7) at (2,-3) {};
\node[block_circle,  very thick, fill = myblue] (node_3_8) at (2.8,-3) {};
\node[block_circle,  very thick, fill = myblue] (node_3_9) at (3.6,-3) {};
\draw[\arrowStyle, very thick] (node_2_7) -- (node_3_7); 
\draw[\arrowStyle, very thick] (node_2_8) -- (node_3_8); 
\draw[\arrowStyle, very thick] (node_2_9) -- (node_3_9);

\node [rounded_block, fill=megalightgray, minimum height=0.75 cm, minimum width = 8cm] at (0,-4.5) {} ;

\node[block_circle,  thick, top color=white, bottom color=myorange, dashed] at (-3.6,-4.5) {};
\node[block_circle,  thick, top color=white, bottom color=myorange, dashed] at (-2.8,-4.5) {};
\node[block_circle,  very thick, fill = myblue] (node_4_3)  at (-2.0,-4.5) {};
\draw[\arrowStyle, very thick] (node_3_3) -- (node_4_3);
\node[block_circle, thick, top color=white, bottom color=myorange, dashed] at (-1.2,-4.5) {};
\node[block_circle,  thick, top color=white, bottom color=myblue, dashed] at (-0.4,-4.5) {};
\node[block_circle,  thick, top color=white, bottom color=myblue, dashed] at (0.4,-4.5) {};

\node[block_circle,  thick, top color = white, bottom color = myblue, dashed]  at (2,-4.5) {};
\node[block_circle,  thick, top color = white, bottom color = myblue, dashed] at (2.8,-4.5) {};
\node[block_circle,  thick, top color = white, bottom color = myblue, dashed] at (3.6,-4.5) {};

\node [single arrow,top  color=red, bottom color=blue,
single arrow head extend=3pt,transform shape, opacity = 0.4, minimum height=155pt, text opacity=1, rotate = 270, anchor=west] 
at (4.9, 0.5){time
};
\node[font = \small] at (-1.6, 1) {\textbf{Honest nodes}};
\node[font = \small] at (2.8, 1) {\textbf{Malicious nodes}};
\end{tikzpicture}
    \caption{Illustration of Example~\ref{ex:meta2}. Nodes are voting for transaction $\tx{x}$ (blue) or $\tx{y}$ (orange). }
    \label{fig:ex:meta2}
\end{figure}

\begin{remark}
We want to note that in Example~\ref{ex:meta2} the attacker heavily relies on the capability of an adversary to immediately adapt its opinion before more than $2$ honest nodes changed their vote to the majority. 
\end{remark}

The next example, the \emph{Bait-and-Switch Attack}, depends less on the adversaries issuance rate but requires a higher amount of weight.

\begin{example}[Bait-and-Switch Attack]\label{ex:BaS}
We consider a situation where the adversary possesses the node with the highest weight. The strategy is to switch frequently the opinions such that the honest nodes are constantly `` chasing the ever-changing heaviest branch''. 
For example, consider $N_h$ honest nodes with total weight $w_h$ and individual weight $w_h/N_h$ and one adversary node with weight $w_a$. 
Let $n_{cr}$ be the largest natural number such that
$$
n_{cr} \cdot \frac{w_h}{N_h} < w_a.
$$
In the beginning, the malicious node spends an output in a transaction $\tx{x}_1$. Then, before $n_{cr}$ nodes with a total weight of less than $w_A$ express their vote, the adversary spends the same output in transaction $\tx{x}_2$, i.e.  creates a conflicting transaction with $\tx{x}_1$, and (implicitly) votes for the new transaction $\tx{x}_2$. Since $\tx{x}_2$ becomes the heaviest branch, all honest nodes will vote for this transaction. The adversary repeats this procedure by creating additional double spends repetitively.
\end{example}

\subsection{Communication and Voting Level}
In the previous sections, we presented examples of how an adversary can harm the liveness of conflicting transactions. The attacker strategies required either substantial control of the communication layer or a high issuance rate combined with considerable weight.
In this section, we prove an impossibility result for safety that involves an attack strategy that uses both levels. 

\begin{definition}[Broken safety]\label{def: broken safety}
We say that \emph{safety is broken} if and only if there exist two honest nodes $i,j$ and conflicting transactions $\tx{x}$ and $\tx{y}$ such that  for some times $s,t$ we have
$$
\AW_{{i,t}}(\tx{x}) > \theta \mbox{ and } \AW_{{j,s}}(\tx{y}) > \theta.
$$
\end{definition}

We have the following ``negative'' result.

\begin{lemma}\label{lem: negative result safety}
Let $q > \theta -0.5$ be the weight of the adversary. Assume that the weight of the honest nodes is equally distributed on sufficiently many honest nodes.
Then, there exists an adversary strategy that breaks safety.
\end{lemma}

\begin{proof}
Let us choose a number of honest nodes $N_h$ sufficiently large such that there exists some $N_h^*<N_h$ such that
\[
\frac{\theta -q}{1-q} < \frac{N_h^*}{N_h} < \frac{0.5}{1-q}.
\]
An attacker starts issuing two conflicting transactions $\tx{x}$ and $\tx{y}$.
The attacker decomposes the honest nodes into two groups $X$ and $Y$ such that each of these groups forms a connected subgraph of the underlying communication layer, while the attacker is connected to both groups. Group $X$ consists of $N_h^*$ nodes and group $Y$ of $N_h-N_h^*$ nodes.
The attacker interferes with the schedule such that nodes in each group only receive blocks from their group. The attacker changes the schedule such that the nodes in $X$ receive transaction~$\tx{x}$ before $\tx{y}$ and the nodes in $Y$ receive~$\tx{y}$ before~$\tx{x}$. All honest nodes prepare their initial statement of their preferred transaction ($\tx{x}$ for group  $X$ and $\tx{y}$ for group $Y$) and send them to their neighbours.  

The attacker sends to $X$ blocks that state that it prefers $\tx{x}$.
As a consequence, nodes from $X$ 
confirm transaction~$\tx{x}$ since $\AW(\tx{x})=(1-q)\frac{N_h^*}{N_h} + q > \theta$. 

After this, the attacker sends blocks to $Y$ (and $X$) that it votes now for transaction~$\tx{y}$. Without the vote of the attacker for transaction $\tx{x}$ the AW of $\tx{x}$ in $X$  reduces to  $\AW(\tx{x})=(1-q)\frac{N_h^*}{N_h} < 0.5$.

Next, the attacker lets $X$ know about the preferences of $Y$. At this point $\AW(\tx{y})>\AW(\tx{x})$  
and as a consequence nodes from $X$ update their preferred reality and vote for $\tx{y}$. This  eventually leads to $\AW(\tx{y})>\theta$ for all nodes. As by Definition~\ref{def: broken safety} safety is broken.
\end{proof}

The above proof indicates that the attacker needs very strong control over the communication layer to conduct such an attack. Nevertheless, it gives a reasonable theoretical security threshold for the protocol's safety. All the more since we can prove safety under the  assumption $q < \theta -0.5$ in Section \ref{sec: SOTV}.

\subsection{Realistic Conditions}

The above examples illustrate that the \emph{two dimensions}, namely the communication and voting level, may interact either in favor of the attacker or in favor of the robustness of the protocol. In all cases, it seems that the attacker needs excellent control of the communication layer of the protocol. Randomness or uncertainty on the communication layer may interfere with the adversary strategy and finally lead to convergence of the honest nodes' opinions. 

We conjecture that these strong assumptions are not met in most reasonable real-world scenarios and that the attacks that rely solely on the communication level are hard to perform in practice.

With a completely random schedule of packages, the system will eventually converge to a consensus state in situations where an attacker controls not more than half of the total weight,  see Theorem~\ref{thm: security Random Blocks}. However, this convergence time can be \emph{impracticably long} for real-world applications and it is possible that safety (for the confirmation) can be broken as shown by Lemma~\ref{lem: negative result safety}.
The theoretical treatment of the inherent randomness of real-world implementation systems is at best in an early state, and a quantification or even its control seems currently out of reach. We refer to~\cite{BBBG:15} for a theoretical approach to describe the entropy related to the scheduling of the transactions. 

The following section proposes a more sophisticated variation that allows a more straightforward theoretical treatment and provides the ``optimal'' safety thresholds.

\section{Synchronized Random Reality Selection}\label{sec: SOTV}

In the previous section, we demonstrated that under several conditions, the protocol presented so far might lead to situations where nodes cannot come to an agreement between several valid options. This section offers a mechanism to overcome this scenario by utilising external randomness. As shown in~\cite{POPOV2021, Capossele2021, PoMu:21} common randomness can successfully navigate a system away from such an undesired situation.

Pre-consensus classes are those classes from which the network reaches a consensus eventually.
The aim of the design of the consensus protocol is, therefore, to construct the protocol so that its global state reaches such a pre-consensus state fast and that from there, the actual consensus state is inevitable.

The OTV is an asynchronous protocol and comes with advantages and disadvantages. One disadvantage is the lack of synchronization possibilities between nodes that could be used against adversarial attacks on the communication level. The arguments and examples in the previous section showed that it is theoretically possible for an attacker to keep the honest nodes in an undecided situation for a long time. To exclude these cases and obtain theoretical results, we use a distributed random number generation (dRNG) process to synchronize the nodes and interfere with a possible adversary. 

We choose a parameter $\dFPCS$ describing the length of epochs between synchronizations times. In other words, once in every $\dFPCS$ time units, we synchronize the nodes with the help of a given dRNG process. This procedure is inspired by the paper~\cite{POPOV2021}, where a dRNG is used to construct a voting-based consensus protocol in a Byzantine environment. The dRNG  allows the consensus protocol to reach a pre-consensus state with a positive (non-zero) probability. This probability is uniform in the opinions and votes of the nodes, and hence, the protocol enters a pre-consensus class in a geometrically distributed number of periods of length $\dFPCS$. In the last step, we then prove that consensus is reached from the pre-consensus state.

We consider a system of $N=N_h+N_a$ nodes with $N_h$ honest nodes and $N_a$ adversarial nodes. The  honest nodes are identified with the set $\nodeset_h =\{1,\ldots, N_h\}$ and the adversarial nodes with $\nodeset_a = \{ N_h +1 , \ldots, N_h+N_a\}$.

We start with stating our model assumptions.

\begin{ass}\label{ass: Synchronised Random Reality Selection} 
We make the following assumptions:
\begin{enumerate}[label=\ref{ass: Synchronised Random Reality Selection}.\arabic*]
\item Every block from an honest node is received by another honest node during time $\ddRNG=\ddRNG(\varepsilon)$ with probability of at least $1-\varepsilon$. The constant $\varepsilon>0$ can be chosen arbitrarily small.
The events for each block are independent of each other.
\item  The adversary controls a proportion $q$ of the weight. The adversary might have an influence on the schedule of the blocks to the extent of~\ref{ass: Synchronised Random Reality Selection}.1.
\item  The set of conflicts $\conflictset$ is fixed and does not vary in time. All nodes perceive the same $\conflictset$.
\item There exists a dRNG that publishes a random variable every $\dFPCS$ unit of times. 
The random variable is uniformly distributed on the interval $[0.5, \theta]$, where $\theta$ is the confirmation threshold; see Section~\ref{sec: WW confirmation rule}. This value is received (independently) by every given node before time $\ddRNG$ (in every epoch) with a probability of at least $1-\varepsilon$. 
\item Honest nodes of cumulative weight of at least $\theta$ issue blocks expressing support for their preferred reality\footnote{In other words, for every conflict $c$  in the preferred reality the node issues at least one block stating that the node votes for this conflict $c$ and  it doesn't issue any block stating that the node votes for a transaction conflicting with $c$.} at least every $\dFPCS/2$ time units with a probability of at least $1-\varepsilon$. 
\end{enumerate}
\end{ass}

Let us comment on the validity of the above assumptions. Assumption~\ref{ass: Synchronised Random Reality Selection}.1  is essentially a probabilistic synchronicity assumption. The fact that the probability $\varepsilon$ can be chosen arbitrarily small is supported by the fact that votes are blocks in the Tangle that can be re-broadcast or obtained by solidification requests; see Section~\ref{sec:solidicationSynchronicity}. The independence assumption is essential and the study of correlated errors is out of the scope of this paper. Assumption~\ref{ass: Synchronised Random Reality Selection}.2 is natural in a probabilistic synchronous model. Assumption~\ref{ass: Synchronised Random Reality Selection}.3 is essentially for ease of presentation.  As nodes will consider only conflicts of a certain age, older than $\dFPCS$,  Assumption~\ref{ass: Synchronised Random Reality Selection}.1 ensures that nodes already have the same perception of the sets of conflicts with a very high probability. 
Assumption~\ref{ass: Synchronised Random Reality Selection}.4 was used in previous work,~\cite{POPOV2021, Capossele2021, PoMu:21}. A sequence of such common random numbers can be either provided by an external source or generated
by the nodes of the system themselves; 
see e.g.~\cite{cascudo2017scrape, Lenstra_Wes17, popov2017decentralized, 
schindlerhydrand, syta2017scalable, wes}.
Let us stress that it is necessary that the randomness of the dRNG is not predictable and obtained in each epoch by the majority of the weight with a positive probability. However, we do not require that all honest nodes agree on this random number.\footnote{The idea is that a ``weak consensus'' on the randomness of the dRNG leads to an eventual ``strong consensus'' on the ledger state.}
The last Assumption~\ref{ass: Synchronised Random Reality Selection}.5 is an (almost) necessary condition to ensure that transaction  have a chance to be confirmed.

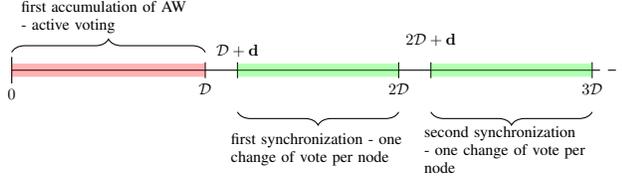
\begin{figure}[h!]
    \centering
    \begin{tikzpicture}
        \draw (\linewidth*0.0,20pt-90pt) -- (\linewidth*0.9,20pt-90pt);
        \draw[dashed] (\linewidth*0.9,20pt-90pt) -- (\linewidth*0.95,20pt-90pt);
        \draw (\linewidth*0.0,25pt-90pt) -- (\linewidth*0.0,15pt-90pt) node[anchor=north, scale=0.6] {$0$};
        \draw (\linewidth*0.9,23pt-90pt) -- (\linewidth*0.9,17pt-90pt) node[anchor=north, scale=0.6] {$3\dFPCS$};
        \draw (\linewidth*0.6,23pt-90pt) -- (\linewidth*0.6,17pt-90pt) node[anchor=north, scale=0.6] {$2\dFPCS$};
        \draw (\linewidth*0.65,23pt-90pt) -- (\linewidth*0.65,17pt-90pt) node[anchor=south, yshift=1em, scale=0.6] {$2 \dFPCS+ \ddRNG$};
         \draw (\linewidth*0.3,23pt-90pt) -- (\linewidth*0.3,17pt-90pt) node[anchor=north, scale=0.6] {$\dFPCS$};
          \draw (\linewidth*0.35,23pt-90pt) -- (\linewidth*0.35,17pt-90pt) node[anchor=south, scale=0.6, yshift=1em] {$\dFPCS+ \ddRNG$};
        \draw[opacity=0.3, line width=5pt, red] (\linewidth*0.0,20pt-90pt) -- (\linewidth*0.3,20pt-90pt);
        \draw [decorate,decoration={brace,amplitude=5pt,raise=1.5ex}]
  (\linewidth*0.0,20pt-90pt) -- (\linewidth*0.3,20pt-90pt) node[midway,yshift=2em, scale=0.6, text width=\linewidth*0.45]{first accumulation of AW - active voting};
        \draw[opacity=0.3, line width=5pt, green] (\linewidth*0.35,20pt-90pt) -- (\linewidth*0.6,20pt-90pt);
        \draw [decorate,decoration={brace, mirror, amplitude=5pt,raise=3.5ex}]
  (\linewidth*0.35,20pt-90pt) -- (\linewidth*0.6,20pt-90pt) node[midway,yshift=-3em, scale=0.6, text width=\linewidth*0.45]{first synchronization - one change of vote per node };
        \draw[opacity=0.3, line width=5pt, green] (\linewidth*0.65,20pt-90pt) -- (\linewidth*0.9,20pt-90pt);
        \draw [decorate,decoration={brace, mirror, amplitude=5pt,raise=3.5ex}]
  (\linewidth*0.65,20pt-90pt) -- (\linewidth*0.9,20pt-90pt) node[midway,yshift=-3em, scale=0.6, text width=\linewidth*0.45]{second synchronization - one change of vote per node};
\end{tikzpicture}
\caption{The different epochs in the synchronisation. }
\label{fig:preconsensus}
\end{figure}

In the beginning, before time $\dFPCS$, the AWs for each conflict $c\in \conflictset$ grow through votes according to the mechanism described in Section~\ref{sec: tip selection algorithms}. At the end of this initial interval, every node has its own perception of the AW of a conflict $c$, written as $\AW_{{i,\dFPCS}}(c)$.

After the arrival of the first dRNG randomness $X$ (between $\dFPCS$ and $\dFPCS+\ddRNG$), every honest node chooses its preferred reality and adheres with it during the next interval of length $\dFPCS$. 

In Algorithm~\ref{alg: selection Conflict Graph}, we describe an iterative  procedure, inspired by \cite{FPCS}, for choosing a preferred reality by a node. 
First, it initialises set $R$ to be the empty set and $U$ to be the set of conflicts $\conflictset$. At every step of the first while-loop, the node finds a conflict $c^*$ in $U$ with the highest AW. If $\AW(c^*)> X$, then we add $c^*$ to $R$, remove all transactions from $U$ conflicting with $R$ and repeat this step. We additionally require $c^*$ to be from $\max_{\conflictset}(U)$ (see Definition~\ref{def: min and max elements}) to guarantee that after adding $c^*$ to $R$, the updated set $R$ is a branch.   If  $\AW(c^*)\le X$, then we run the next iterative procedure (while-loop) which updates $R$ by $c^*$, where  $c^{*}$ is the conflict $c$ in $U$ attaining the largest hash of the concatenation $c||X$\footnote{We assume that $c||X$ can be treated as a binary string for a proper usage of the hash function as noted in Remark~\ref{rem: hash function}} and proceed similarly until $U$ becomes empty.  By construction, the resulting set $R$ is a maximal branch or a reality. We summarize these results in the following proposition.

\begin{proposition}
The resulting set $R$ in Algorithm~\ref{alg: selection Conflict Graph} is a reality.
\end{proposition}

\begin{algorithm}[t]
\caption{Reality selection algorithm with common coin}
\label{alg: selection Conflict Graph}
\KwData{Conflict Graph $\ConflictGraph=(\conflictset,E)$, common randomness $X$ distributed uniformly in $[0.5,\theta]$}
\KwResult{preferred reality ${R}\in \branchset$}
${R}\gets \emptyset$\\
$U\gets \conflictset$\\
\While{$|U|\neq 0$}{
$c^*\gets \argmax \{ \AW(c): c\in \max_{\conflictset}(U)\}$ {\scriptsize{\Comment*[r]{use $\max\hash(c)$ for breaking ties}}}
\eIf{$\AW(c^*)>X$}{
${R} \gets {R} \cup \{c^*\}$\\ $U \gets U  \setminus \{N_{\conflictset}(c^{*}) \cup\{c^{*}\}\}$
}
{
break the while-loop}
}
\While{$|U|\neq 0$}{
$c^*\gets \argmax \{ \hash(c||X): c\in \max_{\conflictset}(U)\}$ \\
${R} \gets {R} \cup \{c^*\}$\\ $U \gets U  \setminus \{N_{\conflictset}(c^{*}) \cup\{c^{*}\}\}$
}
\end{algorithm}

Denote by $\supporter_{\ledger_{i,t}}^{(h)}(\tx{x})$ the set of honest nodes seen from node $i$ at time $t$ that issued a block that votes for a  transaction $\tx{x}$ (for a similar definition of supporters, see Definition~\ref{def: AW and supporters}). The \emph{honest AW} of $\tx{x}$ seen from node $i$ at time $t$ is defined as
$$
\AW_{{i,t}}^{(h)} (\tx{x}) := \sum_{j \in \supporter_{\ledger_{i,t}}^{(h)}(\tx{x})} \weight(j)
$$

Due to Assumption~\ref{ass: Synchronised Random Reality Selection}.5 and since the honest nodes change their vote at most once, every other honest node sees this vote with a very high probability.  In other words, every honest node has the same perception of the votes of all other honest nodes (with high probability).
In this case, we can speak of the  \emph{honest AW seen by the honest nodes} of a transaction $\tx{x}$:
\begin{equation}\label{eq:honestAW}
    \AW_{t}^{(h)}(\tx{x}) := \AW_{1,t}^{(h)}(\tx{x})
\end{equation}
 if it holds that $\AW_{{i,t}}^{(h)}(\tx{x}) = \AW_{{j,t}}^{(h)}(\tx{x})$ for all $1\le i,j \le N_h$.
 
Adversarial nodes may change their opinions. In particular, they can do this close to the threshold time $\dFPCS$ such that honest nodes may have different perceptions of the adversarial votes. However, this difference in perception is bounded by the weight of the adversary. 
For every $c\in \conflictset$ we define, similar to \cite{FPCS}, the \emph{regions (or intervals) of adversarial control} as
\begin{equation}
    I_t(c) = [\AW_{t}^{(h)}(c), \AW_{t}^{(h)}(c)+ q];
\end{equation}
see  Fig.~\ref{fig:regionsOfControl}.
The lower (resp. upper boundary) of this interval is precisely the overall AW of the conflict when all malicious nodes vote against (resp. for) it. 
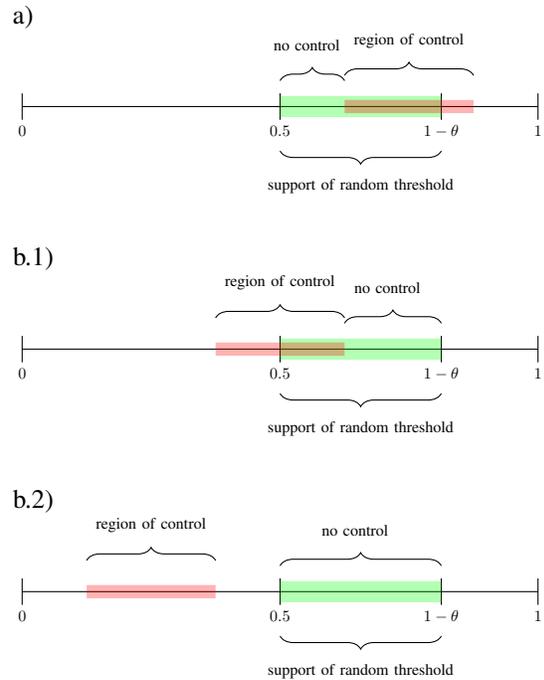
\begin{figure}[h!]
    
    a)
    
    \begin{tikzpicture}
        \draw (\linewidth*0.1,0pt) -- (\linewidth*0.9,0pt);
        \draw (\linewidth*0.1,5pt) -- (\linewidth*0.1,- 5pt) node[anchor=north, scale=0.6] {$0$};
        \draw (\linewidth*0.5,5pt) -- (\linewidth*0.5,-5pt) node[anchor=north, scale=0.6] {$0.5$};
        \draw (\linewidth*0.75,5pt) -- (\linewidth*0.75,-5pt) node[anchor=north, scale=0.6] {$1-\theta$};
        \draw (\linewidth*0.9,5pt) -- (\linewidth*0.9,-5pt) node[anchor=north, scale=0.6] {$1$};
        \draw[opacity=0.3, line width=8pt, green] (\linewidth*0.5,0pt) -- (\linewidth*0.75,0pt);
        \draw [decorate,decoration={brace, mirror, amplitude=5pt,raise=1.5ex}]
  (\linewidth*0.5,-10pt) -- (\linewidth*0.75,-10pt) node[midway,yshift=-2em, scale=0.6]{support of random threshold};
        \draw[opacity=0.3, line width=5pt, red] (\linewidth*0.6,0pt) -- (\linewidth*0.8,0pt);
        \draw [decorate,decoration={brace, amplitude=5pt,raise=1.5ex}]
  (\linewidth*0.6,5pt) -- (\linewidth*0.8,5pt) node[midway,yshift=2em, scale=0.6]{region of control};
  \draw [decorate,decoration={brace, amplitude=5pt,raise=1.5ex}]
  (\linewidth*0.5,3pt) -- (\linewidth*0.6,3pt) node[midway,yshift=2em, scale=0.6, text width=\linewidth*0.2]{no control};
\end{tikzpicture}
\vspace{0.5cm}

b.1)

\begin{tikzpicture}
        \draw (\linewidth*0.1,0pt) -- (\linewidth*0.9,0pt);
        \draw (\linewidth*0.1,5pt) -- (\linewidth*0.1,- 5pt) node[anchor=north, scale=0.6] {$0$};
        \draw (\linewidth*0.5,5pt) -- (\linewidth*0.5,-5pt) node[anchor=north, scale=0.6] {$0.5$};
        \draw (\linewidth*0.75,5pt) -- (\linewidth*0.75,-5pt) node[anchor=north, scale=0.6] {$1-\theta$};
        \draw (\linewidth*0.9,5pt) -- (\linewidth*0.9,-5pt) node[anchor=north, scale=0.6] {$1$};
        \draw[opacity=0.3, line width=8pt, green] (\linewidth*0.5,0pt) -- (\linewidth*0.75,0pt);
        \draw [decorate,decoration={brace, mirror, amplitude=5pt,raise=1.5ex}]
  (\linewidth*0.5,-10pt) -- (\linewidth*0.75,-10pt) node[midway,yshift=-2em, scale=0.6]{support of random threshold};
        \draw[opacity=0.3, line width=5pt, red] (\linewidth*0.4,0pt) -- (\linewidth*0.6,0pt);
        \draw [decorate,decoration={brace, amplitude=5pt,raise=1.5ex}]
  (\linewidth*0.4,5pt) -- (\linewidth*0.6,5pt) node[midway,yshift=2em, scale=0.6]{region of control};
  \draw [decorate,decoration={brace, amplitude=5pt,raise=1.5ex}]
  (\linewidth*0.6,3pt) -- (\linewidth*0.75,3pt) node[midway,yshift=2em, scale=0.6, text width=\linewidth*0.2]{no control};
\end{tikzpicture}
\vspace{0.5cm}

b.2)

\begin{tikzpicture}
        \draw (\linewidth*0.1,0pt) -- (\linewidth*0.9,0pt);
        \draw (\linewidth*0.1,5pt) -- (\linewidth*0.1,- 5pt) node[anchor=north, scale=0.6] {$0$};
        \draw (\linewidth*0.5,5pt) -- (\linewidth*0.5,-5pt) node[anchor=north, scale=0.6] {$0.5$};
        \draw (\linewidth*0.75,5pt) -- (\linewidth*0.75,-5pt) node[anchor=north, scale=0.6] {$1-\theta$};
        \draw (\linewidth*0.9,5pt) -- (\linewidth*0.9,-5pt) node[anchor=north, scale=0.6] {$1$};
        \draw[opacity=0.3, line width=8pt, green] (\linewidth*0.5,0pt) -- (\linewidth*0.75,0pt);
        \draw [decorate,decoration={brace, mirror, amplitude=5pt,raise=1.5ex}]
  (\linewidth*0.5,-10pt) -- (\linewidth*0.75,-10pt) node[midway,yshift=-2em, scale=0.6]{support of random threshold};
        \draw[opacity=0.3, line width=5pt, red] (\linewidth*0.2,0pt) -- (\linewidth*0.4,0pt);
        \draw [decorate,decoration={brace, amplitude=5pt,raise=1.5ex}]
  (\linewidth*0.2,5pt) -- (\linewidth*0.4,5pt) node[midway,yshift=2em, scale=0.6]{region of control};
  \draw [decorate,decoration={brace, amplitude=5pt,raise=1.5ex}]
  (\linewidth*0.5,3pt) -- (\linewidth*0.75,3pt) node[midway,yshift=2em, scale=0.6, text width=\linewidth*0.2]{no control};
\end{tikzpicture}
\caption{Region of adversarial control.}
\label{fig:regionsOfControl}
\end{figure}

We summarize the above considerations in the following statement. 
\begin{lemma}
Assume that the honest nodes have the same perceptions on the honest AWs. Then, for all $i$, $1\le i \le N_h$, it holds that
\begin{equation}
    \AW_{{i,t}}(c) \in I_t(c).
\end{equation}
\end{lemma}
The above holds for every adversary strategy that satisfies Assumption~\ref{ass: Synchronised Random Reality Selection}.2. 
The idea is now to choose the support of the dRNG in such a way that independent of the honest AWs and the adversarial strategy all honest nodes will decide on the same reality with a positive probability. Every $\dFPCS$ time units we have therefore also a positive probability that all nodes decide on the same reality. It takes, hence, a geometrically distributed number of such intervals until all honest nodes agree on the same reality. 

\begin{definition}\label{def:asymptConsensusState}[Convergence to a consensus state]
We say that the protocol \emph{converges to a consensus state} if and only if there exist some reality $R$ and some (random) time $T$ such that
\begin{equation}
    \AW_{{i,t}} (R) > \theta, \forall i\in\{1,\ldots, N_h\}, \forall t>T.
\end{equation}
\end{definition}
\begin{remark}
Definition~\ref{def:asymptConsensusState} is similar to the definition of a consensus state; see Definition~\ref{def:consensusState}. While it describes the asymptotic behaviour of the protocol, it delivers not a practicable criterion for confirmation.\footnote{The ``probabilistic'' reason for this is that $T$ is not a stopping time.} A ``confirmation rule'', as in Definition~\ref{def:TTC}, however, is  always  susceptible to possible ``re-orgs''\footnote{A re-org is the procedure that a transaction that was confirmed is no longer in the preferred reality.} of the ledger state; see also Lemma~\ref{lem: negative result safety}. Quantifying the probabilities that such re-orgs happen depends on the precise communication and adversarial models and is out of this paper's scope. 
\end{remark}

This discussion can be turned into a formal protocol description written in Algorithm~\ref{alg: voting protocol} and we obtain  the following theorem.

\begin{algorithm}[t]
\caption{Voting protocol for a node $i$}
\label{alg: voting protocol}
$e\gets 0${\scriptsize{\Comment*[r]{epochs index}}}
$X_e\gets 0${\scriptsize{\Comment*[r]{common random variable}}}
\While{the node did not confirm a reality}
{
Obtain reality $R_i$ by Algorithm~\ref{alg: selection Conflict Graph} with $X_e$ \\
Before time $(e+1)\dFPCS$ issue new blocks with references selected by Algorithm~\ref{alg: URTS} for $R_i$\\
Before time $(e+1) \dFPCS$ receive new blocks\\
$e\gets e+1$\\
Wait time $\ddRNG$ to get common r.v. $X_e$\\
}
\end{algorithm}

\begin{theorem}[Liveness and safety - Synchronisation]\label{thm: SOTV}
Let 
$$
    q< \min\left\{1-\theta, \theta - \tfrac12\right\}
$$
be the weight of the adversary.
Then, under Assumption~\ref{ass: Synchronised Random Reality Selection}, the protocol (described by Algorithm~\ref{alg: voting protocol}) converges to a consensus state.
\end{theorem}

\begin{proof}
We start the protocol at time $t_0=0$ with a fixed set of conflicts $\conflictset$ of size $|\conflictset|$ and let the nodes exchange their votes until time $\dFPCS$. We let $\varepsilon>0$ be arbitrary but fixed and determine its value at the end of the proof. Every node waits until time $\dFPCS+\ddRNG$. If the node received the first random  number $X_1$ it will perform Algorithm~\ref{alg: selection Conflict Graph} with $X_1$ as a random number. If a node did not receive the random number on time it will use Algorithm~\ref{alg: selection Conflict Graph} with the threshold of $\theta$ (instead of random $X_1$). Between $\dFPCS+\ddRNG$ and $2 \dFPCS$ every honest node will not change its preferred reality. Let $A_1$ be the event that all honest nodes voted for their preferred reality and that these votes are seen by all other honest nodes. Let $B_1$ be the event that all honest nodes expressed their preferred reality on time, see Assumption~\ref{ass: Synchronised Random Reality Selection}.5, and $C_1$ that all these blocks arrived at every other honest node before time $2\dFPCS$. Since $A_1 = B_1 \cap C_1$ we have that 
\begin{align*}
     \P (A_1) &= \P(C_1 | B_1) \P(B_1)  \\
     &\geq   (1-\varepsilon)^{|\conflictset| N_h} (1-\varepsilon)^{N_h} \cr &=(1-\varepsilon)^{N_h(|\conflictset|+1)}.  
\end{align*}
At time $2\dFPCS+ \ddRNG$ with probability  of at least  $(1-\varepsilon)^{N_h}$ the new random number $X_2$ is received by all honest nodes. Hence, with probability 
\begin{equation*}\label{eq:peps}
    p(\varepsilon):= (1-\varepsilon)^{N_h(|\conflictset|+1)}
\end{equation*}
all honest nodes agree on the honest AWs, defined in Equation~\eqref{eq:honestAW} and the threshold $X_2$. We write $\AW^{(h)} (c) := \AW^{(h)}_{{2\dFPCS}}(c)$.
Let us note here that no honest node can perceive the honest AW. However, for the analysis, we assume a \emph{perfect view} or total information on the status of the system.

We start a recursive argument on the Conflict Graph by initialising $R=\emptyset$ and $U=\conflictset$. Define the conflict chosen by Algorithm~\ref{alg: selection Conflict Graph} inside the first while-loop at every iteration  $ c^* := \argmax \{\AW^{(h)}(c), \quad c\in \max_{\conflictset}(U)\}$.
We distinguish two cases.

\textbf{Case A:} $\AW^{(h)}(c^*) >0.5$. The support of the random threshold does lie above $0.5$; see also Figure~\ref{fig:regionsOfControl}. More, precisely, the probability $\xi_A$ that every node will include this conflict in its preferred reality (using Algorithm~\ref{alg: selection Conflict Graph}) satisfies $\xi_A > \AW^{(h)}(c^*) - 0.5>0$.  
All conflicts that conflict with $c^*$, i.e.  the neighbours in the Conflict Graph $N_{\conflictset}(c^*)$, are not preferred. 
Note here, that since every honest node might have a different perception of the actual AWs, it may run Algorithm~\ref{alg: selection Conflict Graph} in a different ``order''. However, as no two neighbours in the Conflict Graph can have more than $0.5$ of the honest AW,  the algorithm  treats all ``A cases`` before the following case.

\textbf{Case B:} $\AW^{(h)}(c^*) \leq 0.5$. In this case, all conflicts in $c^* \cup N_{\conflictset}(c^*)$ have an honest AW of less than $0.5.$ (This is because, in Algorithm~\ref{alg: selection Conflict Graph}, nodes treat conflicts in the order of ``decreasing AW''.) Since $q< \theta - 0.5$, with a positive probability $\xi_B$ none of these conflicts will have AWs above the threshold $X_2$ and none of them will be added to the preferred reality in the first while-loop of Algorithm~\ref{alg: selection Conflict Graph}.

We now remove the conflicts $c^* \cup N_{\conflictset}(c^*) $ from the set $U$ and continue this procedure until the set $U$ is the empty set. We set $\xi = \min\{\xi_A, \xi_B\}.$ Let $K$ be the size of the largest maximal independent set in the Conflict Graph.
Eventually, with a positive probability of at least $\xi^K$ the nodes agree on the preferred conflicts originating from case A. The nodes have to fill up the maximal branch with the second while-loop in Algorithm~\ref{alg: selection Conflict Graph}. Since they agree on the value of $X_2$ they also agree on the preferred reality.

Altogether, with a positive probability of at least $p(\varepsilon) \cdot  \xi^K$ all honest nodes vote for the same reality during the next epoch of length $\dFPCS$.
If this happens, an AW of more than~$\theta$ is obtained in the next epoch. Otherwise, we repeat this procedure until it is satisfied. The number of epochs necessary follows a geometric random variable.
\end{proof}

\begin{remark}
The above proof offers a possibility to estimate the ``consensus time'' $T$. In fact, its expectation  is bounded above by $\dFPCS \cdot ( 1 + (p(\varepsilon) \cdot \xi^K)^{-1}).$ This quantitative analysis is one main difference to Theorem~\ref{thm: security Random Blocks}, where no bounds on the ``consensus time'' are obtained. Another crucial difference is that Theorem~\ref{thm: SOTV} does not require assumptions on the randomness of the packages and issuance as in Assumption~\ref{ass: random Package Delay}.
\end{remark}

\begin{remark}
The assumption that the set of conflicts is fixed reduces to the assumption that the set of conflicts is bounded during the run-time of the protocol. The results, therefore, also apply to  sets of conflicts that may evolve over time. However, the quantitative bounds in the proof get worse for larger sets of conflicts. 
\end{remark}


\section{Performance studies}\label{sec: implementation}

\noindent We summarize some of the performance analysis obtained in  \cite{RobustnessTangle} via agent-based simulations to  validate the performance of the presented concepts.
The used simulator \cite{otv-simulator}
is written in Go and is open source. 
In this simulator, the necessary components of the consensus protocol are implemented, however, some of them are simplified. In the following we give a short description but refer to \cite{RobustnessTangle} for more details and further simulation results.

The simulated environment reflects a situation in which network participants are connected in a peer-to-peer network, where each node has  the same number of neighbors. Nodes can gossip, receive blocks, request for missing blocks, and state their opinions whenever conflicts occur. The  underlying network topology is modeled by a  Watts-Strogatz network. In order to mimic a real world behaviour the simulator allows to specify the network delay and packet loss for each node's connection.

Nodes are modeled as different independent agents that concurrently issue new blocks. This means that different nodes can have different perceptions of the Tangle and Approval Weights, at any given moment of time. The number of nodes does not change during the simulation period, and all the honest actors are actively participating in the consensus mechanism.  While the simulator allows to model different weight distributions, we focus here on the case of a Zipf distribution with $s=0$, i.e. every node has the same weight. 

Here, we focus on the robustness of the consensus protocol against the Bait-and-Switch attack, \ref{ex:BaS}, and illustrate the influence of the Synchronized Random Reality Selection (SRRS) introduced in Section \ref{sec: reality selection algorithm}. 

We present simulation studies with the following specific setup. We consider $N=100$ honest nodes with equal weight and one adversary node with weight $q$ (out of a total weight of 1). The block issuance time interval of nodes follows a Poisson distribution with issuance rates proportional to the nodes' weight.  The total  throughput is approximately constant at about 100 blocks per second. The  parents count (or number of references) is set to $k=8$.
The default confirmation threshold is set to $\theta=2/3$.
The peer-to-peer network is a realization of a Watts-Strogatz network with rewiring probability $1$ and 8 neighbors for each node.
The latency between two nodes in the peer-to-peer network is  set to be $0.1$ seconds and we assume the adversary to have no influence on the communication layer. 
The maximal simulation time is set to 60 seconds.

\begin{figure}
        \includegraphics[width=0.5\textwidth]{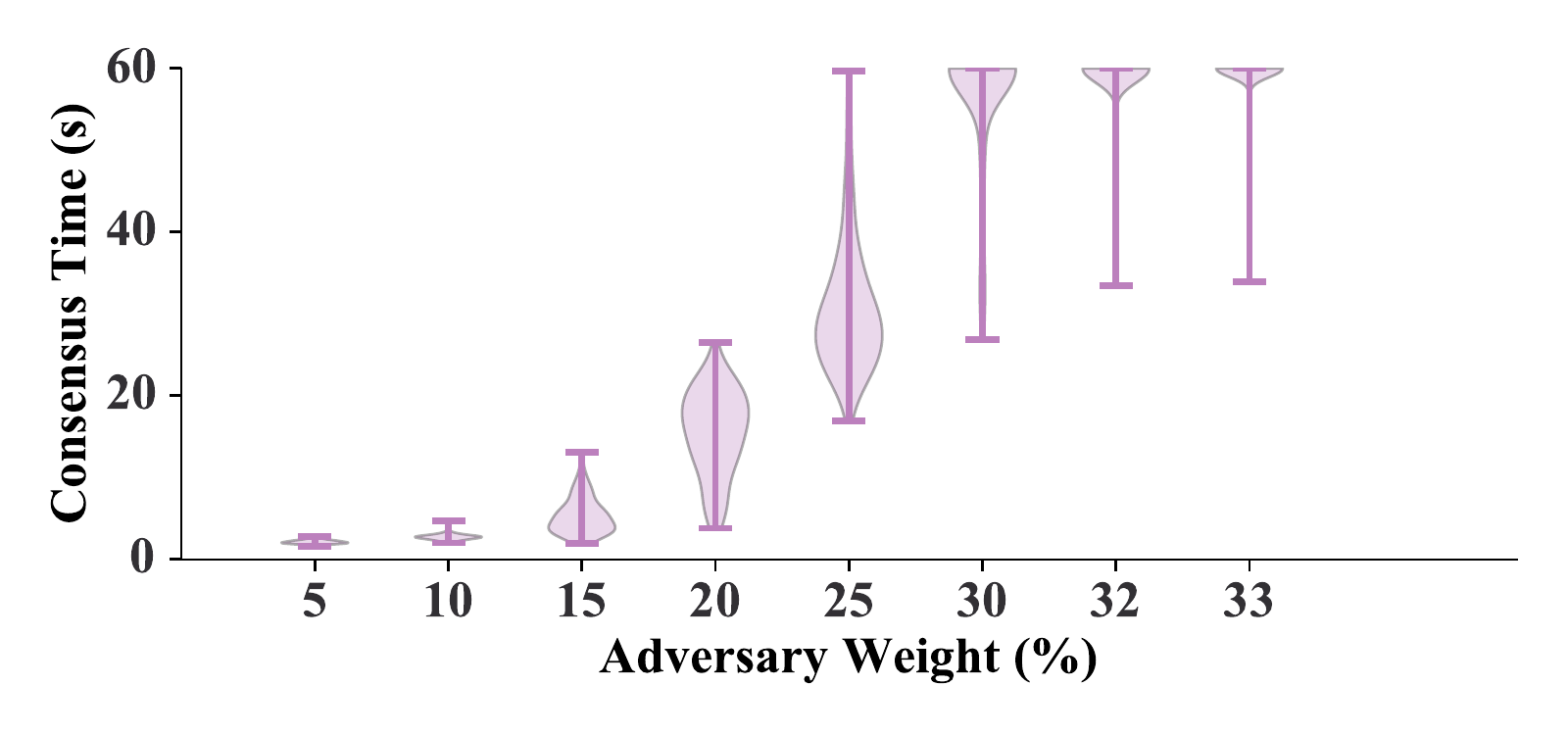}
        \caption{Consensus time distributions under Bait-and-Switch attack, without SRRS ($N=100$), taken from \cite{RobustnessTangle}.}
        \label{fig:bs_wo_srrs}
\end{figure}
\begin{figure}
        \includegraphics[width=0.5\textwidth]{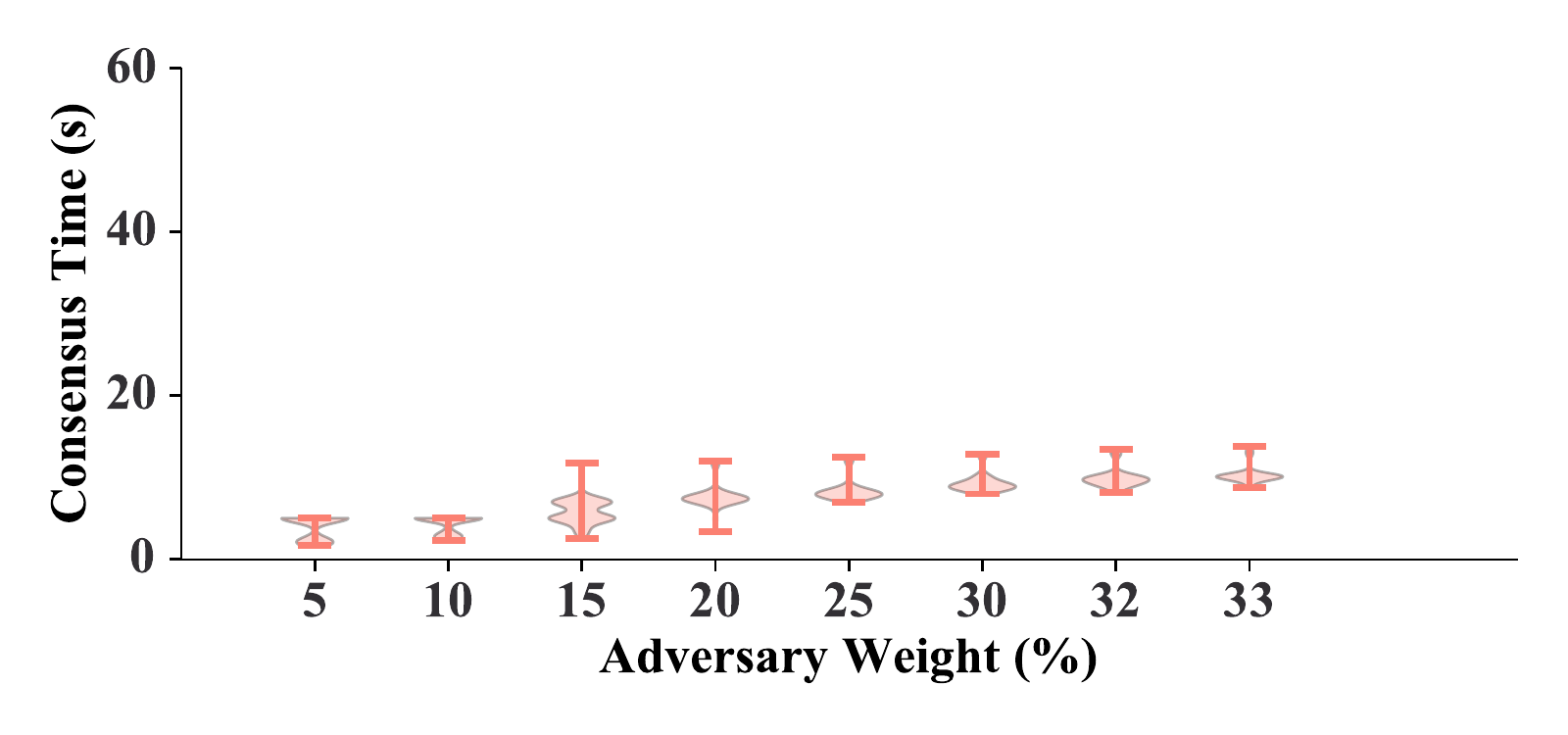} 
        \caption{Consensus time distributions under Bait-and-Switch attack, with SRRS ($N=100$), taken from \cite{RobustnessTangle}.}
        \label{fig:bs_w_srrs}
\end{figure}

The access to all Tangles of all nodes in the simulator allows to ``objectively'' measure the confirmation time as proposed in \cite{RobustnessTangle} for each node. These can be combined to extract the consensus time, which is defined as the time between the creation of a conflict and the time when all honest nodes confirm the same spending or branch. As such, for any given conflict, it is strictly larger than the confirmation time at any node. By measuring the consensus time, the safety and liveness of the protocol can be analyzed. 

Figure \ref{fig:bs_wo_srrs} shows the consensus time for the Bait-and-Switch strategy as a function of the adversarial weight if SRRS is disabled.
It is interesting to note that there this some ``inherent randomness'' in the protocol as blocks are issued randomly. This seems sufficient to guarantee the security against an attacker with at most $20\%$ of total weight. In Figure~\ref{fig:bs_w_srrs} we see the effectiveness of the SRRS, that makes the protocol robust against the Bait-and-Switch attack up to the theoretical limit of $q=1/3$.

We conclude this section with a brief analysis of the performance with the degree of decentralization and the size of the network. This also allows to support the values for the growth of the Witness Weight in Figure~\ref{fig: growth of the approval weight}. 
Figure \ref{fig:ct_s} shows the confirmation time distributions for several Zipf parameters $s$ with $N = 100$. The confirmation time increases with the ``decentralization'' of the network, as also discussed in Section~\ref{sec: Approval Weight}.  Nevertheless, Figure~\ref{fig: growth of the approval weight} shows, that in the extreme case where all nodes have equal weight, i.e. $s=0$, transaction are still confirmed within 2 seconds.
In  Figure~\ref{fig:ct_n} we show the dependence of the confirmation times with respect to the size of the network, for $s=0.9$. As described in  Section \ref{fig: growth of the approval weight}, the Witness Weight increases slower with a larger number of nodes. However, as Figure~\ref{fig:ct_s} shows the increase is sublinear, resulting in low confirmation times of $\sim$3 seconds, even for 1000 nodes.

\begin{figure}
        \includegraphics[width=0.5\textwidth]{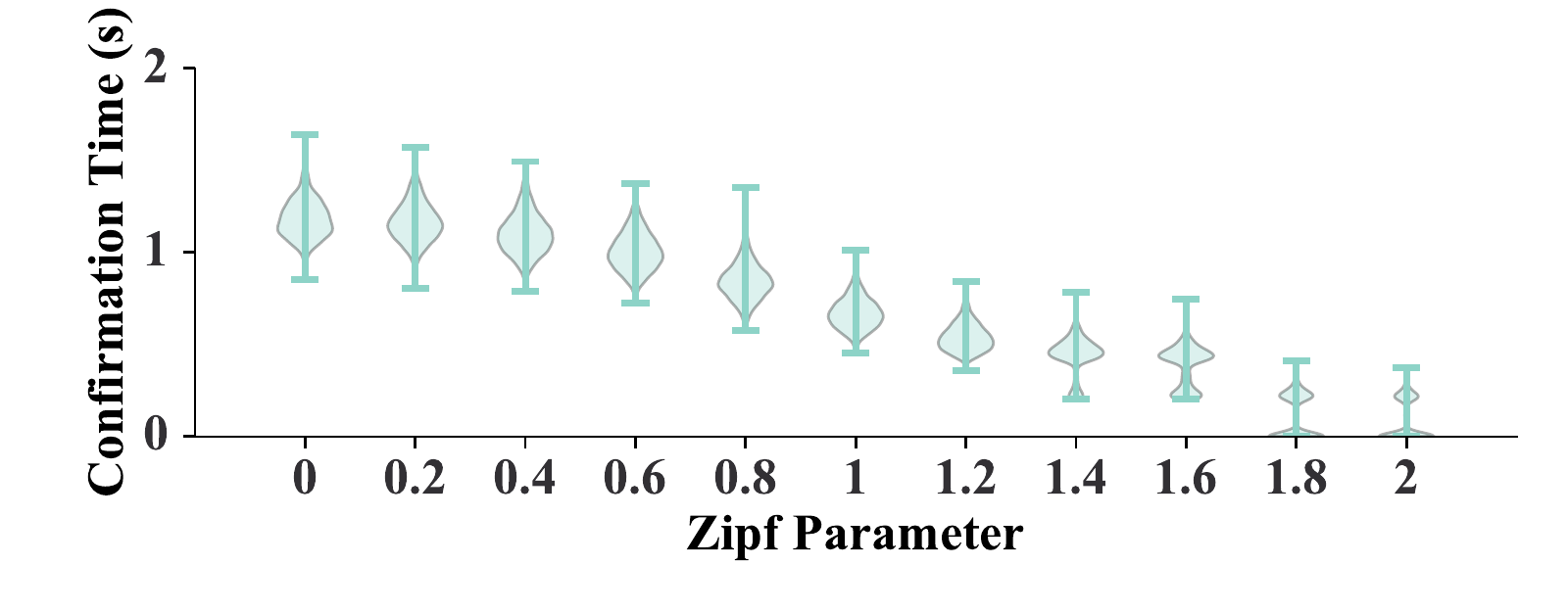} 
        \caption{Confirmation time distributions of blocks with the Zipf parameter $s$, taken from \cite{RobustnessTangle}.}
        \label{fig:ct_s}
        
\end{figure}
\begin{figure}
        \includegraphics[width=0.5\textwidth]{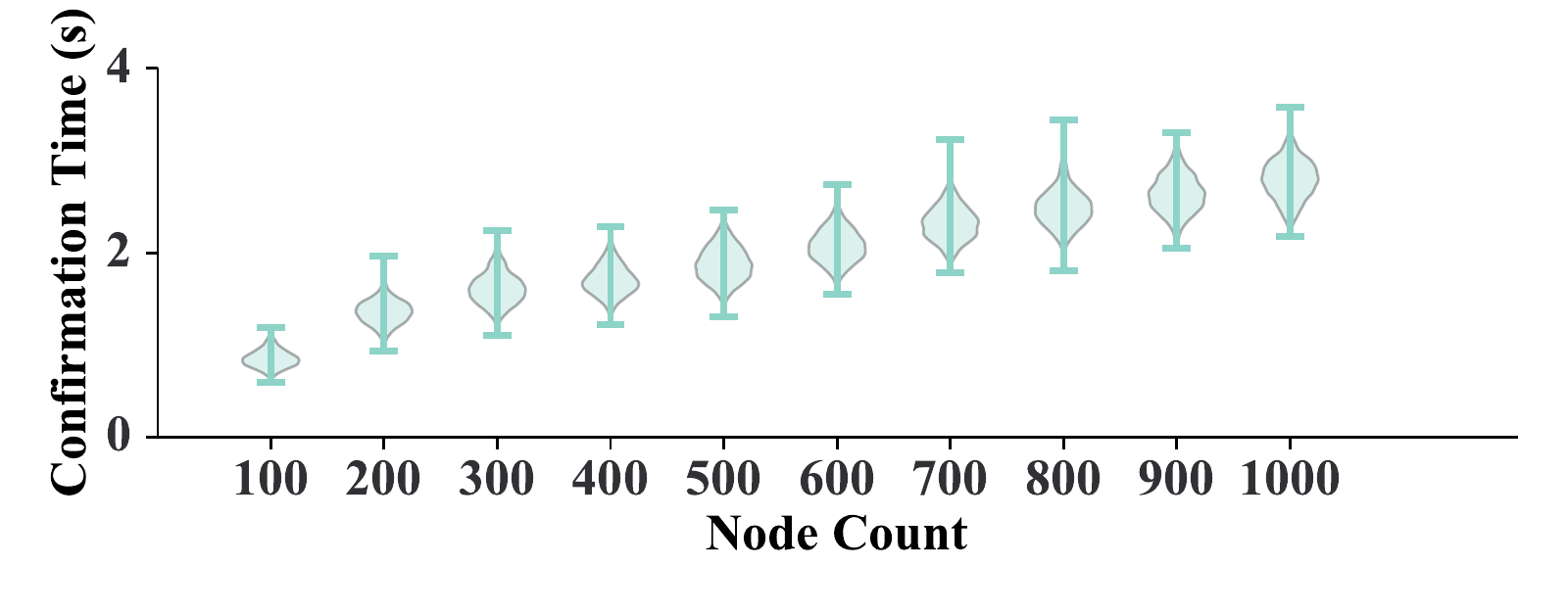} 
        \caption{Confirmation time distributions of blocks with the number of nodes, for $s=0.9$, taken from \cite{RobustnessTangle}.}
        \label{fig:ct_n}
\end{figure}

\section{Outlook - Future Research} \label{sec: outlook}

The proposed consensus mechanism in combination with the Reality-based Ledger supports the parallelisation of many processes, such as processing,  booking and voting. This can lead to a significant performance boost since it can enable multi-threaded concurrency. The potential for multi-threadedness of our solution, the capability to work in an asynchronous setting and  the leaderless approach can offer a highly performant consensus and ledger solution. Detailed and sound performance analysis will be necessary to validate theoretically predicted properties.

Since the ledger can be progressed without having global knowledge of new transaction additions to the ledger, it is possible that nodes can reach consensus with our mechanism  even without learning about all blocks. As a consequence, the approach may enable certain sharding solutions directly on the Tangle layer, in which nodes only observe  a  proportion of the total ledger. However, this approach may lower performance and potentially lower security and/or liveness. To address the viability of our solution for a sharded scenario key questions such as necessary assumptions and a full security analysis are vital.

The weight system from which the Approval Weight is derived can be constructed from multiple sources and in various settings. For example, the weight may be derived from the token value and the system can be operated permissioned or permissionless. A different approach is to obtain the weights through reputation systems, which has so far received little attention.

By introducing the transaction reference in addition to the block reference in Section \ref{sec: Voting}, the orphanage of transactions can be reduced through Algorithm~\ref{alg: URTS}.
However, it does not solve the problem entirely. For instance, an honest transaction can be referenced (directly) only by eventually rejected transactions and may never reach sufficient AW to be considered confirmed. This can be improved in several ways.
First, nodes may keep their ``own'' transactions as tips until they are confirmed.  This resembles an automated way of reattaching blocks. 
Second, nodes may also retain transactions that are in their preferred reality but for which they have not yet voted for in the tip pool. The transactions may then be supported via a transaction reference. 
Third, one could allow block and transaction references to be conflicting for a given block. The transaction can then be prioritised over block references in a transaction. This enables an efficient way to remove parts of branches from the referenced aggregated branch. Another possible solution for a more accurate voting is to introduce more reference types which would eventually allow nodes to remove more explicitly certain branches from the supported branches of referenced blocks.
The above examples demonstrate that solutions for the Tip Selection Algorithm can be found that mitigate or reduce orphanage, however, they require thorough analysis to cover edge cases.

\section{Conclusion}\label{sec: conclusion}

We have introduced a new leaderless consensus protocol that can be seen as a generalisation of the Nakamoto consensus. Our protocol is based on the Tangle, which not only forms a partially ordered  communication record between participants in a peer-to-peer network, but also serves as an efficient way to implicitly vote on the history of the underlying ledger.
These nodes are associated with reputation-based weights which are used to reach consensus on the acceptance of transactions to the ledger. 
The leaderless nature of the protocol allows asynchronous and concurrent writing access to the ledger. 
It also eliminates the need for shared ``memory pools'' for pending transactions and the special roles of miners or validators.

We provide formal definitions and proofs for the functionalities of the protocol, as well as pseudo-code for the various core algorithms. 
Furthermore, liveness and security of the protocol are analysed and several attack scenarios discussed in detail. We proved an impossibility result for safety in the asynchronous communication model. However, by introducing a synchronisation mechanism that  utilises a common random coin, we proved theoretical results on the safety of the protocol. Finally, we presented initial simulation studies that confirm the performance of the protocol with confirmation times in the order of second, and robustness up to a theoretical upper bound of the adversary weight of $1/3$.

\appendices 

\section{Estimates on Confluence Time}\label{sec:confluenceTime}

This section gives an upper bound on the confluence time $\tau_c$.

In the case where the network is in a low load regime, we can assume that the tip pool size is small. Then after several approvals, all new transactions will indirectly reference this transaction.
In the high load regime, the tip pool size $L_0 \gg k$ and the confluence time can be larger. Denote $K(t)$ the number of tips that approve the given transaction~$x$ at time~$t$. A new transaction at time~$t$ chooses~$k$ tips based on the state of the Tangle at time $t-h$. Hence, the probability of a new transaction approving at least one of the $K(t-h)$ tips that are approving $x$ is given by
\begin{equation}
    1 - \left(1 - \frac{K(t-h)}{L_0}\right)^k.
\end{equation}
As mentioned above, during a time interval $h$  we have that $\lambda h$ new tips arrive and $\lambda h $ tips are approved. Hence,  the probability that a transaction that was a tip at time $t-h$ is no longer a tip at time $t$ is 
\begin{equation}
    \frac{\lambda h}{L_0} = \frac{k-1}k.
\end{equation}
Therefore, at time $t$ we have that $(1)/k K(t-h)$ previous tips are still tips and $(k-1)/k K(t-h)$ have been referenced and are no longer tips. We denote by $A$ the set of the tips referencing $x$ that are still tips and by $B$ the tips referencing $x$ that got approved in $[t-h, h]$. We write
\begin{equation}
    p_A = \frac{K(t-h)}{k L_0} \mbox{, resp.  }p_B=\frac{ (k-1) K(t-h)}{k L_0}
\end{equation}
 for the probabilities to choose a given parent from the set $A$, resp. the set $B$.
Let $p_1$ be the probability to approve at least one transaction from $B$ but not from $A$ and let $p_2$ be the probability that at least two parents are chosen from the set $A$. Let $Y_A$ be the number of tips approved from set $A$. Then, note that in the first event, the number of tips that reference the given transaction increases by a factor $1$ and in the second event the number of tips decreases by a factor $Y_A-1$.

The probability of the first event can  be described by a binomial distribution. In fact, 
\begin{equation}
p_1 = \sum_{i=1}^k { k \choose i} p_B^i (1-p_A-p_B)^{k-i}.    
\end{equation}
Since $p_B$ is assumed to be small the two leading  terms are for $i\in \{1,2\}$ and we obtain
\begin{equation}
    p_1 \approx k p_B + \frac12 k(k-1) p_B^2.
\end{equation}

The random variable $Y_A$ follows a Binomial distribution $Bin(k, p_A)$, hence, 
\begin{equation}
\P [Y_1 \ge 2] = \sum_{i=2}^k    {k \choose i } p_A^i (1-p_A)^{k-i}.
\end{equation}
For $K(t-h)$ small, and, thus, $p_A$ small, the leading term in the above expression is for $i=2$. Hence, the second event happens  with probability approximately equal to
\begin{equation}
p_2= \frac12 k (k-1) p_A^2,
\end{equation}
and the tip pool size is reduced essentially by $1$.
Similarly to~\cite{popov2015} we can write now a differential equation for $K(t).$ We consider only the first order terms of $p_1$ and $p_2$ since we assume $K(t)$ to be small:
\begin{equation}
    \frac{d K(t)}{dt} = (p_1 -p_2) \lambda \approx \lambda \frac{ (k-1) K(t-h)}{ L_0} 
\end{equation}
Using Equation~\eqref{eq:L0} we can write
\begin{equation}
    \frac{d K(t)}{dt} \approx  \frac{ (k-1)^2 K(t-h)}{k h },
\end{equation}
with boundary condition $K(0)=1$. Following the lines of~\cite{popov2015} we obtain a solution of the form
\begin{equation}
    K(t) = \exp \left( W\left(\frac{(k-1)^2}{k} \right) \frac{t}{h}    \right),
\end{equation}
where $W(\cdot)$ is the so-called Lambert $W$-function. Taking the logarithm at both sides we find that the time when $K(t)$ reaches $\varepsilon L_0$ is roughly
\begin{equation}
    \tau_c \approx \frac{h}{W\left(\frac{(k-1)^2}{k} \right)} \left( \log L_0 + \log \varepsilon \right).
\end{equation}
For large $k$ we can approximate $W\left(\frac{(k-1)^2}{k} \right)\approx 2 \log (k-1) - \log k \approx \log k$ and obtain
\begin{equation}\label{eq:tauCklarge}
    \tau_c \approx \frac{h}{\log k} \log(L_0) \approx \frac{1}{\log k} h \log(\lambda h).
\end{equation}

\section{Illustrative example}\label{sec: toy example}
\begin{figure*}
    \centering
    \definecolor{megalightgray}{RGB}{244, 244, 244}
\definecolor{mygray}{RGB}{240, 240, 240}
\definecolor{myblue}{RGB}{102, 178, 255}
\definecolor{myred}{RGB}{255, 102, 102}
\tikzstyle{rounded_block}=[draw, rectangle, thick, minimum height=\heightBlock cm, minimum width = \widthBlock cm, text centered, rounded corners, draw=darkgray, font = \small]

\begin{tikzpicture}[use Hobby shortcut, scale = 0.95]
\def\xCoordinate{0.0}
\def\yCoordinate{0.0}
\def\yAdd{1.5}
\def\xAdd{0.75}
\def\innerSepar{-1.5}
\def\lineWidth{1.5}
\def\roundedCorners{1}
\def\opacityInternal{0.5}
\def\minHeight{20}
\def\fracyAdd{1/6}
\def\lineWidthBelow{3}
\def\minWidth{20}
\def\scaleFactorSupp{0.12}
\def\minWidthCM{\minWidth*0.0352778*3/4}
\tikzstyle{block}=[draw, rectangle, minimum height=\minHeight pt, minimum width = \minWidth pt, text centered, rounded corners=\roundedCorners pt, draw=darkgray, font=\large]
\tikzstyle{block_colored}=[draw, rectangle, minimum height= 2 pt, minimum width = \minWidth * 5 / 6 pt, rounded corners=\roundedCorners  pt, draw=darkgray]
\def\arrowStyle{-latex}

\node[font = \small] at (0, -0.5) {\textbf{Tangle}};
\node[block] (node_1_0) at (0, -1.5) {$\genesis$}; 
\node[block]  (node_2_0) at (-1, -3) { $x$}; 
\draw[line width = 3, line cap = round, myred](-1.25, -3.25)--(-0.75, -3.25);
\node[block] (node_2_1) at (1, -3) { $y$};
\draw[line width = 3, line cap = round, myblue](0.75, -3.25)--(1.25, -3.25);
\draw[\arrowStyle] (node_2_0) -- (node_1_0); 
\draw[\arrowStyle] (node_2_1) -- (node_1_0); 
\node[block] (node_3_0) at (-1, -4.5) {$u$}; 
\node[block] (node_3_1) at (1, -4.5) {$z$}; 
\draw[line width = 3, line cap = round, brown](0.75, -4.75)--(1.25, -4.75);
\draw[line width = 3, line cap = round, myred](-1.25, -4.75)--(-0.75, -4.75);
\draw[\arrowStyle] (node_3_1) -- (node_2_1);
\draw[\arrowStyle] (node_3_0) -- (node_2_0);
\node[block] (node_4_0) at (0.5, -6) {$w$}; 
\draw[line width = 3, line cap = round, teal](0.25, -6.25)--(0.75, -6.25);
\node[block] (node_4_1) at (2, -6) {$v$}; 
\draw[line width = 3, line cap = round, myblue](1.75, -6.25)--(2.25, -6.25);
\draw[\arrowStyle] (node_4_0) -- (node_2_0); 
\draw[\arrowStyle,dashed] (node_4_0) -- (node_3_1);
\draw[\arrowStyle] (node_4_1) -- (node_3_1);

\def\xAdd{4.5}
\node[font = \small] at (1+\xAdd, -0.5) {\textbf{Ledger DAG}};

\node[block] (ledger_1_0) at (1+\xAdd, -1.5) {$\tx{\genesis}$}; 
\node[block]  (ledger_2_0) at (\xAdd, -3) { $\tx{x}$}; 
\node[block]  (ledger_2_1) at (1+\xAdd, -3) { $\tx{y}$};
\node[block]  (ledger_2_2) at (2+\xAdd, -3) { $\tx{z}$}; 
\node[block] (ledger_3_0) at (-0.5+\xAdd, -4.5) { $\tx{u}$};
\node[block] (ledger_3_1) at (0.5+\xAdd, -4.5) { $\tx{w}$};
\node[block] (ledger_3_2) at (2+\xAdd, -4.5) { $\tx{v}$};
\draw[\arrowStyle] (ledger_2_0) -- (ledger_1_0); 
\draw[\arrowStyle] (ledger_2_1) -- (ledger_1_0); 
\draw[\arrowStyle] (ledger_2_2) -- (ledger_1_0);
\draw[\arrowStyle] (ledger_3_0) -- (ledger_2_0);
\draw[\arrowStyle] (ledger_3_1) -- (ledger_2_0); 
\draw[\arrowStyle] (ledger_3_2) -- (ledger_2_2); 


\def\xAdd{10}
\node[font = \small] at (\xAdd, -0.5) {\textbf{Conflict DAG}};

\node[block, fill = white] (conflict_1_0) at (\xAdd, -1.5) {$\tx{\genesis}$}; 
\node[block, fill = white]  (conflict_2_0) at (\xAdd-1, -3) { $\tx{x}$}; 
\node[block]  (conflict_2_1) at (1+\xAdd, -3) { $\tx{y}$};
\draw[\arrowStyle] (conflict_2_0) -- (conflict_1_0); 
\draw[\arrowStyle] (conflict_2_1) -- (conflict_1_0); 
\node[block] (conflict_3_0) at (-1.5+\xAdd, -4.5) { $\tx{u}$};
\node[block] (conflict_3_1) at (-0.5+\xAdd, -4.5) { $\tx{w}$};
\draw[\arrowStyle] (conflict_3_0) -- (conflict_2_0); 
+\xAdd 
\draw[\arrowStyle] (conflict_3_1) -- (conflict_2_0); 
+\xAdd 


\def\xAdd{14}
\path
  (15.2, -0.8) coordinate (z0)
  (14.4, -1.1) coordinate (z1)
  (14.4, -1.5) coordinate (z12)
  (14.5, -3) coordinate (z2)
  (14, -3.95) coordinate (z22)
  (13.5, -4.2) coordinate (z23)
 (14 , -5.1) coordinate (z3)
 (15.2,-3) coordinate (z4);
  \draw[closed, draw = black, dashed, fill = megalightgray] (z0) .. (z1) .. (z12) .. (z2) .. (z22) .. (z23) .. (z3) .. (z4);
\node at (\xAdd+.7, -3.5) {$R$};


\node[font = \small] at (\xAdd, -0.5) {\textbf{Conflict Graph}};

\node[block, fill = white] (conflict_1_0) at (\xAdd, -4.5) {$\tx{x}$}; 
\node[block, fill = white]  (conflict_1_1) at (\xAdd, -3) { $\tx{y}$}; 

\node[block, fill = white]  (conflict_2) at (\xAdd-1, -1.5) { $\tx{u}$}; 
\node[block, fill = white]  (conflict_3) at (\xAdd+1, -1.5) { $\tx{w}$}; 

\draw (conflict_1_0) -- (conflict_1_1);
\draw (conflict_2) -- (conflict_1_1);
\draw (conflict_3) -- (conflict_1_1);
\draw (conflict_2) -- (conflict_3);

\matrix[ampersand replacement=\&] (mtrx) at (0,-8) {
        \node[draw] (species1) {
            \begin{tabular}{c c}
            \multicolumn{2}{c}{\small{Node's weight}}\\
                \colorbox{myred}{0.3} & \colorbox{myblue}{0.1} \\ \colorbox{brown}{0.2} & \colorbox{teal}{0.4}
            \end{tabular}
        };
\\
} ;
\matrix[ampersand replacement=\&] (ww) at (5.3,-8) {
        \node (species1) [shape=rectangle,draw] {
            \begin{tabular}{ c c c c c c c}
                \multicolumn{7}{c}{\small{Block's Witness Weight}}
               \\
               $\genesis$ & $x$ & $y$ & $u$ & $z$ & $w$ & $v$ \\
                $1$ & $0.7$ & $0.7$ & $0.3$ & $0.7$ & $0.4$ & $0.1$
            \end{tabular}
        };
\\
} ;

\matrix[ampersand replacement=\&] (ww) at (12.5,-8) {
        \node (species1) [shape=rectangle,draw] {
            \begin{tabular}{ c c c c c c c}
                \multicolumn{7}{c}{\small{Transaction's Approval Weight}}
               \\
                 $\tx{\genesis}$ & $\tx{x}$ & $\tx{y}$ & $\tx{u}$ & $\tx{z}$ & $\tx{w}$ & $\tx{v}$ \\
                 $1$ & $0.7$ & $0.3$ & $0.3$ & $0.7$ & $0.4$ & $0.1$
            \end{tabular}
        };
\\
} ;

\end{tikzpicture}
    \caption{The Tangle, the Ledger DAG, the Conflict DAG and the Conflict Graph are shown. The Tangle starts with the genesis $\genesis$ and includes six other blocks $x,y,z,u,v,w$. Blocks $x$ and $y$ contain directly conflicting transactions $\tx{x}$ and $\tx{y}$. Similarly, blocks $u$ and $w$ contain directly conflicting transactions $\tx{u}$ and $\tx{w}$. Weights of four issuing nodes, which are identified with unique colors, are depicted. The WW of blocks and AW of transactions are computed.  In addition, the preferred reality $R$ is highlighted on the Conflict Graph. }
    \label{fig:example 1}
\end{figure*}
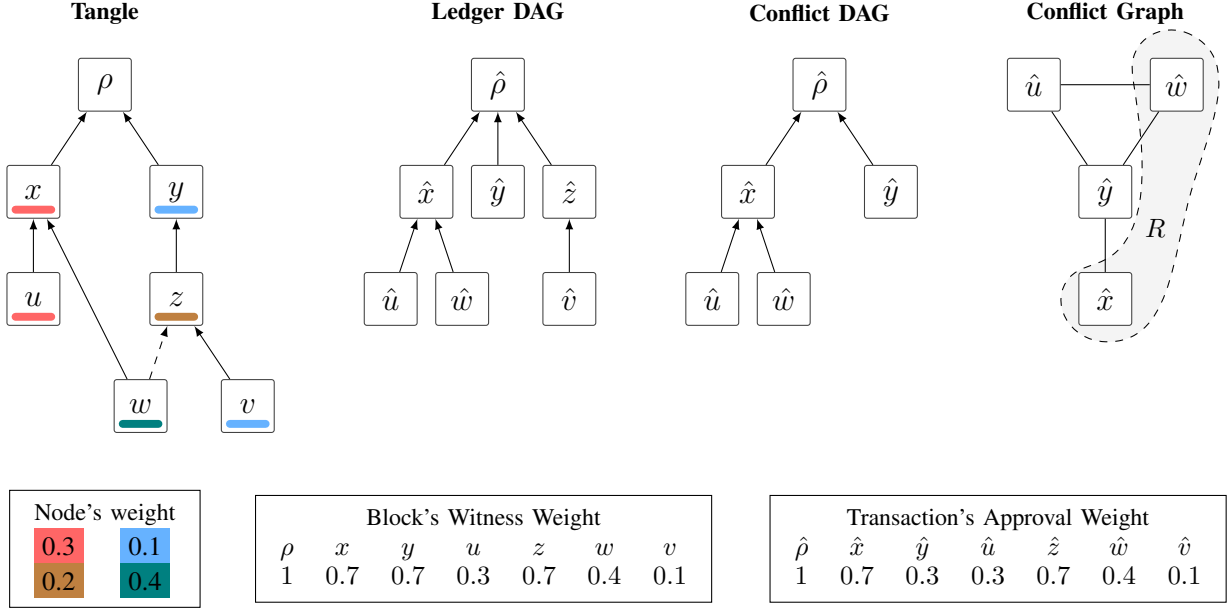
In this section, we demonstrate in Figure~\ref{fig:example 1} the most important concepts introduced in the paper using a toy example. In this example, blocks have two references which are identical in some cases.

The Tangle starts with the genesis $\genesis$ and six blocks are issued in the order $x,y,z,v,w,u$ by four distinct nodes which are identified with unique colors (red, blue, brown, green) and have weights $0.3,0.1,0.2,0.4$. In Figure~\ref{fig:example 1} we demonstrate the Tangle, the Ledger DAG, the Conflict DAG and the Conflict Graph. Transactions $\tx{x}$ and $\tx{y}$ consume the same output of $\tx{\genesis}$, thereby they are directly conflicting transactions. Similarly,  we say that $\tx{u}$ and $\tx{w}$ are directly conflicting as their input is the same output of $\tx{x}$. Thus, the Conflict DAG consists of the genesis $\tx{\genesis}$ and $\tx{x},\tx{y},\tx{u},\tx{w}$ and can be seen as the subDAG of the Ledger DAG induced by its vertices (see Section~\ref{sec:ledgerState}). The Conflict Graph shows the conflicting dependencies between $\tx{x},\tx{y},\tx{u},\tx{w}$, e.g. $\tx{y}$ is connected with $\tx{x}$ as they are directly conflicting and $\tx{y}$ is connected with all conflict-successors of $\tx{x}$, i.e. $\tx{u}$ and $\tx{w}$.

To demonstrate the steps of our protocol we discuss the actions from the point of view of the ``green'' node for issuing block $w$.
Before block $w$ was issued (i.e. at time when blocks $x,y,z,v$ were issued only), the preferred reality  (see Algorithm~\ref{alg:selectionConflictGraph}) for the node was $R=\{\tx{x}\}$ as $\AW(\tx{x})>\AW(\tx{y})$. Suppose that the  node decided to issue a block $w$ and selected the two tips $x$ and $z$ by Algorithm~\ref{alg: URTS}. Since the voting branch of $x$ is $\branch^{(p)}_{\votingset}(x)=\{\tx{x}\}\subseteq R$ and the voting branch of $z$ is $\branch^{(p)}_{\votingset}(z)=\{\tx{y}\}\not\subseteq R$ (see Definition~\ref{def: voting branch}), the node set a block reference from $w$ to $x$ only. After checking that the maximal contained branch  of transaction $\tx{z}$ is the main branch (or the empty set), the node put a transaction reference from $w$ to $z$ shown in Figure~\ref{fig:example 1} by the dashed arrow.

We observe that the Approval Weight of transactions is often equal to the Witness Weight of the corresponding blocks. However, this is not always the case. For instance, the Approval Weight of transaction $\tx{y}$  is the sum of weights of nodes supporting it. In this case, the ``brown'' and ``blue'' nodes are the supporters of $\tx{y}$, but not the ``green'' node because of the transaction reference from $w$ to $z$. Therefore, $\AW(\tx{y})=0.2+0.1=0.3$. On the other hand, $\WW(y) = 0.2+0.1+ 0.4 = 0.7$ since the ``green'' node witnesses the block $y$.

To find the preferred reality, a node must follow Algorithm~\ref{alg:selectionConflictGraph}. Specifically, the reality $R$ is constructed step-by-step by looking at the Conflict Graph (see Figure~\ref{fig:example 1}). At the first step, one includes $\tx{x}$ in $R$ as it attains the highest Approval Weight and it is the closest vertex to the genesis. Then we remove $\tx{x}$ and all conflicts which are conflicting with $\tx{x}$, i.e. $\tx{y}$ is removed. At the second step, we choose $\tx{w}$ as its Approval Weight is higher that the one of $\tx{u}$. After this step, we remove both $\tx{w}$ and $\tx{u}$. Since the empty set remains, we finish with constructing reality $R=\{\tx{x},\tx{w}\}$.

We also highlight that if at the next moment the ``brown'' node, which is supposed to be honest, decides to issue a new block and attach it to block $w$ (with a block
reference), then it would change its vote on conflicting transaction $\tx{y}$ (see Definition~\ref{def: change and current vote}). Specifically, the Approval Weight of $\tx{y}$ would be dropped by the weight of the ``brown'' node and become $0.1$. In contrast, the Approval Weight of $\tx{x}$ would gain and become $0.9$.

\section{Glossary}\label{sec: glossary}

\textbf{Approval Weight} A function that computes the ``relative'' part of the network that approves a given transaction 

\textbf{Conflict} A transaction that consumes the same output as a distinct transaction

\textbf{Conflicting transactions} Two transactions that contain two transactions in their past cones which consume the same output of some transaction

\textbf{Cone} A set of vertices in a DAG that are reachable from a given vertex by following the directions (past cone) and the opposite directions (future cone) of edges in the DAG. 

\textbf{Branch} A set of conflicts which does not contain conflicting transactions and is past-closed

\textbf{Branch DAG} A DAG that represents the relations between branches 

\textbf{Ledger DAG} A data structure that stores all transactions in the form of a DAG

\textbf{Tangle DAG} A data structure that stores all blocks in the form of a DAG

\textbf{Voting DAG} An augmented DAG that represents a combination of the Tangle DAG and the Ledger DAG and is used for determining voting cones

\textbf{Genesis} The transaction that is the ultimate predecessor of any transaction of the UTXO ledger.

\textbf{Block} An element of the Tangle DAG, constituted of identified data that refer to at least two blocks

\textbf{Node} A machine that is a part of the network. Its role is to issue new blocks and validate pre-existing ones

\textbf{Reality} A maximal branch

\textbf{Solidification} The process of retrieving missing blocks in the past cone of a given block which can be requested by a node

\textbf{Witness Weight} A function that computes the ``relative'' part of the network that approves a given block

\section*{Acknowledgment}

The authors would like to thank the developer team of the GoShimmer software, for supporting this study with the prototype implementation of the IOTA 2.0 protocol.
They also thank precious staff members of the IOTA Foundation and members of the IOTA community for their feedback and criticism. 
\newpage

\bibliographystyle{IEEEtran}
\bibliography{bibliography}


\end{document}